\documentclass[runningheads]{llncs}

\usepackage{todonotes}

\usepackage{enumerate}
\renewcommand{\subparagraph}[1]{\noindent{#1}}

\newcommand{\stface}{$\diamondsuit$-face\xspace}
\newcommand{\stfaces}{$\diamondsuit$-faces\xspace}
\usepackage{xspace}
\usepackage[ruled, vlined,boxed,linesnumberedhidden]{algorithm2e}
\SetKwRepeat{Struct}{record \{}{\}}%

\newcommand{\Int}{\KwSty{int}\xspace}
\newcommand{\Boolean}{\KwSty{bool}\xspace}
\newcommand{\Vertex}{\KwSty{Vertex}\xspace}
\newcommand{\Edge}{\KwSty{Edge}\xspace}
\newcommand{\multigraph}{graph\xspace}
\newcommand{\multigraphs}{multigraphs\xspace}

\usepackage{microtype}

\usepackage{array}
\usepackage{xcolor}
\usepackage{gensymb}
\usepackage{siunitx}
\graphicspath{{img/}}
\usepackage{amsfonts}
\usepackage{amssymb}
\usepackage[normalem]{ulem}
\newcommand{\NULL}{\textsc{NULL}\xspace}

\usepackage{comment}
\usepackage{wrapfig}

\newcommand{\remove}[1]{}
\newcommand{\bigoh}{\ensuremath{\mathcal{O}}}

\renewcommand{\int}{int}
\newtheorem{observation}{Observation}
\newtheorem{claimx}{Claim}
\usepackage{caption}
\usepackage{subcaption}

\renewenvironment{proof}
{{\bf Proof:}}{\hspace*{\fill}$\Box$\par\vspace{2mm}}

\usepackage{subcaption}
\usepackage{hyperref}
\usepackage[capitalise]{cleveref}

\Crefname{algorithm}{Listing}{Listings}
\Crefname{section}{Section}{Sections}
\Crefname{observation}{Observation}{Observations}
\Crefname{property}{Property}{Properties}
\Crefname{lemma}{Lemma}{Lemmata}
\Crefname{claim}{Claim}{Claims}
\Crefname{claimx}{Claim}{Claims}
\Crefname{figure}{Fig.}{Figs.}
\Crefname{figure}{Fig.}{Figs.}
\Crefname{enumi}{Condition}{Conditions}
\usepackage{color}
\definecolor{realblue}{rgb}{0,0,1}
\definecolor{blue}{rgb}{0.274,0.392,0.666}
\definecolor{darkerblue}{rgb}{0.094,0.455,0.804}
\definecolor{darkblue}{rgb}{0.063,0.306,0.545}
\definecolor{red}{rgb}{0.627,0.117,0.156}
\definecolor{green}{RGB}{51, 153, 102}
\definecolor{orange}{rgb}{0.903,0.739,0.382}
\definecolor{realred}{rgb}{1,0,0}

\newcommand{\codecomment}[1]{{{\textcolor{green}{#1}\xspace}}}

\renewcommand{\emph}[1]{\textcolor{realblue}{\bf #1}\xspace}

\hypersetup{colorlinks,linkcolor={realblue},citecolor={realblue},urlcolor={realblue}} 

\titlerunning{Efficient Enumeration of Drawings and Combinatorial Structures}

\author{Giordano {Da Lozzo} \and  Giuseppe {Di Battista} \and  Fabrizio Frati \and \\Fabrizio Grosso \and Maurizio Patrignani}
\authorrunning{G. {Da Lozzo}, G. {Di Battista}, F. Frati, F. Grosso,  and M. Patrignani} 
\institute{Roma Tre University, Rome, Italy\\  \email{\{giordano.dalozzo,giuseppe.dibattista,fabrizio.frati,\\fabrizio.grosso,maurizio.patrignani\}@uniroma3.it}}

\usepackage{fullpage}

\begin{document}

\title{Efficient Enumeration of Drawings and Combinatorial Structures for Maximal Planar Graphs
\thanks{Research partially supported by PRIN projects no.\ 2022ME9Z78 ``NextGRAAL: Next-generation algorithms for constrained GRAph visuALization'' and no.\  2022TS4Y3N ``EXPAND: scalable algorithms for EXPloratory Analyses of heterogeneous and dynamic Networked Data''.}}

\maketitle

\begin{abstract}
We propose efficient algorithms for enumerating the notorious combinatorial structures of maximal planar graphs, called canonical orderings and Schnyder woods, and the related classical graph drawings  by de Fraysseix, Pach, and Pollack [Combinatorica, 1990] and by Schnyder [SODA, 1990], called canonical drawings and Schnyder drawings, respectively.
To this aim (i) we devise an algorithm for enumerating 
special $e$-bipolar orientations of maximal planar graphs, called {\em canonical orientations}; (ii) we establish bijections between canonical orientations and canonical drawings, and between canonical orientations and Schnyder drawings; and (iii) we exploit the known correspondence between canonical orientations and canonical orderings, and the known bijection between canonical orientations and Schnyder woods. 
All our enumeration algorithms have $\bigoh(n)$ setup time, space usage, and delay between any two consecutively listed outputs, for an $n$-vertex maximal planar~graph. 
\end{abstract}

\section{Introduction}

In the late eighties, de Fraysseix, Pach, and Pollack~\cite{fpp-ss-88,dpp-hdpgg-90} and Schnyder~\cite{DBLP:conf/soda/Schnyder90} independently and almost simultaneously solved a question posed by Rosenstiehl and Tarjan~\cite{rt-rpl-86} by proving that every maximal planar graph, and consequently every planar graph, admits a planar straight-line drawing in a $\bigoh(n)\times \bigoh(n)$ grid. Since resolution and size are measures of primary importance for the readability of a graph representation~\cite{gdbook}, the result by de Fraysseix, Pach, and Pollack~\cite{fpp-ss-88,dpp-hdpgg-90} and by Schnyder~\cite{DBLP:conf/soda/Schnyder90} has a central place in the graph visualization literature. It also finds heterogeneous applications in other research areas, for example in knot theory~\cite{cks-mr-02,hl-nrm-01,hlp-cck-99} and computational complexity~\cite{bdy-dml-06,GanianHK0ORS16,ss-dsg-04}.  

The drawing algorithms presented by de Fraysseix, Pach, and Pollack and by Schnyder have become foundational for the graph drawing research area; see, e.g.,~\cite{nr-pgd-04,t-hgd-13}. The combinatorial structures conceived for these algorithms have been used to solve a plethora of problems in graph drawing~\cite{AlamBFKKU13,DBLP:conf/compgeom/AngeliniCCLR19,br-scd-06,DBLP:conf/cocoon/LozzoDF20,dor-tcg-94,d-gdt-10,DujmovicESW07,DBLP:journals/order/Felsner01,fz-sw-08,kant1996drawing,DBLP:journals/jocg/NollenburgPR16} and beyond~\cite{DBLP:journals/algorithmica/BarbayAHM12,DBLP:journals/gc/BonichonGHPS06,BoseDHLMW09,BoseGS05,cgh-ce-98,h-scirmt-06,HeKL99}. In a nutshell, de Fraysseix, Pach, and Pollack's algorithm works iteratively, as it draws the vertices of a maximal planar graph one by one, while maintaining some geometric invariants on the boundary of the current drawing. The order of insertion of the vertices ensures that each newly added vertex is in the outer face of the already drawn graph and that such a graph is biconnected; this order is called \emph{canonical ordering} (sometimes also \emph{shelling order}). Schnyder's algorithm is based on a partition of the internal edges of a maximal planar graph into three trees rooted at the outer vertices of the graph and satisfying certain combinatorial properties; these trees form a so-called \emph{Schnyder wood}. The three paths connecting each vertex to the roots of the trees define three regions of the plane, and the number of faces of the graph in such regions determines the vertex coordinates. 

At first sight, canonical orderings and Schnyder woods appear to be distant concepts. However, Schnyder~\cite{DBLP:conf/soda/Schnyder90} already observed that there is a simple algorithm to obtain a Schnyder wood of a maximal planar graph $G$ from a canonical ordering of $G$. The connection between the two combinatorial structures is deeper than this and it is best explained by the concept of \emph{canonical orientation}. Given a canonical ordering $\pi$ of $G$, the canonical orientation of $G$ with respect to $\pi$ is the directed graph obtained from $G$ by orienting each edge away from the vertex that comes first in $\pi$. de Fraysseix and Ossona de Mendez~\cite{DBLP:journals/dm/FraysseixM01} proved that there is a bijection between the canonical orientations and the Schnyder woods of $G$. 

In this paper, we consider the problem of enumerating the above combinatorial structures and the corresponding graph drawings. The ones we present are, to the best of our knowledge, the first enumeration algorithms for drawings of graphs. An \emph{enumeration algorithm} lists all the solutions of a problem, without duplicates, and then stops. Its efficiency is measured in terms of setup time, space usage, and maximum elapsed time (\emph{delay}) between the outputs of two consecutive solutions; see, e.g.,~\cite{DBLP:journals/dam/AvisF96,d-acp-11,r-cg-03,w-eea-16}. We envisage notable applications of graph drawing enumeration algorithms with polynomial delay: 
\begin{enumerate}[\bf (i)]
	\item The possibility of providing a user with several alternative drawings optimizing different aesthetic criteria, giving her the possibility of selecting the most suitable for her needs; enumerating techniques may become an important tool for graph drawing software.
	\item Machine-Learning-based graph drawing tools are eager of drawings of the same graph for their training; linear-time delay enumeration algorithms may provide a powerful fuel for such tools.
	\item Computer-aided systems for proving or disproving geometric and topological statements concerning graph drawings may benefit from enumeration algorithms for exploring the solution space of graph drawing problems.
\end{enumerate}


The enumeration of graph orientations has a rich literature. 
Consider an undirected graph $G$ with $n$ vertices and $m$ edges. In \cite{DBLP:journals/dam/ConteGMR18} algorithms are presented for generating the acyclic orientations of $G$ with $\bigoh(m)$ delay, the cyclic orientations with $\tilde\bigoh(m)$ delay, and the acyclic orientations with a prescribed single source with $\bigoh(n m)$ delay; see also earlier works on the same
 problem~\cite{DBLP:journals/ipl/BarbosaS99,DBLP:journals/jal/Squire98}. In \cite{DBLP:journals/algorithmica/BlindKV20}  the $k$-arc-connected orientations of $G$ are enumerated with $\bigoh(knm^2)$ delay. Of special interest for our paper is the enumeration of $e$-bipolar orientations of $G$. Let $e=(s,t)$ be an edge of $G$; an \emph{$e$-bipolar orientation} of $G$ (often called \emph{$st$-orientation}) is an acyclic orientation of $G$ such that $s$ and $t$ are the only source and the only sink of the orientation. de Fraysseix, Ossona de Mendez, and Rosenstiehl~\cite{DBLP:journals/dam/FraysseixMR95} provided an algorithm for enumerating the $e$-bipolar orientations of $G$ with polynomial delay. Setiawan and Nakano~\cite{setiawan2011listing} showed how suitable data structures and topological properties of planar graph drawings can be used in order to bound the delay of the algorithm by de Fraysseix {\em et al.} to $\bigoh(n)$, if $G$ is a biconnected planar graph. The link between these algorithm and our paper resides in another result by de Fraysseix and Ossona de Mendez~\cite{DBLP:journals/dm/FraysseixM01}: They proved that there exists a bijection between the canonical orientations and the bipolar orientations of $G$ {\em such that every internal vertex has at least two incoming edges}. Our enumeration algorithm for canonical orientations follows the strategy devised by de Fraysseix {\em et al.}~\cite{DBLP:journals/dam/FraysseixMR95} and enhanced by Setiawan and Nakano~\cite{setiawan2011listing} for enumerating bipolar orientations of biconnected planar graphs. However, the requirement that every internal vertex has at least two incoming edges dramatically increases the complexity of the problem and reveals new and, in our opinion, interesting topological properties of the desired orientations.

We present the following main results.
Let $G$ be an $n$-vertex maximal plane graph. 
\begin{itemize}
\item First, we show an algorithm that enumerates all the canonical orientations of $G$. The algorithm works recursively. Namely, it applies one or two operations (edge contraction and edge removal) to $G$. Each application of an operation results in a smaller graph, whose canonical orientations are enumerated recursively and then modified into canonical orientations of $G$ by orienting the contracted or removed edges. 

In order for the recursive algorithm to have small delay, we need to apply an edge contraction or removal only if the corresponding branch of computation is going to produce at least one canonical orientation of $G$. We thus identify necessary and sufficient conditions for a subgraph of $G$ to admit an orientation that can be extended to a canonical orientation of $G$. Further, we establish topological properties that determine whether applying an edge contraction or removal results in a graph satisfying the above conditions. Also, we design data structures that allow us to efficiently test for the satisfaction of these properties and to apply the corresponding operation if the test is successful. 
\item Second, as we show that canonical orderings are topological sortings of canonical orientations, our algorithm for enumerating canonical orientations allows us to obtain an algorithm that enumerates all canonical orderings of $G$. Furthermore, as canonical orientations are in bijection with Schnyder woods~\cite[Theorem 3.3]{DBLP:journals/dm/FraysseixM01}, our algorithm for enumerating canonical orientations allows us to obtain an algorithm that enumerates~all~Schnyder~woods~of~$G$. 
\item Third, we show that if we apply de Fraysseix, Pach, and Pollack's algorithm with two distinct canonical orderings corresponding to the same canonical orientation, the algorithm outputs the same planar straight-line drawing of $G$. This is the key fact that we use in order to establish a bijection between the canonical orientations of $G$ and the planar straight-line drawings of $G$ produced by de Fraysseix, Pach, and Pollack's algorithm. Together with our algorithm for the enumeration of canonical orientations, this allows us to enumerate such drawings.
%
\item Fourth, we prove that the planar straight-line drawings of $G$ obtained by Schnyder's algorithm are in bijection with the Schnyder woods. This, together with the bijection between canonical orientations and Schnyder woods and together with our algorithm for the enumeration of canonical orientations, allows us to enumerate the planar straight-line drawings of $G$ produced by Schnyder's algorithm.
\end{itemize}

All our enumeration algorithms have $\bigoh(n)$ setup time, $\bigoh(n)$ space usage, and $\bigoh(n)$ worst-case delay.

We remark that a different approach for the enumeration of canonical orientations might be based on the fact that the canonical orientations of a maximal plane graph form a distributive lattice~$\cal L$~\cite{DBLP:journals/combinatorics/Felsner04}. By the fundamental theorem of finite distributive lattices~\cite{Birkhoff-ros37}, there is a finite poset $P$ whose order ideals correspond to the elements of~$\cal L$ and it is known that $|P|$ is polynomial in~$n$~\cite[page 10]{DBLP:journals/combinatorics/Felsner04}. Enumerating the order ideals of $P$ is a studied problem. In~\cite{DBLP:journals/dam/HabibMNS01} an algorithm is presented that lists all order ideals of $P$ in $\bigoh(\Delta(P))$ delay, where $\Delta(P)$ is the maximum indegree of the covering graph of~$P$. However, the algorithm has three drawbacks that make it unsuitable for solving our problems. First, the guaranteed delay of the algorithm is amortized, and not worst-case. Second, the algorithm uses $\bigoh(w(P)\cdot |P|) = \bigoh(n^3)$ space, where $w(P)=\bigoh(n)$ is the width of $P$, and $\bigoh(|P|^2) = \bigoh(n^4)$ preprocessing time. Third and most importantly, each order ideal is produced twice by the algorithm, rather than just once as required by an enumeration algorithm.
Similarly, the algorithms in~\cite{pruesse1993gray,squire95,STEINER1986317} are affected by all or by part of the three drawbacks above.

\smallskip


The paper is organized as follows.
\cref{se:preliminaries} contains basic definitions and properties. The subsequent sections show how to enumerate:
canonical orientations (\cref{se:canonical-orientation});
canonical orderings and de Fraysseix, Pach, and Pollack drawings (\cref{se:canonical-drawings}); and
Schnyder woods and Schnyder drawings (\cref{se:schnyder}).
Conclusions and open problems are in \cref{se:conclusions}.



\section{Preliminaries} \label{se:preliminaries}

In the following, we provide basic definitions and concepts.

%

\subparagraph{\bf Graphs with multiple edges.} For technical reasons, we consider graphs and digraphs with multiple edges; edges with the same end-vertices are said to be \emph{parallel}. We only consider (di)graphs without \emph{self-loops}, i.e., edges with identical end-vertices. A graph without parallel edges is said to be \emph{simple}. For a graph $G$, we denote the degree of a vertex $v$ of $G$ by $\deg_G(v)$. 
Let $\cal D$ be a digraph. 
A \emph{source} (resp.\ \emph{sink}) of $\cal D$ is a vertex with no incoming (resp.\ no outgoing) edges. We say that $\cal D$ is \emph{acyclic} if it contains no directed cycle.
The \emph{underlying graph} of $\cal D$ is the undirected graph obtained from $\cal D$ by ignoring~the~edge~directions. 
An \emph{orientation} of an undirected graph $G$ is a digraph whose underlying graph is $G$. 

\subparagraph{\bf Planar graphs.} A \emph{drawing} $\Gamma$ of a graph maps each vertex to a point in the plane and each edge to a Jordan arc between its end-vertices. The drawing $\Gamma$ is \emph{planar} if no two edges cross, it is \emph{straight-line} if each edge is mapped to a straight-line segment, and it is a \emph{grid} drawing if all vertices have integer coordinates. Clearly, a graph might only admit a planar straight-line drawing if it is simple.

A planar drawing partitions the plane into connected regions, called \emph{faces}. The only unbounded face is the \emph{outer face}; the other (bounded) faces are \emph{internal}. Two planar drawings of the same connected planar graph are \emph{equivalent} if they determine the same circular order of the edges incident to each vertex. A \emph{planar embedding} is an equivalence class of planar drawings.
A \emph{plane graph} is a planar graph equipped with a planar embedding and a designated outer face. When talking about a subgraph $G'$ of a plane graph $G$, we always assume that $G'$ inherits a plane embedding from $G$; sometimes we write \emph{plane subgraph} to stress the fact that the subgraph has an associated plane embedding. A \emph{maximal planar graph} is a planar graph without parallel edges to which no edge can be added without losing planarity or simplicity. A \emph{maximal plane graph} is a maximal planar graph with a prescribed embedding. 
A vertex or edge of a plane \multigraph is \emph{internal} if it is not incident to the outer face, and it is \emph{outer} otherwise. An internal edge is a \emph{chord} if both its end-vertices are outer.

Let $G$ be a plane graph and let $e$ be an edge of $G$ with end-vertices $u$ and $v$. The \emph{contraction of} $e$ \emph{in} $G$ is an operation that removes $e$ from $G$ and that ``merges'' $u$ and $v$ into a new vertex $w$. Suppose that the clockwise cyclic order of the edges incident to $u$ is $e, e^u_1,\dots,e^u_h$ and that the clockwise cyclic order of the edges incident to $v$ is $e,e^v_1,\dots,e^v_k$. Then, the clockwise cyclic order of the edges incident to $w$ is set to $e^u_1,\dots,e^u_h, e^v_1,\dots,e^v_k$. 
Suppose that, in $G$, there exist a vertex $x$, an edge $e' = (u, x)$, and an edge $e'' = (v, x)$. Then, the contraction of $e$ in $G$ turns $e'$ and $e''$ into a pair of parallel edges incident to $w$ and $x$. Also, suppose that there exists an edge $e'$ that is parallel to $e$ in $G$. Then, the contraction of $e$ in $G$ turns $e'$ into a self-loop incident to $w$.


\subparagraph{\bf Connectivity.} 
A graph is \emph{connected} if it contains a path between any two vertices.
A \emph{cut-vertex} (resp. \emph{separation pair}) in a graph is a vertex (resp. a pair of vertices) whose removal disconnects the graph.
A graph is \emph{biconnected} (\emph{triconnected}) if it has no cut-vertex (resp.\ no separation pair).
Note that a single edge or a set of parallel edges forms a biconnected graph.
A \emph{split pair} of $G$ is either a pair of adjacent vertices or a separation pair. 
The \emph{components}~of~$G$ \emph{separated by} a split pair $\{u,v\}$ are defined as follows.
If $e$ is an edge of $G$ with end-vertices $u$ and $v$, then it is a component of $G$ separated by $\{u,v\}$; note that each parallel edge between $u$ and $v$ determines a distinct component of $G$ separated~by~$\{u,v\}$.
Also, let $G_1,\dots,G_k$ be the connected components of $G \setminus \{u,v\}$. The subgraphs of $G$ induced by $V(G_i) \cup \{u,v\}$, minus all parallel edges between $u$ and $v$, are components of $G$ separated by $\{u,v\}$, for $i=1,\dots,k$.
We will exploit the following. 

\begin{property} \label{obs:split-biconnected}
    Let $H$ be a biconnected plane graph and let $C$ be a cycle of $H$. Then the plane graph $H_C$ consisting of the vertices and of the edges of $H$ that lie in the interior or on the boundary of $C$ is biconnected.
\end{property}
\begin{proof}
In a biconnected plane graph, each face is bounded by a cycle~\cite{w-nspg-32}. Note that all the internal faces of $H_C$ are also internal faces of $H$. Thus, since $H$ is biconnected, these faces are bounded by cycles. Moreover, the outer face of $H_C$ is bounded by the cycle $C$, by construction. By \cite[Theorem~10.7]{bondy2011graph}, if the boundary of each face of a plane graph is a cycle, then the graph is biconnected. Therefore, $H_C$ is~biconnected.\end{proof}




\subparagraph{\bf Planar st-graphs.} 
Given two vertices $s$ and $t$ of an undirected graph $G$, an orientation of $G$ is an \emph{$st$-orientation} if (i) it is acyclic and (ii) $s$ and $t$ are its unique source and unique sink, respectively.  A digraph is a \emph{planar $st$-graph} if and only if it is an $st$-orientation and it admits a planar embedding $\mathcal E$ with $s$ and $t$ on the outer face; such a digraph  together with $\mathcal E$ is a \emph{plane $st$-graph}. A face $f$ is a \emph{\stface} 
if 
%
its boundary consists of two directed paths 
from a common source $s_f$ to a common sink $t_f$. 
The following observations are well known.

\begin{observation}
[\protect{\cite[Lemma~1]{DBLP:journals/algorithmica/LozzoBFPR20}}]
\label{obs:st-face}
    A plane digraph with $s$ and $t$ on the outer face is a plane $st$-graph if and only if each of its faces is a \stface.
\end{observation}

\begin{observation}[\protect{\cite[Lemma~4.2]{gdbook}}]\label{obs:bimodality}
Let $G$ be a plane $st$-graph. Then, all the incoming (resp.\ all the outgoing) edges incident to any vertex $v$ of $G$ appear consecutively around $v$.
\end{observation}

Let $\cal D$ be a plane $st$-graph and let $e=(u,v)$ be an edge of $\cal D$. 
The \emph{left face} (resp. \emph{right face}) of $e$ is the face to the left (resp. right) of $e$ while moving from $u$ to~$v$. 
The \emph{left path} $p_l$ (resp.\ the \emph{right path} $p_r$) of $\cal D$ consists of the edges of $\cal D$ whose left face (resp. whose right face) is the outer face of $\cal D$.     

\medskip
In the following, let $G$ be a maximal plane graph and let $(u,v,z)$ be the cycle delimiting its outer face, where $u$, $v$, and $z$ appear in this counter-clockwise order along~the~cycle. 

\begin{figure}
    \centering
    \begin{subfigure}{0.24\textwidth}
    \centering
    \includegraphics[page=1, width=\textwidth]{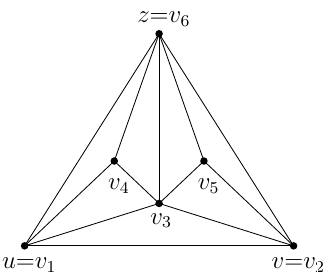}
    \caption{\label{fig:small-example-canonical-ordering-1}}
    \end{subfigure}
    \hfil
    \begin{subfigure}{0.24\textwidth}
    \centering
    \includegraphics[page=2, width=\textwidth]{img/simple-example.pdf}
    \caption{\label{fig:small-example-canonical-ordering-2}}
    \end{subfigure}
    \hfil
    \begin{subfigure}{0.24\textwidth}
    \centering
    \includegraphics[page=3, width=\textwidth]{img/simple-example.pdf}
    \caption{\label{fig:small-example-canonical-orientation}}
    \end{subfigure}
    \hfil
    \begin{subfigure}{0.24\textwidth}
    \centering
    \includegraphics[page=4, width=\textwidth]{img/simple-example.pdf}
    \caption{\label{fig:small-example-Schnyder}}
    \end{subfigure}
    \caption{(a), (b) The two canonical orderings with first vertex $u$ of a maximal plane graph $G$. (c) The unique canonical orientation with first vertex $u$ of $G$. (d) The unique Schnyder~wood~of~$G$.}
    \label{fig:small-example}
\end{figure}

\subparagraph{\bf Canonical orderings.} A \emph{canonical ordering of $G$ with first vertex $u$} is a labeling of the vertices $v_1=u, v_2=v, v_3,  \dots, v_{n-1}, v_n=z$ meeting the following requirements for every $k=3,\dots,n-1$; see \cref{fig:small-example-canonical-ordering-1,fig:small-example-canonical-ordering-2} and refer to \cite{dpp-hdpgg-90}. 

\begin{enumerate}[(CO-1)]
\item The plane subgraph $G_{k} \subseteq G$ induced by $v_1,v_2,\dots,v_k$ is $2$-connected; let $C_k$ be the cycle bounding its outer face;
\item $v_{k+1}$ is in the outer face of $G_{k}$, and its neighbors in $G_{k}$ form an (at
least $2$-element) subinterval of the path $C_{k}-(u,v)$. 
\end{enumerate}


A \emph{canonical ordering} of $G$ is a canonical ordering of $G$ with first vertex $x$, where $x$ is a vertex in $\{u,v,z\}$. Finally, if $G'$ is a maximal planar graph, a \emph{canonical ordering} of $G'$ is a canonical ordering of a maximal plane graph isomorphic to $G'$.

\begin{property} \label{pr:one-in-the-future}
Let $\pi=(v_1=u, v_2=v, v_3,  \dots, v_n=z)$ be a canonical ordering of a maximal plane graph $G$. For $i=3,\dots,n-1$, each vertex $v_i$ has at least two neighbors $v_j$ with $j<i$ and one neighbor $v_j$ with $j>i$. 
\end{property}

\begin{proof}
Recall that, for any $i=3,\dots,n$, we denote by $G_{i}$ the plane subgraph of $G$ induced by $v_1,v_2,\dots,v_{i}$. For $i=4,\dots,n-1$, the fact that $v_i$ has at least two neighbors $v_j$ with $j<i$ directly follows by condition (CO-2) of $\pi$. Furthermore, $v_3$ is adjacent to $v_1$ and $v_2$, since $G_3$ is biconnected, by condition (CO-1) of $\pi$. Suppose, for a contradiction, that, for some $i\in \{3,\dots,n-1\}$, the vertex $v_i$ has no neighbor $v_j$ with $j>i$. By condition (CO-1) of $\pi$, we have that $G_{i}$ is $2$-connected.  Also, we have that $v_i$ is in the outer face of $G_{i-1}$; this comes from condition (CO-2) of $\pi$ if $i>3$, and from the fact that $G_{i-1}$ is an edge if $i=3$. Hence, $v_i$ is incident to the outer face of $G_i$. Let $k\geq i$ be the largest index such that $v_i$ is incident to the outer face of $G_{k}$. Let $C_k$ be the cycle delimiting the outer face of $G_{k}$. Then $v_{k+1}$ is adjacent to a vertex that comes before $v_i$ and to a vertex that comes after $v_i$ in the path $C_k-(u,v)$, as otherwise $v_i$ would also be incident to the outer face of $G_{k+1}$, which would contradict the maximality of $k$. However, the fact that $v_{k+1}$ is not adjacent to $v_i$ implies that the neighbors of $v_{k+1}$ in $G_k$ do not form an interval of $C_k-(u,v)$, a contradiction to condition (CO-2) of $\pi$. 
\end{proof}

\subparagraph{\bf Canonical orientations.} Let $\pi=(v_1,\dots,v_n)$ be a canonical ordering of $G$ with first vertex~$u$. Orient every edge $(v_i,v_j)$ of $G$ from $v_i$ to $v_j$ if and only if $i<j$. The resulting orientation is the \emph{canonical orientation of $G$ with respect to $\pi$}. We say that an orientation $\cal D$ of $G$ is a \emph{canonical orientation with first vertex $u$} if there exists a canonical ordering $\pi$ of $G$ with first vertex $u$ such that $\cal D$ is the canonical orientation of $G$ with respect to $\pi$; see~\cref{fig:small-example-canonical-orientation}.  A \emph{canonical orientation} of $G$ is a canonical orientation with first vertex $x$, where $x$ is a vertex in $\{u,v,z\}$. Finally, if $G'$ is a maximal planar graph, a \emph{canonical orientation} of $G'$ is a canonical orientation of a maximal plane graph isomorphic~to~$G'$.

\subparagraph{\bf Schnyder woods.} A \emph{Schnyder wood} $(\mathcal T_1,\mathcal T_2, \mathcal T_3)$ of $G$ is an assignment of directions and of the colors $1$, $2$ and~$3$ to the internal edges of $G$ such that the following two properties hold; see~\cref{fig:small-example-Schnyder} and refer to~\cite{DBLP:conf/soda/Schnyder90}. Let $i-1 = 3$, if $i=1$, and let $i+1=1$, if $i=3$.

\begin{enumerate}[(S-1)]
	\item For $i=1,2,3$, each internal vertex $x$ has one outgoing edge $e_i$ of color $i$. The outgoing edges $e_1$, $e_2$, and $e_3$ appear in this counter-clockwise order at $x$. Further, for $i=1,2,3$, all the incoming edges at $x$ of color $i$ appear in the clockwise sector between the edges $e_{i+1}$ and $e_{i-1}$. 
	\item At the outer vertices $u$, $v$, and $z$, all the internal edges are incoming and of color $1$, $2$, and $3$, respectively. 
\end{enumerate}
	\begin{center}
		\begin{tabular}{ccc}
			\centering
			\includegraphics[page=1]{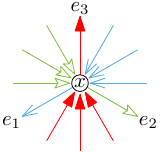}
			&~~~~~~~~~~~~~~~~&
			\includegraphics[page=2]{schnyder-conditions}
			\\
			{\bf (S-1)} 			
			&~~~~~~~~~~~~~~~~& 
			{\bf (S-2)} 	
		\end{tabular}
	\end{center}

\noindent
Finally, if $G'$ is a maximal planar graph, a \emph{Schnyder wood}  of $G'$ is a Schnyder wood of a maximal plane graph isomorphic to $G'$.


\section{Canonical Orientations} \label{se:canonical-orientation}

In~\cite[Lemma 3.6, Lemma 3.7, Theorem 3.3]{DBLP:journals/dm/FraysseixM01}, de Fraysseix and Ossona De Mendez proved the following characterization, for which we provide here an alternative proof.  

\begin{theorem}[\cite{DBLP:journals/dm/FraysseixM01}]\label{th:canonical-orientation-characterization}
Let $G$ be a maximal plane graph and let $(u,v,z)$ be the cycle delimiting its outer face, where $u$, $v$, and $z$ appear in this counter-clockwise order along the cycle. An orientation $\cal D$ of $G$ is a canonical orientation with first vertex $u$ if and only if $\cal D$ is a $uz$-orientation in which every internal vertex has at least two incoming edges.
\end{theorem}

\begin{proof}
($\Longrightarrow$) Consider any canonical orientation $\cal D$ with first vertex $u$. We prove that $\cal D$ is a $uz$-orientation as in the statement. Let $\pi=(v_1=u,v_2=v,v_3,\dots, v_{n-1},v_n=z)$ be any canonical ordering of $G$ such that $\cal D$ is the canonical orientation of $G$ with respect to $\pi$. By the construction of $\cal D$ from $\pi$, an edge $(v_i,v_j)$ is directed from $v_i$ to $v_j$ if and only if $i<j$. This implies that $\cal D$ is an acyclic orientation, that $u$ is a source in $\cal D$, and that $z$ is a sink in $\cal D$.
Furthermore, $v_2=v$ has one incoming edge in $\cal D$, namely $(v_1,v_2)$, and at least one outgoing edge in $\cal D$, namely $(v_2,v_3)$. Finally, for $i=3,4,\dots,n-1$, by \cref{pr:one-in-the-future}, we have that $v_i$ has at least two incoming edges and at least one outgoing edge in $\mathcal D$. 




\begin{figure}[t]
\begin{subfigure}{.32\textwidth}
	\includegraphics[page=1,width=\textwidth]{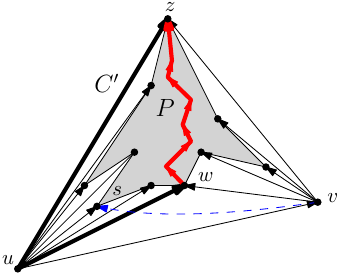}
	\caption{}
	\label{fig:v3-is-third}
\end{subfigure}
\begin{subfigure}{.32\textwidth}
\hfil	\includegraphics[page=2,width=\textwidth]{img/canonical-orderings-canonical-orientations.pdf}
    \caption{}
    \label{fig:vKplus1-internal}
\end{subfigure}
\hfil
\begin{subfigure}{.32\textwidth}
    \includegraphics[page=9,width=\textwidth]{img/canonical-orderings-canonical-orientations.pdf}
 \caption{}
	\label{fig:large-face}
\end{subfigure}
\caption{Illustrations for the proof of \cref{th:canonical-orientation-characterization}. 
(a) Illustration for the proof that $w=v_3$; the directed path $P$ from $w$ to $z$ is (thick) red, the edge $(v,s)$ is (thin dashed) blue, the edges of the cycle $C'$ are thick.
(b) Illustration for the proof that $v_{k+1}$ lies in the outer face of $G_k$; the directed path $P$ from  $v_{k+1}$ to $z$ is (thick) red, the edges of the cycle $C_k$ are thick. (c) Illustration for the proof that the neighbors of $v_{k+1}$ form a subinterval of the path $C_{k}-(u,v)$; the face incident to both $w_r$ and $v_{k+1}$ having length greater than three~is~shaded~yellow.}
\label{fig:canonical-orientation-characterization}
\end{figure}

($\Longleftarrow$) Let $\cal D$ be a $uz$-orientation in which every internal vertex has at least two incoming edges.
We prove that $\cal D$ is a canonical orientation of $G$ with first vertex $u$. Consider any topological sorting $\pi=(v_1,\dots,v_n)$ of $\cal D$. We show that $\pi$ is a canonical ordering of $G$ with first vertex $u$. 
Clearly, we have $v_1=u$ and $v_n=z$, as $u$ is the only source of $\cal D$ and $z$ is the only sink of $\cal D$. 
Furthermore, we have $v_2=v$, as every vertex different from $u$ and $v$ has at least two incoming edges, and hence at least two vertices come before it in any topological sorting of $\cal D$. Let $w$ be the vertex that is incident to an internal face of $G$ together with $u$ and $v$. 
We prove that every vertex of $\cal D$ different from $u$ and $v$ is a successor of $w$, hence $v_3=w$. 
 Suppose, for a contradiction, that there exists a vertex $s \neq w$ that is a source of the plane digraph $\cal D'$ obtained from $\cal D$ by removing $u$ and $v$; refer to \cref{fig:v3-is-third}. Consider any directed path $P$ from $w$ to $z$ in $\cal D'$ and note that $s$ does not belong to $P$. Since $s$ has at least two incoming edges in $\cal D$, edges from $u$ and $v$ to $s$ exist. Furthermore, in $\cal D$, the vertex $s$ lies either in the interior of the cycle $C'$ bounded by the edge $(u,w)$, by $P$, and by the edge $(u,z)$, or in the interior of the cycle $C''$ bounded by the edge $(v,w)$, by $P$, and by the edge $(v,z)$. In the former case,  the edge $(v,s)$ crosses $C'$, while in the latter case, the edge $(u,s)$ crosses $C''$. In both cases, we get a contradiction to the planarity of $\cal D$. This contradiction proves that $v_3=w$.


It remains to prove that, for $k=3,\dots,n-1$, the ordering $\pi$ satisfies conditions~(CO-1) and~(CO-2) of a canonical ordering. In order to prove condition~(CO-1), we proceed by induction on $k$. In the base case, $k=3$; that $G_3$ is biconnected comes from the fact that it coincides with the cycle $(u,v,w)$. Suppose now that $G_k$ is biconnected, for some $k\in \{3,\dots,n-1\}$. Since $v_{k+1}$ has at least two incoming edges, by assumption, and since the end-vertices of such edges different from $v_{k+1}$ belong to $G_{k}$, since $\pi$ is a topological sorting of $\mathcal D$, it follows that $G_{k+1}$ is biconnected, which concludes the proof of condition~(CO-1). In order to prove condition~(CO-2), we first prove that, for $k=3,\dots,n-1$, the vertex $v_{k+1}$ is in the outer face of $G_{k}$ (refer to \cref{fig:vKplus1-internal}); suppose, for a contradiction, that $v_{k+1}$ lies in the interior of $C_k$. Consider a directed path $P$ in $\cal D$ from $v_{k+1}$ to $z$. Such a path exists as $z$ is the only sink of $\cal D$; moreover, $P$ does not contain any vertex of $G_k$ (and hence of $C_k$) given that $\pi$ is a topological sorting of $\cal D$, hence every vertex $v_j$ of $P$ is such that $j>k$. Since $v_{k+1}$ lies in the interior of $C_k$, while $z$ lies in its exterior, by the Jordan curve's theorem we have that $P$ crosses $C_k$, a contradiction which proves that $v_{k+1}$ is in the outer face of $G_{k}$. 
We now prove that the neighbors of $v_{k+1}$ in $G_{k}$ form an (at
least $2$-element) subinterval of the path $C_{k}-(u,v)$; let $(w_1=u,w_2,\dots,w_m=v)$ be such a path; refer to \cref{fig:large-face}. Let $w_p$ and $w_q$ be the neighbors of $v_{k+1}$ such that $p$ is minimum and $q$ is maximum. Note that $p<q$, given that $v_{k+1}$ has at least two incoming edges in $\cal D$; that such edges connect $v_{k+1}$ to vertices in $C_{k}$ follows by the planarity of $G$. Suppose, for a contradiction, that there exists an index $r$ with $p<r<q$ such that $w_r$ is not a neighbor of $v_{k+1}$. Then the internal face of $G_{k+1}$ incident to both $v_{k+1}$ and $w_r$ is also incident to $w_{r-1}$ and $w_{r+1}$, hence its length is larger than three. However, since the vertices $v_{k+2},\dots,v_n$ and their incident edges lie in the outer face of $G_{k+1}$, such an internal face of $G_{k+1}$ is also  a face of $G$, a contradiction to the fact that $G$ is a maximal plane graph. This concludes the proof~of~condition~(CO-2)~and~of~the~theorem.~\end{proof}


Our proof of \cref{th:canonical-orientation-characterization} implies the following.

\begin{lemma}\label{le:from-orientation-to-ordering}
Consider any canonical orientation $\mathcal D$ with first vertex $u$ of a maximal plane graph $G$. Then any topological sorting of $\mathcal D$ is a canonical ordering of $G$ with first vertex $u$.    
\end{lemma}

\begin{proof}
Let $(u,v,z)$ be the cycle delimiting the outer face of $G$, where $u$, $v$, and $z$ appear in this counter-clockwise order along the cycle. By \cref{th:canonical-orientation-characterization}, we have that $\mathcal D$ is a $uz$-orientation in which every internal vertex has at least two incoming edges. The proof of \cref{th:canonical-orientation-characterization} shows that every topological sorting of a $uz$-orientation in which every internal vertex has at least two incoming edges is a canonical ordering of $G$ with first vertex $u$.
\end{proof}

\begin{figure}
   \centering
    \includegraphics[page=11]{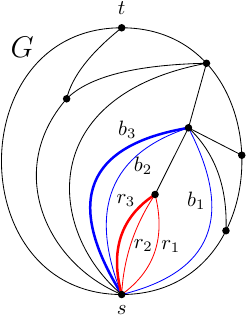}
    \caption{Illustration for the definition of a well-formed graph.}
    \label{fig:well-formed}
\end{figure}

Given two parallel edges $h_1$ and $h_2$ with end-vertices $s$ and $x$ in a plane \multigraph, we denote by $\ell(h_1,h_2)$ the open region of the plane bounded by $h_1$ and $h_2$; we say that $\ell(h_1,h_2)$ is a \emph{multilens} if it contains no vertices in its interior. Observe that $\ell(h_1,h_2)$ might contain edges parallel to $h_1$ and $h_2$ in its interior, or it might coincide with an internal face of the \multigraph. 
The leftmost edge of a maximal set of parallel edges is said to be \emph{loose}, whereas the other edges of such a set are \emph{nonloose}.  
In \cref{fig:well-formed}, any two of the (red) parallel edges $r_1$, $r_2$, and $r_3$ form a multilens, the two (blue) parallel edges $b_2$ and $b_3$ form a multilens, whereas neither $b_2$ nor $b_3$ forms a multilens with their (blue) parallel~edge~$b_1$; the multilenses $\ell(r_1,r_2)$, $\ell(r_2,r_3)$, and  $\ell(b_2,b_3)$ are also faces; loose edges are thick ($b_3$ and $r_3$), whereas nonloose edges are thin ($b_1$, $b_2$, $r_1$, and $r_2$). 

The following two definitions introduce the concepts most of this section will deal with.

\begin{definition}\label{def:well-formed}
A biconnected plane \multigraph $G$ with two distinguished vertices $s$ and $t$ is called \emph{well-formed} if it satisfies the following conditions (refer to \cref{fig:well-formed}):
\begin{enumerate}[WF1:]
    \item\label{cond:st} $s$ and $t$ are both incident to the outer face of $G$ and $s$ immediately precedes $t$ in clockwise order along the cycle $C_o$ bounding the outer face;
    \item\label{cond:faces} all the internal faces of $G$ have either two or three incident vertices;
    \item\label{cond:multiple} multiple edges, if any, are all incident to $s$; and
    \item\label{cond:lens} if there exist two parallel edges $h_1$ and $h_2$ with end-vertices $s$ and $x$ such that $\ell(h_1,h_2)$ is not a multilens, then there exist two parallel edges $h'_1$ and $h'_2$ between $s$ and a vertex $y\neq x$ such that $\ell(h'_1,h'_2)$ is a multilens and such that $\ell(h'_1,h'_2)\subset \ell(h_1,h_2)$.
\end{enumerate}
Vertices $s$ and $t$ are called \emph{poles} of $G$.
\end{definition}

\begin{definition}\label{def:good-orientation}
An $st$-orientation $\cal D$ of a well-formed biconnected plane \multigraph $G$ with poles $s$ and $t$ is \emph{inner-canonical} if every internal vertex of $G$ has at least two incoming edges in $\cal D$.
\end{definition}

We introduce notation that will be used throughout this section. Let $G$ be a well-formed biconnected plane \multigraph with poles $s$ and $t$. Let $(w_0=s,w_1,\dots,w_k=t)$ be the right path $p_r$ of $G$. Let $e_1,e_2,\dots,e_m$ be the counter-clockwise order of the edges incident to $s$, where $e_1$ is the first edge of $p_r$ and $e_m$ is the unique edge of the left path of $G$, by Condition WF1. Let $v_1,\dots,v_m$ be the end-vertices of $e_1,\dots,e_m$ different from $s$, respectively. 
Moreover, denote by $G^*$ the plane multigraph resulting from the contraction of $e_1$ in $G$.
Also, if $G$ contains parallel edges, let $j\in\{1,\dots,m-1\}$ be the smallest index such that $e_{j}$ and $e_{j+1}$ define a multilens of $G$; denote by $G^-$ the plane \multigraph resulting from the removal of $e_1,\dots,e_j$ from $G$. The next lemmata prove that, under certain conditions, $G^*$ and $G^-$ are well-formed \multigraphs.

\begin{lemma} \label{le:contract-well-formed}
Suppose that $G$ does not contain parallel edges between $s$ and $w_1$. Then $G^*$ is a well-formed biconnected plane \multigraph with~poles~$s$~and~$t$.
%

\begin{figure}[t]
    \centering
    \includegraphics[page=8]{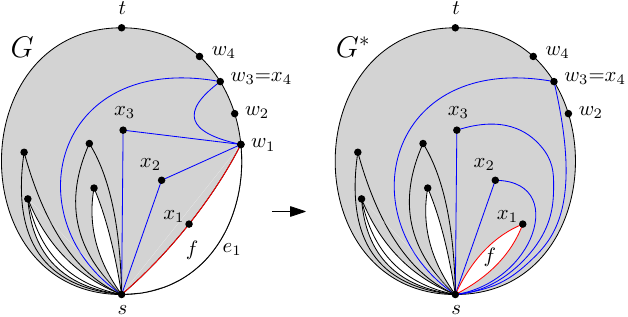}
    \caption{Illustration for the contraction of the edge $e_1$.}
    \label{fig:contraction}
\end{figure}

\end{lemma}

\begin{proof}
    Since $G$ is biconnected, if $G^*$ contains a cut-vertex, then this is necessarily $s$. It follows that $\{s,w_1\}$ is a split pair of $G$. Let $G_1,G_2,\dots,G_k$ be the components separated by $\{s,w_1\}$, where $k\geq 3$ and $G_1$ coincides with the edge $(s,w_1)$ of the right path of $G$. Since $G$ does not contain parallel edges, we have that $G_2,\dots,G_k$ are not single edges.
    Then the internal face of $G$ that is incident to $G_2$ and $G_3$ is incident to at least four vertices, a contradiction to Condition WF2 of $G$. This proves that $G^*$ is biconnected.
    
    We next prove that $G^*$ is well-formed; refer to \cref{fig:contraction}. 
    \begin{itemize}
    \item Condition WF1 follows from the fact that the left path of $G^*$, as well as the left path of $G$, is the edge $(s,t)$; this trivially follows from the fact that $t\neq w_1$, since no two parallel edges between $s$ and $w_1$ exist in $G$. 
    \item In order to prove Condition WF2, observe that $G^*$ contains no face incident to a single vertex, as the same is true for $G$, by Condition WF2 for $G$, and since no two parallel edges between $s$ and $w_1$ exist in $G$. Furthermore, that every face of $G^*$ has at most three incident vertices descends from the fact that $G$ satisfies \cref{cond:faces} and that the contraction of an edge cannot increase the number of vertices incident to a face. 
    \item Condition WF3 follows from the fact that $G$ satisfies Condition WF3 and that $G^*$ is obtained from $G$ by the contraction of an edge that has $s$ as an end-vertex; thus, new multiple edges, if any, are all incident to $s$. 
    \item Finally, we prove Condition WF4. Consider any pair $(e_1,e_2)$ of parallel edges of $G^*$ such that $\ell(e_1,e_2)$ is not a multilens. By Condition WF3, both $e_1$ and $e_2$ are incident to $s$. If $e_1$ and $e_2$ are also parallel edges of $G$ (that is, they do not become parallel edges because of the contraction of $(s,w_1)$), then there exist two parallel edges $e_3$ and $e_4$ such that $\ell(e_3,e_4)$ is a multilens and such that $\ell(e_3,e_4)\subset \ell(e_1,e_2)$ in $G^*$ as the same is true in $G$, given that $G$ satisfies Condition WF4. Otherwise, $e_1$ and $e_2$ are not parallel in $G$. Let $x_1, x_2, \dots, x_k$ be the common neighbors of $s$ and $w_1$ in $G$, listed in the order in which they appear in a clockwise visit of the adjacency list of $w_1$, starting at the vertex $x_1$ that belongs to the internal face $f$ of $G$ incident to $(s, w_1)$. Let $e^\circ=(s, x_1)$ and $e^\diamond=(w_1, x_1)$ be the edges incident to $f$ and different from $(s, w_1)$. Consider any pair $(e_1, e_2)$ of parallel edges of $G^*$ that are not parallel in $G$ and such that $\ell(e_1, e_2)$ is not a multilens. Then we have $e_1=(s, x_i)$ and $e_2=(s, x_i)$, with $2\leq i\leq k$. However, in $G^*$, it holds that the edges $e^\circ$ and $e^\diamond$ define a multilens and that $\ell(e^\circ, e^\diamond)\subset\ell(e_1, e_2)$. Therefore Condition WF4 holds for $G^*$. 
\end{itemize}
This concludes the proof that $G^*$ is well-formed, and the proof of the lemma.
\end{proof}

\begin{lemma} \label{le:remove-well-formed}
Suppose that $G$ contains parallel edges and let $j\in\{1,\dots,m-1\}$ be the smallest index such that $e_{j}$ and $e_{j+1}$ define a multilens of $G$. Suppose also that either $j=1$, or $j>1$ and $v_2,\dots,v_j$ are not incident to the outer face of $G$. 
Then the graph $G^-$ is a well-formed biconnected plane \multigraph with poles $s$ and $t$.

\end{lemma}

\begin{proof}
\begin{figure}[tb!]
\centering
\begin{subfigure}{.9\textwidth}
\centering
	\includegraphics[page=3]{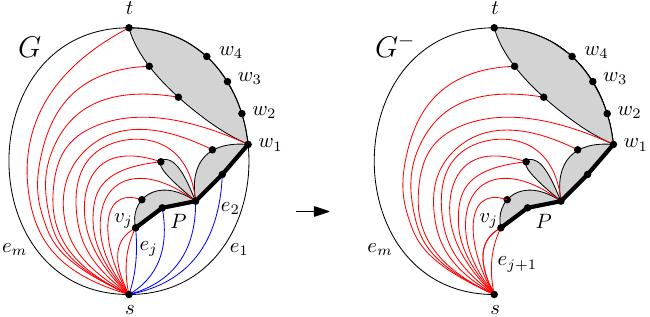}
    \caption{\label{fig:remove-well-formed-parallel-to-e1-exist}There exist parallel edges between $s$ and $w_1$.}
\end{subfigure}
\\
\begin{subfigure}{.9\textwidth}
\centering
	\includegraphics[page=16]{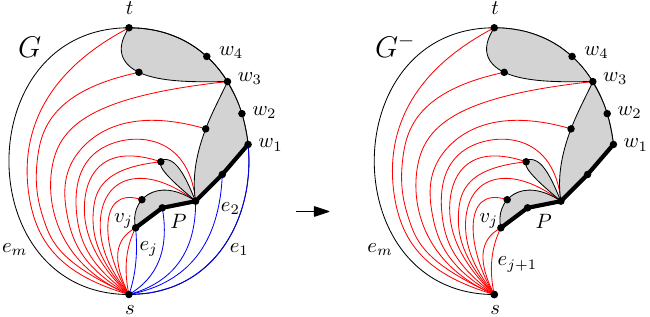}
\caption{\label{fig:remove-well-formed-parallel-to-e1-do-not-exist}There exist no parallel edges between $s$ and $w_1$.}
\end{subfigure}
\caption{\label{fig:remove-well-formed}Illustration for the removal of the edges $e_1,\dots,e_j$.
}
\end{figure}
    {\em We start with the proof for the case in which $j=1$.}
    In this case, we have that $e_2$ is also an edge between $s$ and $w_1$. Then $G^-$ clearly is a well-formed biconnected plane graph with poles $s$ and $t$. In particular, although $G^-$ does not contain the multilenses of $G$ which have $e_1$ on their boundary, no bounded region delimited by two parallel edges contains such multilenses in its interior in $G$, given that $e_1$ is incident to the outer face of $G$. 

    {\em We now consider the case in which $j>1$ and $v_2,\dots,v_j$ are not incident to the outer face of $G$}; refer to \cref{fig:remove-well-formed}. We first prove that $G^-$ is biconnected. In order to do that, we just need to prove that its outer face is bounded by a cycle, as the biconnectivity then follows from \cref{obs:split-biconnected}. By the minimality of $j$, for $i=1,\dots j$, the internal face of $G$ delimited by $e_i$ and $e_{i+1}$ is triangular. Consider the plane subgraph $P$ of $G$ formed by the edges of such triangular faces that are not incident to $s$. We argue that $P$ is a path between $w_1$ and $v_j$. Consider a clockwise Eulerian visit of the outer face of $P$ that starts at $w_1$ and ends at $v_j$. Suppose, for a contradiction, that during this visit a vertex is encountered more than once. This implies the existence of two parallel edges $e_a$ and $e_b$ with $a,b\in \{1,\dots,j\}$ and with $v_a=v_b$, which contradicts the minimality of $j$, either directly (if $e_a$ and $e_b$ define a multilens) or by Condition WF4 (otherwise). Observe that no vertex $v_i$ with $2\leq i\leq j$ is incident to the outer face of $G$, by hypothesis. Thus, the union of the path $P\cup e_{j+1}$ and of the path $C_o-e_1$ is a cycle $C^-_o$, which bounds the outer face of $G^-$.

    We next prove that $G^-$ is well-formed. 
    \begin{itemize}
        \item Condition WF1 follows from the fact that the left path of $G^-$, as well as the left path of $G$, is the edge $(s,t)$, given that $j<m$ and that $e_m$ is the left path of $G^-$. 
        \item Condition WF2 follows from the fact that $G$ satisfies Condition WF2  and that every internal face of $G^-$ is also an internal face of $G$, given that the edges $e_1,e_2,\dots,e_j$ all lie outside $C^-_o$. 
        \item Condition WF3 follows from the fact that $G$ satisfies Condition WF3 and that the edge set of $G^-$ is a subset of the edge set of $G$. 
        \item Finally, we prove Condition WF4. Consider any pair $(e_p,e_q)$ of parallel edges of $G^-$, where w.l.o.g.\ $p<q$, such that $\ell(e_p,e_q)$ is not a multilens. We prove that $\ell(e_p,e_q)$ contains a multilens in its interior in $G^-$. Since $G^-$ is a subgraph of $G$, we have that the edges $e_p$ and $e_q$ also belong to $G$. Since $G$ satisfies Condition WF4, it contains two edges $e_r$ and $e_s$ such that $\ell(e_r,e_s)$ is a multilens and such that $\ell(e_r,e_s)\subset \ell(e_p,e_q)$, which implies that $p\leq r\leq q$ and $p\leq s\leq q$. Since $e_p$ belongs to $G^-$, it follows that $p>j$. This implies that $r>j$ and that $s>j$, hence $e_r$ and $e_s$ also belong to $G^-$ and thus $\ell(e_p,e_q)$ contains a multilens in its interior in $G^-$. 
    \end{itemize}
This concludes the proof that $G^-$ is well-formed, and the proof of the lemma.
\end{proof}

Inner-canonical orientations of $G^*$ and $G^-$ can be used to construct inner-canonical orientations of $G$, as in the following two lemmata.

\begin{lemma} \label{le:expand-contraction}
Let $\mathcal D^*$ be an inner-canonical orientation of $G^*$. The orientation $\mathcal D$ of $G$ that is obtained from $\mathcal D^*$ by orienting the edge $(s,w_1)$ away from $s$ and by keeping the orientation of all other edges unchanged is inner-canonical. 
\end{lemma}

\begin{proof}
In view of \cref{obs:st-face}, in order to prove that $\mathcal D$ is an $st$-orientation, it suffices to show that all its faces are \stfaces. First observe that every face of $G$, except for the internal face $f$ incident to $(s,w_1)$ and for the outer face, is also a face of $G^*$ and that its incident edges are oriented in the same way in $\mathcal D$ and $\mathcal D^*$. Hence each such a face is a \stface. The face $f$ is bounded in $G$ by the edges $(s, x_1)$ and $(s, w_1)$, which are both outgoing $s$ in $\mathcal D$, and by the edge $(w_1, x_1)$, which is outgoing $w_1$ in $\mathcal D$. Therefore $f$ is a \stface with source $s$ and sink $x_1$. The outer face of $G$ is bounded  by the edge $(s, t)$ and the directed path $s, w_1, w_2, \dots, t$, which are both outgoing $s$ in $\mathcal D$. Therefore the outer face is a \stface with source $s$ and sink $t$. Thus $\mathcal D$ is a plane $st$-graph. Each internal vertex $w$ of $G$ is also an internal vertex of $G^*$, and thus it has two incoming edges in $\mathcal D$ since the same property is true in $\mathcal D^*$. 
\end{proof}

\begin{lemma} \label{le:expand-removal}
Let $\mathcal D^-$ be an inner-canonical orientation of $G^-$. The orientation $\mathcal D$ of $G$ that is obtained from $\mathcal D^-$ by orienting the edges $e_1,e_2,\dots,e_j$ away from $s$ and by keeping the orientation of all other edges unchanged is inner-canonical. \end{lemma}

\begin{proof}
Clearly, $\mathcal D$ is an $st$-orientation, given that $\mathcal D^-$ is an $st$-orientation and that all the edges $e_1,e_2,\dots,e_j$ are oriented away from the single source $s$ of $\mathcal D^-$. Every internal vertex $w$ of $G$ is either an internal vertex of $G^-$, and thus it has two incoming edges in $\mathcal D$ since the same property is true in $\mathcal D^-$, or is an end-vertex of an edge $e_i$, for some $i\in\{2,\dots,j\}$. In the latter case, $w$ has at least two incoming edges in $\mathcal D$, namely $e_i$ and at least one incoming edge it also has in $\mathcal D^-$.     
\end{proof}

A crucial consequence of  \cref{le:contract-well-formed,le:remove-well-formed,le:expand-contraction,le:expand-removal} is the following.

\begin{lemma}\label{le:at-least-one-inner-canonical}
Every well-formed biconnected plane \multigraph $G$ with poles $s$ and $t$ has at least one inner-canonical orientation.
\end{lemma}

\begin{proof}
The proof is by induction on the number of edges of $G$. In the base case, $G$ is a single edge between $s$ and $t$. Then the orientation of such an edge from $s$ to $t$ trivially is inner-canonical. For the inductive case, we distinguish two cases.

{\bf There exist no parallel edges between $s$ and $w_1$}; refer to \cref{fig:contraction}. By \cref{le:contract-well-formed}, the plane \multigraph $G^*$  obtained by the contraction of $e_1=(s,w_1)$ in $G$ is biconnected and well-formed (with poles $s$ and $t$). Thus, by induction, it admits an inner-canonical orientation~$\mathcal D^*$. By \cref{le:expand-contraction}, orienting the edge $e_1$ away from $s$ and keeping the orientation of all other edges unchanged turns $\mathcal D^*$ into an inner-canonical orientation $\mathcal D$ of $G$.

\textbf{There exist parallel edges between $s$ and $w_1$}; refer to \cref{fig:remove-well-formed-parallel-to-e1-exist}. In order to prove that $G$ admits an inner-canonical orientation, it suffices to prove that 
the index $j$, defined in \cref{le:remove-well-formed} as the smallest index such that $e_j$ and $e_{j+1}$ define a multilens, exists.
Indeed, if such an index exists, we have that, by \cref{le:remove-well-formed}, the plane \multigraph $G^-$ obtained from~$G$ by removing the edges $e_1,e_2,\dots,e_j$ is well-formed and thus, by induction, it admits an inner-canonical orientation $\mathcal D^-$. Also, by \cref{le:expand-removal}, orienting the edges $e_1,e_2,\dots,e_j$ away from $s$ turns $\mathcal D^-$ into an inner-canonical orientation $\mathcal D$ of $G$, which proves the statement.

We now show that $j$ exists. By hypothesis, there exist two parallel edges between $s$ and $w_1$. Since $e_1$ connects $s$ and $w_1$, it follows that $e_1$ is one of such edges. Let $e_h$ be a distinct edge also connecting $s$ and $w_1$. By Condition WF4, there exist two edges $e_p$ and $e_{p+1}$ that define a multilens and such that $1\leq p<p+1\leq h$. We show that choosing $j$ as the smallest index $p$ such that $e_p$ and $e_{p+1}$ define a multilens allows \cref{le:remove-well-formed} to be applied. If $j=1$, then \cref{le:remove-well-formed} trivially applies. If $j>1$, consider any edge $e_i$ such that $1<i\leq j$. We argue that $v_i$ is not incident to the outer face of $G$, which allows \cref{le:remove-well-formed} to be applied. Suppose the contrary, for a contradiction. We have $v_i\neq w_1$, as otherwise $e_1$ and $e_i$ would be parallel edges, and thus, by Condition WF4, there would exist two edges $e_p$ and $e_{p+1}$ that define a multilens and such that $1\leq p<p+1\leq i$, contradicting the minimality of~$j$. Since $v_i\neq w_1$, we have that $v_i$ lies in the exterior of $\ell(e_1, e_h)$. From this and from the fact that $e_i$ appears between $e_h$ and $e_1$ in left-to-right order around $s$, we have that $e_i$ crosses the cycle composed of $e_1$ and $e_h$, contradicting the planarity of $G$.
\end{proof}

\cref{sec:alg-ICE,sec:alg-ICE-fast} are devoted to the proof of the following main result.

\begin{theorem}\label{th:innercanonical-orientation-plane-uv}
	Let $G$ be a well-formed biconnected plane \multigraph with $\varphi$ edges. There exists an algorithm with $\bigoh(\varphi)$ setup time and $\bigoh(\varphi)$ space usage that lists all the inner-canonical orientations of $G$ with $\bigoh(\varphi)$ delay.
\end{theorem}



Provided that \cref{th:innercanonical-orientation-plane-uv} holds, we can prove the following.

\begin{lemma} \label{th:canonical-orientation-plane-uv}
	Let $G$ be an $n$-vertex maximal plane graph and let $(u,v,z)$ be the cycle delimiting its outer face, where $u$, $v$, and $z$ appear in this counter-clockwise order along the outer face of $G$. There exists an algorithm with $\bigoh(n)$ setup time
 and $\bigoh(n)$ space usage that lists all the canonical orientations of $G$ with first vertex $u$ with $\bigoh(n)$ delay.
\end{lemma}

\begin{proof}
Since $G$ is a biconnected, in fact triconnected, well-formed plane graph with poles $u$ and $z$, it suffices to prove that any inner-canonical orientation of $G$ is also a canonical orientation of $G$ with first vertex $u$, and vice versa. Namely, this and the fact that $G$ has $\bigoh(n)$ edges imply that the algorithm in \cref{th:innercanonical-orientation-plane-uv} enumerates all canonical orientations of $G$ within the stated bounds.

By \cref{th:canonical-orientation-characterization}, any canonical orientation of $G$ with first vertex $u$ is a $uz$-orientation such that every internal vertex has at least two incoming edges, hence it is an inner-canonical orientation of $G$.  Conversely, any inner-canonical orientation $\cal D$ of $G$ is also canonical. Indeed, by definition $\cal D$ is a $uz$-orientation such that every internal vertex has at least two incoming edges. By \cref{th:canonical-orientation-characterization}, we have that $\cal D$ is a canonical orientation with first vertex $u$.
\end{proof}

%

\begin{theorem} \label{th:canonical-orientation-plane-and-planar}
	Let $G$ be an $n$-vertex maximal plane (resp.\ planar) graph. There exists an algorithm $\mathcal{A}_1$ (resp.\ $\mathcal{A}_2$) with $\bigoh(n)$ setup time and $\bigoh(n)$ space usage that lists all canonical orientations of $G$ with $\bigoh(n)$ delay.
\end{theorem}

\begin{proof}
The algorithm $\mathcal{A}_1$ uses the one for the proof of \cref{th:canonical-orientation-plane-uv} three times, namely once for each choice of the first vertex among the three vertices incident to the outer face of $G$. 
The algorithm $\mathcal{A}_2$ uses the algorithm $\mathcal{A}_1$ applied $4n-8$ times; this is because there are $4n-8$ maximal plane graphs that are isomorphic to $G$. Namely, the cycle $(u,v,z)$ delimiting the outer face of a maximal plane graph isomorphic to $G$ can be chosen among the $2n-4$ facial cycles (of any planar drawing) of $G$, and the vertices can appear in counter-clockwise order $u,v,z$ or $u,z,v$ along the boundary of the outer face. Note that any two orientations produced by different applications of algorithm $\mathcal{A}_2$  differ on the source of the orientation, or on the sink of the orientation, or on the non-source and non-sink vertex incident to~the~outer~face.\end{proof}

\subsection{The Inner-Canonical Enumerator Algorithm}\label{sec:alg-ICE}

\begin{figure}[tb!]
    \centering
\includegraphics[scale =.8]{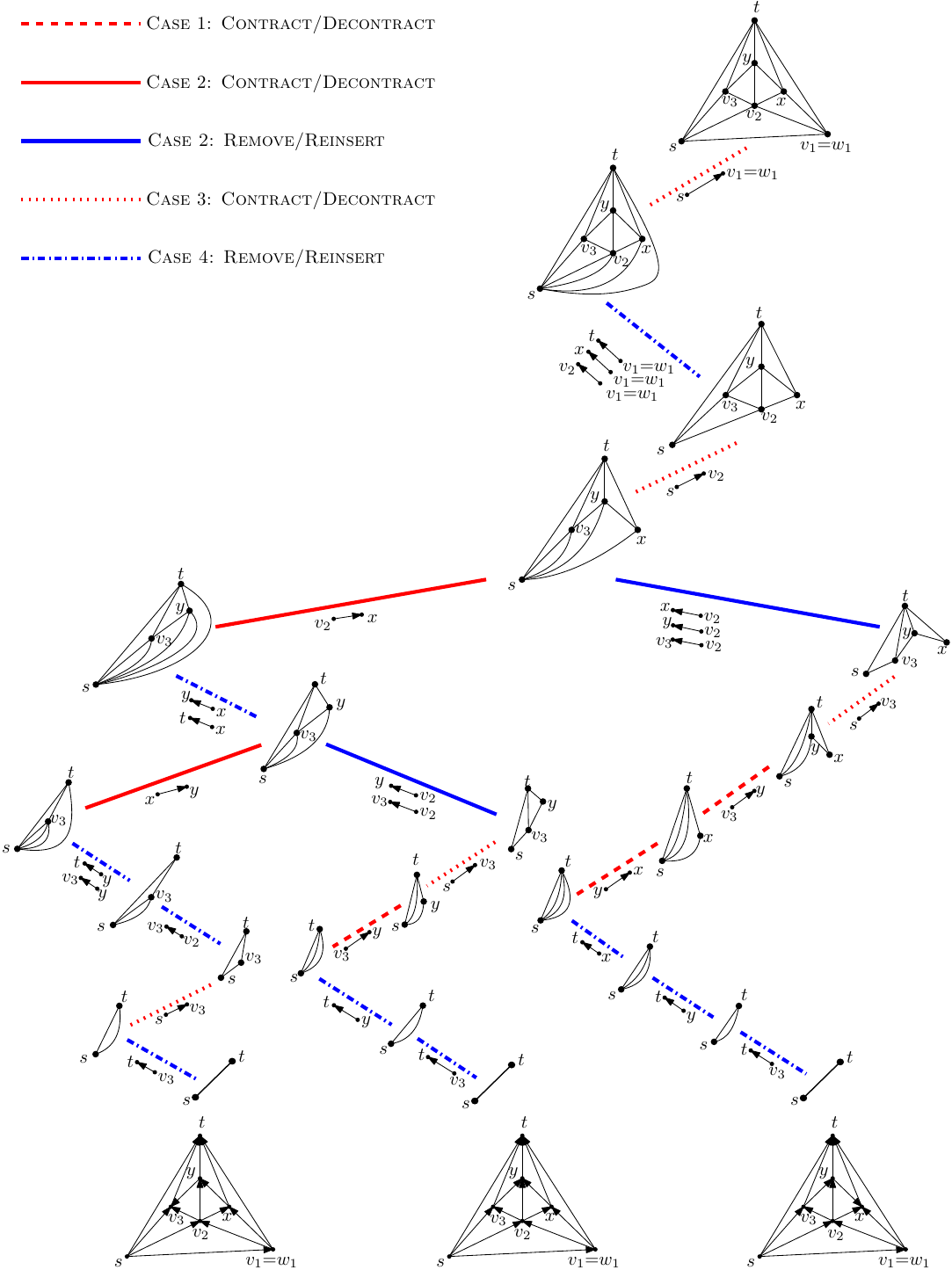}
    \caption{Illustration of the call tree of an execution of the {\sc ICE} algorithm. Each of the three inner-canonical orientations of the graph in the root node is shown below the~corresponding~leaf~node.}
    \label{fig:computation}
\end{figure}

We are now ready to describe an algorithm that takes in input a well-formed biconnected plane \multigraph $G$ with poles $s$ and $t$, and enumerates all its inner-canonical orientations (confr.\ \cref{th:innercanonical-orientation-plane-uv}). 
The algorithm, which we call {\sc Inner-Canonical Enumerator} ({\sc ICE}, for short), works recursively; refer to \cref{fig:computation}. In the base case, $G$ is the single edge $e_m=(s,t)$, and its unique inner-canonical orientation is the one in which the edge $e_m$ is directed from $s$ to $t$. Otherwise, the algorithm distinguishes four cases. In Cases 1 and 2, $G$ contains parallel edges and $e_1$ is the unique edge between $s$ and $w_1$. Let $j\in\{2,\dots,m-1\}$ be the smallest index such that $e_{j}$ and $e_{j+1}$ define a multilens of $G$; note that $j>1$ by the above assumption. In Case 1, there exists an index $i\in \{2,\dots,j\}$ such that $v_i$ is incident to the outer face of $G$, while in Case 2 such an index does not exist. In Case 3, $G$ does not contain parallel edges. Finally, in Case 4, $G$ contains parallel edges between $s$ and $w_1$. Note that exactly one of Cases~1--4~applies~to~$G$. 

\begin{itemize}
    \item In Cases 1 and 3, {\bf we contract} the edge $(s,w_1)$. Let $G^*$ be the resulting plane \multigraph and note that, by \cref{le:contract-well-formed}, $G^*$ is biconnected and well-formed. Thus, the  {\sc ICE} algorithm can be applied recursively in order to enumerate all the inner-canonical orientations of $G^*$. The {\sc ICE} algorithm then obtains all the inner-canonical orientations of $G$ as follows: For every inner-canonical orientation $\mathcal D^*$ of $G^*$, the algorithm constructs one inner-canonical orientation of $G$ by orienting the edge $(s,w_1)$ away from $s$ and by keeping the orientation of all other edges unchanged, where some edges that are incident to $s$ in $G^*$ are instead incident to $w_1$ in $G$; these are all outgoing $s$ in $\mathcal D^*$ and all outgoing $w_1$ in~$\mathcal D$. 

    \item In Case 4, {\bf we remove} the edges $e_1,e_2,\dots,e_j$. Let $G'$ be the resulting plane \multigraph and note that, by \cref{le:remove-well-formed}, $G'$ is biconnected and well-formed. Thus, the {\sc ICE} algorithm can be applied recursively in order to enumerate all the inner-canonical orientations of $G'$. The {\sc ICE} algorithm then obtains all the inner-canonical orientations of $G$ as follows: For every inner-canonical orientation $\mathcal D'$ of $G'$, the algorithm constructs one inner-canonical orientation of $G$ by orienting the edges $e_1,e_2,\dots,e_j$ away from $s$ and by keeping the orientation of all other edges unchanged. 

    \item In Case 2, the {\sc ICE} algorithm branches and applies {\bf both the contraction and the removal operations}. More formally, first we contract the edge $(s,w_1)$, obtaining a well-formed biconnected plane \multigraph $G^*$, by \cref{le:contract-well-formed}. From every inner-canonical orientation $\mathcal D^*$ of $G^*$, the algorithm constructs one inner-canonical orientation of $G$, same as in Cases 1 and 3. After all the inner-canonical orientations of $G^*$ have been used to produce inner-canonical orientations of $G$, we remove the edges $e_1,e_2,\dots,e_j$ from $G$, obtaining  a well-formed biconnected plane \multigraph $G'$, by \cref{le:remove-well-formed}. From every inner-canonical orientation $\mathcal D'$ of $G'$, the algorithm constructs one inner-canonical orientation of $G$, same as in Case 4.
\end{itemize}

We remark that the {\sc ICE} algorithm outputs an inner-canonical orientation every time the base case applies. The next three lemmata prove the correctness of the algorithm. We will later describe, in \cref{sec:alg-ICE-fast},  how to efficiently implement it. 

\begin{lemma} \label{le:every-inner}
Every orientation of $G$ listed by the {\sc ICE} algorithm is inner-canonical. 
\end{lemma}

\begin{proof}
The proof is by induction on the size of $G$. The statement is trivial in the base case, hence suppose that one of Cases 1--4 applies. 

In Cases 1, 2, and 3, by \cref{le:contract-well-formed}, the graph $G^*$ constructed by the algorithm is well-formed. Hence, by induction, every orientation $\mathcal D^*$ that is an output of the recursive call to {\sc ICE} with input $G^*$ is inner-canonical. Starting from $\mathcal D^*$, the {\sc ICE} algorithm constructs one orientation $\mathcal D$ of $G$ by orienting the edge $(s,w_1)$ away from $s$ and by keeping the orientation of all other edges unchanged. By \cref{le:expand-contraction}, we have that $\mathcal D$ is inner-canonical.
    
In Cases 2 and 4, by \cref{le:remove-well-formed}, the graph $G^-$ constructed by the algorithm is well-formed. Hence, by induction, every orientation $\mathcal D^-$ that is an output of the recursive call to {\sc ICE} with input $G^-$ is inner-canonical. Starting from $\mathcal D^-$, the {\sc ICE} algorithm constructs one orientation $\mathcal D$ of $G$ by orienting the edges $e_1,e_2,\dots,e_j$ away from $s$ and by keeping the orientation of all other edges unchanged. By \cref{le:expand-removal}, we have that $\mathcal D$ is inner-canonical.
  
This completes the induction and hence the proof of the lemma.
\end{proof}

\begin{lemma} \label{le:all-inner}
The {\sc ICE} algorithm outputs all the inner-canonical orientations of $G$. 
\end{lemma}

\begin{proof}
The proof is by induction on the size of $G$. The statement is trivial in the base case, when $G$ is the single edge $(s,t)$. Otherwise, suppose that one of Cases 1--4 applies. 

Suppose, for a contradiction, that there exists an inner-canonical orientation $\mathcal D$ of $G$ that is not generated by the {\sc ICE} algorithm. We distinguish two cases based on the structure of $G$ and on the orientation of the edges in $\mathcal D$. In Case A, we have that $G$ satisfies Case 1 of the algorithm, or $G$ satisfies Case 3 of the algorithm, or $G$ satisfies Case 2 of the algorithm and the edge $(v_1,v_2)$ is outgoing $v_1$ in $\mathcal D$; recall that $e_1,\dots,e_m$ is the counter-clockwise order of the edges incident to $s$, where $e_1$ is the edge in the right path of $G$, and that $v_1,\dots,v_m$ are the end-vertices of $e_1,\dots,e_m$ different from $s$, respectively. In Case B, we have that $G$ satisfies Case 4 of the algorithm, or $G$ satisfies Case 2 of the algorithm and the edge $(v_1,v_2)$ is outgoing $v_2$ in $\mathcal D$.
In Case A, consider the orientation $\mathcal D^*$ of $G^*$ resulting from the contraction of $e_1$ in $\mathcal D$. We prove below that $\mathcal D^*$ is inner-canonical. Then, by induction, it is generated by the algorithm. Therefore, since by expanding the edge $e_1$ (as in Cases 1, 2, and 3 of the algorithm) we obtain $\mathcal D$, we get a contradiction.
Analogously, in Case B, consider the orientation $\mathcal D^-$ of $G^-$ resulting from the removal of the edges $e_1, e_2, \dots,e_j$ in $\mathcal D$, where $j$ is the smallest index such that $e_{j}$ and $e_{j+1}$ define a multilens of $G$. We prove below that $\mathcal D^-$ is inner-canonical. Then, by induction, it is generated by the algorithm. Therefore, since by reinserting the edges $e_1, e_2,\dots e_j$ (as in Cases 2 and 4 of the algorithm) we obtain $\mathcal D$, we get a contradiction.
It remains to prove that $\mathcal D^*$ in Case A and $\mathcal D^-$ in Case B are inner-canonical orientations.

{\bf The orientation $\mathcal D^*$ is inner-canonical in Case A.} 
First, every internal vertex of $G^*$ has the same incident edges in $\mathcal D$ and in $\mathcal D^*$ (up to renaming the end-vertex of some of these edges from $v_1$ to $s$), hence every internal vertex has at least two incoming edges in $\mathcal D^*$ since the same is true in $\mathcal D$. Second, we show that $\mathcal D^*$ is an $st$-orientation of $G^*$. In the following, we first assume that the edge $(v_1,v_2)$ is outgoing $v_1$ in $\mathcal D$ and show that $\mathcal D^*$ is an $st$-orientation of $G^*$ under this assumption. We will then show that the edge $(v_1,v_2)$ is indeed outgoing $v_1$ in $\mathcal D$. 

In view of \cref{obs:st-face}, in order to prove that $D^*$ is an $st$-orientation, it suffices to show that all its faces are \stfaces. Let $f$ be the internal face of $G$ incident to $e_1$. Since we are not in Case 4 of the algorithm, we have that $f$ is a triangular face delimited by the cycle $(s,v_1,v_2)$. Observe that every face of $G^*$, except for $f$ and for the outer face, is also a face of $G$ and that its incident edges are oriented in the same way in $\mathcal D^*$ and $\mathcal D$. Hence, each such a face is a \stface of $\mathcal D^*$. The face $f$ is bounded in $G^*$ by two parallel edges between $s$ and $v_2$: One of them is $e_2$ and the other one is the edge $(v_1, v_2)$ after the contraction that identifies $v_1$ and $s$. Both these edges are outgoing $s$ in $\mathcal D^*$, since $e_2$ is outgoing $s$ in $\mathcal D$ and since $(v_1,v_2)$ is outgoing $v_1$ in $\mathcal D$, by hypothesis. Therefore, $f$ is a \stface in $\mathcal D^*$ with source $s$ and sink $v_2$. The outer face is bounded in $G^*$ by the edge $(s, t)$ and the directed path $s, w_2, \dots, t$, which are both outgoing $s$ in $\mathcal D^*$. Thus, the outer face of $\mathcal D^*$ is a \stface with source $s$ and sink $t$. This concludes the proof that $\mathcal D^*$ is an $st$-orientation.

It remains to prove that the edge $(v_1,v_2)$ is outgoing $v_1$ in $\mathcal D$. If $G$ satisfies Case 2 of the algorithm, then the statement trivially holds true by the hypotheses of Case A. In Cases 1 and 3, let $k$ be the smallest index such that $v_k$ is a neighbor of $s$ incident to the outer face; such an index exists as $v_m=t$ is incident to the outer face of $G$. Since in Case 3 there are no parallel edges in $G$ and in Case 1 there are no two edges $e_h$ and $e_{h+1}$ defining a multilens, for any $h\in \{1,\dots,k\}$, we have that the internal faces of $G$ delimited by $e_i$ and $e_{i+1}$, for $i = 1,\dots, k-1$, are triangular. Consider the plane subgraph $P$ of $G$ formed by the edges of such triangular faces that are not incident to $s$. Such a graph is a path between $v_1$ and $v_k$ (confr.\ with the proof of \cref{le:remove-well-formed}). Suppose, for a contradiction, that the edge $(v_1, v_2)$ is outgoing $v_2$ in $\mathcal D$. Consider the maximal directed subpath $v_x,v_{x-1},\dots,v_1$ of $P$ such that the edge $(v_i,v_{i-1})$ is outgoing $v_i$, for $i=x,\dots,2$.
If $x=k$, we have that $\mathcal D$ contains a directed cycle formed by the path $P$ and the subpath of $p_r$ between $v_1$ and $v_k$, which contradicts the fact that $\mathcal D$ is an $st$-orientation.
If $x<k$, consider the edges $(v_x,v_{x-1})$ and $(v_x,v_{x+1})$, that are both outgoing $v_x$ in $\mathcal D$.
First, $(s,v_x)$ is incoming $v_x$ in $\mathcal D$ and it is the only edge incident to $v_x$ that follows $(v_x,v_{x+1})$ and precedes $(v_x,v_{x-1})$ in counter-clockwise order around $v_x$.
Second, by \cref{obs:bimodality}, all the edges of $G$ incident to $v_x$ that follow $(v_x,v_{x-1})$ and precede $(v_x,v_{x+1})$ in counter-clockwise order around $v_x$ are outgoing $v_x$. Hence, $v_x$ has only one incoming edge. Since $v_x$ is an internal vertex of $G$, we have a contradiction to the fact that $\mathcal D$ is inner-canonical.
This concludes the proof that the edge $(v_1,v_2)$ is outgoing $v_1$ in $\mathcal D$ in Case~A.

{\bf The orientation $\mathcal D^-$ is inner-canonical in Case B.} First, all the internal vertices of $G^-$ are also internal vertices of $G$, and thus they have two incoming edges in $\mathcal D^-$ as they also do in $\mathcal D$. It remains to prove that $\mathcal D^-$ is an $st$-orientation. With this aim, in view of \cref{obs:st-face}, it suffices to show that all the faces of $\mathcal D^-$ are \stfaces. Since each internal face of $\mathcal D^-$ is also an internal face of $\mathcal D$, we have that it is a \stface. We now show that the outer face of $\mathcal D^-$ is a \stface. The left path of the outer face of $G^-$ coincides with the edge $(s,t)$, hence it is a directed path from $s$ to $t$. We now need to prove that the right path $p^-_r$ of  $G^-$ is also a directed path from $s$ to $t$ in $\mathcal D^-$. Part of $p^-_r$ is the subpath of $p_r$ between $v_1$ and $t$; this is a directed path from $v_1$ to $t$ in $\mathcal D^-$, since $p_r$ is a directed path from $s$ to $t$ in $\mathcal D$. 
If the edges $e_1$ and $e_2$ define a multilens (which implies that Case 4 applies), then $p^-_r$ is completed with the edge $e_2$, which is directed from $s$ to $v_1$ in $\mathcal D^-$. 
Otherwise, we have that the internal faces of $G$ delimited by the edges $e_i$ and $e_{i+1}$, for $i = 1,\dots, j-1$, are triangular. Consider the plane subgraph $P$ of $G$ formed by the edges of such triangular faces that are not incident to $s$. Such a graph is a path between $v_1$ and $v_j$ (confr.\ with the proof of \cref{le:remove-well-formed}).
In order to prove that $p^-_r$ is a directed path from $s$ to $t$ in $\mathcal D^-$, it suffices to prove that that $P$ is oriented from $v_j$ to $v_1$ in $\mathcal D$ (and thus also in $\mathcal D^-$).
First, we show that $P$ is oriented from $v_j$ to $v_1$ in $\mathcal D$ under the assumption that the edge $(v_1,v_2)$ is outgoing $v_2$. Then, we show that such an edge is outgoing $v_2$ in Case B.
Consider the maximal directed subpath $v_x,\dots,v_1$ of $P$ such the edge $(v_i,v_{i-1})$ is outgoing $v_i$, for $i=x,\dots,2$.
If $x=j$, we have that such a subpath coincides with $P$, and thus $P$ is directed from $v_j$ to $v_1$ in $\mathcal D$, as desired.
Otherwise (i.e., when $1<x<j$), we have that both the edges $(v_x,v_{x-1})$ and $(v_x,v_{x+1})$ are outgoing $v_x$. As in the discussion for Case~A, this implies that $v_x$ has only one incoming edge in $\mathcal D$, which is not possible since $\mathcal D$ is inner-canonical. This concludes the proof that the path $P$ is directed from $v_j$ to $v_1$ in $\mathcal D$, under the assumption that the edge $(v_1,v_2)$ is outgoing $v_2$.

Next, we prove that the edge $(v_1,v_2)$ is outgoing $v_2$ in Case B. In Case~2, this is true by hypothesis. In Case~4, $G$ contains parallel edges between $s$ and $v_1$. 
Let $e_k$ be the edge parallel to $e_1$ with the smallest index. We have that, by the construction of $P$, the edge $(v_2,v_1)$ follows $e_1$ and precedes $e_k$ in clockwise order around $v_1$. If $v_1 = t$, then $(v_2,v_1)$ is outgoing $v_2$ since $\mathcal D$ is an $st$-orientation. Otherwise, we have that the edge $(v_1=w_1,w_2)$ exists and is outgoing $v_1$. Also, such an edge follows $e_1$ and precedes $e_k$ in counter-clockwise order around $v_1$. Therefore, by \cref{obs:bimodality}, all the edges of $G$ incident to $v_1$ that follow $e_1$ and precede $e_k$ in clockwise order around $v_1$ are incoming $v_1$. This concludes the proof that the edge $(v_1,v_2)$ is outgoing $v_2$ in $\mathcal D$ in Case B, and the proof of the lemma.
\end{proof}

\begin{lemma} \label{le:all-inner-once}
The {\sc ICE} algorithm outputs every inner-canonical orientation of $G$~once. 
\end{lemma}

\begin{proof}
The proof is by induction on the size of $G$. The statement is trivial in the base case, when $G$ is the single edge $(s,t)$; indeed, in this case no recursion is applied and hence the algorithm outputs the (unique) inner-canonical orientation of $G$ only once. 

Suppose now that $G$ contains more than one edge. Also suppose, for a contradiction, that the algorithm produces (at least) twice the same inner-canonical orientation $\mathcal D$ of $G$. We distinguish three cases.

First, suppose that the algorithm produces $\mathcal D$ both by a ``decontraction'' of an inner-canonical orientation $\mathcal D^*_1$ of $G^*$ and by a decontraction of an inner-canonical orientation $\mathcal D^*_2$ of $G^*$. We show that $\mathcal D^*_1$ and $\mathcal D^*_2$ are the same orientation. Indeed, $\mathcal D$ is obtained (from each of $\mathcal D^*_1$ and $\mathcal D^*_2$) by orienting the edge $(s,w_1)$ away from $s$ and by keeping the orientation of all other edges unchanged. Hence, if $\mathcal D^*_1$ and $\mathcal D^*_2$ were different, then also the orientations of $G$ resulting from the decontractions of $\mathcal D^*_1$ and $\mathcal D^*_2$ would be different, while they are both equal to $\mathcal D$. Since $\mathcal D^*_1$ and $\mathcal D^*_2$ are the same orientation, by induction, the algorithm outputs such an orientation only once, hence  the algorithm outputs $\mathcal D$ only once, as well, a contradiction.

Second, suppose that the algorithm produces $\mathcal D$ both by a ``reinsertion'' of directed edges in an inner-canonical orientation $\mathcal D^-_1$ of $G^-$ and by a reinsertion of directed edges in an inner-canonical orientation $\mathcal D^-_2$ of $G^-$. We show that $\mathcal D^-_1$ and $\mathcal D^-_2$ are the same orientation. Indeed, $\mathcal D$ is obtained (from each of $\mathcal D^-_1$ and $\mathcal D^-_2$) by orienting the edges $e_1,e_2,\dots,e_j$ away from $s$ and by keeping the orientation of all other edges unchanged. Hence, if $\mathcal D^-_1$ and $\mathcal D^-_2$ were different, then also the orientations of $G$ resulting from the reinsertion of $e_1,e_2,\dots,e_j$ in $\mathcal D^-_1$ and $\mathcal D^-_2$ would be different, while they are both equal to $\mathcal D$. Since $\mathcal D^-_1$ and $\mathcal D^-_2$ are the same orientation, by induction, the algorithm outputs such an orientation only once, hence the algorithm outputs $\mathcal D$ only once, as well, a contradiction.

Finally, suppose that the algorithm produces $\mathcal D$ both by a decontraction of an inner-canonical orientation $\mathcal D^*_1$ of $G^*$ and by a reinsertion of directed edges in an inner-canonical orientation $\mathcal D^-_2$ of $G^-$. We show that the edge $(v_1,v_2)$ of $G$ is oriented differently in the inner-canonical orientation $\mathcal D^*$ of $G$ resulting from the decontraction of $\mathcal D^*_1$ and in the inner-canonical orientation $\mathcal D^-$ of $G$ resulting from the reinsertion of directed edges in $\mathcal D^-_2$. This contradicts the fact that $\mathcal D^*$ and $\mathcal D^-$ are both equal to $\mathcal D$. On the one hand, in $\mathcal D^*_1$, the vertex $v_1$ is identified with $s$, hence the edge $(v_1,v_2)$ is outgoing $v_1$. On the other hand, in $\mathcal D^-_2$, the edge $(v_1,v_2)$ of $G$ belongs to the right path of the outer face of $G^-$, with $s,v_2,v_1,t$ in this order along such a path. Hence, the edge $(v_1,v_2)$ is outgoing $v_2$. This completes the induction and hence the proof of the lemma.
\end{proof}

\cref{le:every-inner,le:all-inner,le:all-inner-once} complete the proof of correctness of the {\sc ICE} algorithm. 

\subsection{Efficient Implementation of the ICE Algorithm}\label{sec:alg-ICE-fast}
%
By \cref{le:every-inner,le:all-inner,le:all-inner-once}, the {\sc ICE} algorithm outputs all and only the inner-canonical orientations of $G$ once. In the following, 
we show how to efficiently implement the {\sc ICE} algorithm in order to achieve the stated bounds (confr.\ \cref{th:innercanonical-orientation-plane-uv}).
The pseudocode of the algorithm is given in \cref{alg:InnerCanonicalEnumerator}. 

In the following, we call \emph{left-to-right order} around $s$ the linear order of the edges incident to $s$ obtained by visiting in clockwise order such edges starting from $e_m$ and ending at $e_1$. Analogously, we call \emph{right-to-left order} around $s$ the linear order of the edges incident to $s$ obtained by visiting in counter-clockwise order such edges starting from $e_1$ and ending at $e_m$. This allows us to properly refer to an edge incident to a neighbor $v_i$ of $s$ as to the \emph{rightmost} (resp.\ \emph{leftmost}) \emph{edge} incident to $v_i$ of a specific type (e.g., the leftmost nonloose parallel edge incident to $v_i$ or the rightmost parallel edge incident to $v_i$).

\begin{algorithm}[tb!]
\SetKwData{G}{G}
\SetKwData{GG}{G(V,E)}
\SetKwData{V}{V} 
\SetKwData{info}{$\mathbf{A}$} 
\SetKwData{thresh}{$\tau$} 
\SetKwData{lsh}{$S$} 
\SetKwData{bbls}{$T$} 
\SetKwData{comms}{C} 
\SetKwData{com}{$C_{s}$} 
\SetKwData{node}{s} 
\SetKwData{cnode}{v} 
\SetKwData{buck}{B} 
\SetKwData{cand}{y} 
\SetKwData{enq}{\textsc{enqueue}} 
\SetKwData{deq}{\textsc{dequeue}} 
\SetKwFunction{dist}{$d_{ST}$}
\DontPrintSemicolon
\SetKwInOut{Input}{Input}
\SetKwInOut{Output}{Output}
\SetKwInOut{Global}{Global}
\Global{\Vertex $S$ \codecomment{(the pole $s$ of a well-formed graph $G$)}\\ 
\Edge{[}\ {]} EDGES  \codecomment{(the array of the edges of $G$)}
}
\Output{A sequence of inner-canonical orientations of $G$}
case $\leftarrow$ \texttt{DetectCase()}\\
\If{case = CONTRACT or case = CONTRACT\&REMOVE}{
    Edge $e_1 \leftarrow$ \texttt{Contract()}\\
    \texttt{InnerCanonicalEnumerator()}\\
    \texttt{Decontract($e_1$)}}
\If{case = REMOVE or case = CONTRACT\&REMOVE}{Edges{[}\ {]} removedEdges $ \leftarrow$ \texttt{Remove()}\\
    \texttt{InnerCanonicalEnumerator()}\\
    \texttt{Reinsert(removedEdges)}}
\If{case = BASE}{\texttt{Output()}}
\BlankLine
 
\caption{Algorithm {\sc InnerCanonicalEnumerator}}
\label{alg:InnerCanonicalEnumerator}\vspace{-0.2cm}
\end{algorithm}

\begin{algorithm}[h!]
  \caption{{Data Structures.} 
  \label{algo:datastructures}
  \underline{Underlined} pointers of a record of type \Vertex might be different from \NULL only if the vertex is $s$.
  \underline{Underlined} pointers of~a record of type \Edge might be different from \NULL only if one of the end-vertices of the edge is $s$.
  } 
  \Struct{Vertex}{
    \Int{} degree;
    \codecomment{integer representing the degree of the vertex}\\
    \Boolean{} is\_outer;  \codecomment{\texttt{True} if the vertex is incident to the outer face, \texttt{False} otherwise}\\
    \Edge{} first\_incident\_to\_s;  \codecomment{reference to the rightmost edge incident to the current vertex and to $s$; 
    \NULL if the vertex is $s$ or if the vertex is not adjacent to $s$}\\
    \Edge{} \underline{$e_1$}; \codecomment{reference to the rightmost edge incident to the vertex and to $s$}\\
    \Edge{} \underline{first\_chord}; \codecomment{reference to the rightmost chord incident to $s$}\\
    \Edge{} \underline{first\_parallel}; \codecomment{reference to the rightmost parallel edge incident to $s$}\\
    \Edge{} \underline{first\_lens}; \codecomment{reference to the rightmost edge of a multilens incident to $s$}
 
    }
  \Struct{Edge}{
    \Vertex{} x, y; \codecomment{references to the end-vertices of the edge}\\
    \Boolean{} oriented\_from\_x\_to\_y; \codecomment{\texttt{True} if the edge is oriented from $x$ to $y$, \texttt{False} otherwise}\\
    \Boolean{} is\_outer; \codecomment{\texttt{True} if the edge is an outer edge, \texttt{False} otherwise}\\
    \Edge{} next\_around\_x, next\_around\_y; 
    \codecomment{Reference to the edge that follows the current edge in counter-clockwise order around $x$ (around $y$) in $G$.}\\
    \Int{} \underline{ord}; \codecomment{integer representing the position of the edge in the left-to-right order around $s$}\\
    \Edge{} \underline{next\_parallel\_with\_me}; \codecomment{reference to the edge $(x,y)$ that follows the current edge in right-to-left order around $s$; not \NULL only if the current edge is one of the parallel edges between $x$ and $y$}\\
    \Edge{} \underline{next\_chord}; 
    \codecomment{reference to the chord that  follows the current edge in the right-to-left order around $s$; not \NULL only if the current edge is a chord}\\
    \Edge{} \underline{next\_nonloose\_parallel}; 
    \codecomment{reference to the nonloose parallel edge that follows the current edge in right-to-left order around $s$; not \NULL only if the current edge is parallel and nonloose}\\
    \Edge{} \underline{next\_nonloose\_lens}; 
    \codecomment{reference to the nonloose edge belonging to a multilens that follows the current edge in right-to-left order around $s$; not \NULL only if the current edge belongs to a multilens}
  }
  \label{apx:data-structures}
\end{algorithm}

\subparagraph{Data structures.} We start by describing the data structures exploited by the algorithm. Vertices and edges of $G$ are modelled by means of the following records (see also \cref{apx:data-structures}).
\begin{description}
\item[Record of type {\bf Vertex}:] For each vertex $x$, the following information is stored: 
\begin{itemize}
\item an integer $\deg_G(x)$;
\item whether $x$ is incident to the outer face or not;
\item a reference to the rightmost edge between $s$ and $x$, if any; 
\item if $x=s$, we also store the following information:
\begin{itemize}
\item a reference to the rightmost edge $e_1$ incident to $x$;
\item a reference to the rightmost chord incident to $x$, if any; 
\item a reference to the rightmost parallel edge incident to $x$, if any; 
and
\item a reference to the rightmost edge belonging to a multilens, if any.
\end{itemize}  
\end{itemize}

\item[Record of type {\bf Edge}:] For each edge $e$, the following information is stored: 
\begin{itemize}
\item a reference to the two end-vertices $x$ and $y$ of $e$;
\item whether $e$ is oriented from $x$ to $y$, or vice versa (this is initialized arbitrarily); 
\item whether $e$ is an outer edge or not; 
\item a reference to the edge incident to $x$ that follows $e$ in counter-clockwise order around $x$ and a reference to the edge incident to $y$ that follows $e$ in counter-clockwise order around $y$ (this information represents the rotation system around $x$ and $y$);
\item if $e$ is incident to $s$, we also store the following information:
\begin{itemize}
\item an integer representing the position of $e$ in left-to-right order around $s$ (we assume that the leftmost edge $(s,t)$ has position $1$);
\item if $e$ is parallel, a reference to the parallel edge $(x,y)$ that follows $e$ in right-to-left order around $s$, if any;
\item if $e$ is a chord, a reference to the chord that follows $e$ in right-to-left order around $s$, if any;
\item if $e$ is parallel and nonloose, a reference to the parallel nonloose edge that follows $e$ in right-to-left order around $s$, if any; and
\item if $e$ belongs to a multilens, a reference to the nonloose edge belonging to a multilens that follows $e$ in right-to-left order around $s$, if any.
\end{itemize}  
\end{itemize}
\end{description}

\noindent
The algorithm exploits the following (global) data structures:

\begin{itemize} 
    \item The input well-formed biconnected plane graph $G = (V,E)$ with poles $s$ and $t$ is represented using records of type {\bf Vertex} for the vertices in $V$ and records of type {\bf Edge} for the edges in $E$. In particular, we maintain the reference to a record $S$ of type {\bf Vertex} corresponding to the vertex $s$.
    \item In order to efficiently output the orientation of all the edges of an inner-canonical orientation, we use an array {\sc EDGES} whose elements are records of type \Edge whose $i$-th entry contains a reference to the edge with id equal to $i$. 
\end{itemize}

\subparagraph{Procedures.} The algorithm builds upon the procedures described below. 
In the remainder, we denote the type-\Vertex record for a vertex $x$ by $\nu(x)$ and the type-\Edge record for an edge $e$ by $\varepsilon(e)$. Clearly, $S = \nu(s)$.
\smallskip


\begin{algorithm}[h!]
\DontPrintSemicolon
\SetKwInOut{Input}{Input}
\SetKwInOut{Output}{Output}
\SetKwInOut{Global}{Global}
\Global{\Vertex $S$; \codecomment{the pole $s$ of a well-formed graph $G$}\\ 
\Edge{[}\ {]} EDGES;  \codecomment{the array of the edges of $G$}}
\Output{A label in \{\texttt{BASE, CONTRACT, REMOVE, CONTRACT\&REMOVE}\} encoding the case of the {\sc ICE} algorithm that applies to $G$}
\BlankLine
$e_1$ $\leftarrow$ $S$.$e_1$ \codecomment{The edge $e_1$ of $G$; for simplicity of description, assume $e_1.x$ is $S$.}\\
\If{$e_1$.next\_around\_x is NULL}{
    \Return \texttt{BASE}; \codecomment{Base Case}
}
\If{$S$.first\_parallel is NULL}{
    \Return \texttt{CONTRACT}; \codecomment{Case 3}
}
\If{$e_1$.next\_parallel\_with\_me is not NULL}{
    \Return \texttt{REMOVE}; \codecomment{Case 4}
}

\If{$S$.first\_chord.ord $\geq$ $S$.first\_lens.ord}{
    \Return \texttt{CONTRACT}; \codecomment{Case 1}
}
\Else{
    \Return \texttt{CONTRACT\&REMOVE}; \codecomment{Case 2}
}

\BlankLine
\caption{\textbf{Procedure {\sc DetectCase}}}
\label{alg:DetectCase}\vspace{-0.2cm}
\end{algorithm}

\noindent{\sc Detect Case.}
This procedure allows us to efficiently determine which of the cases of the {\sc ICE} algorithm applies to $G$; refer to the pseudocode of \cref{alg:DetectCase}.
%
We perform checks in the following order, assuming that previous checks have not concluded which case of the {\sc ICE} algorithm we are in.
\begin{itemize}
\item We access $e_1$ via $S$ and check if the reference in $\varepsilon(e_1)$ to the edge incident to $s$ that follows $e_1$ in counter-clockwise order around $s$ is \NULL. In the positive case, we are in the {\em Base Case}. 
\item We check if the reference in $S$ to the rightmost parallel edge incident to $s$ is \NULL. In the positive case, we are in {\em Case 3}. 
\item We access $e_1$ via $S$ and check if the reference in $\varepsilon(e_1)$ to the parallel edge with the same end-vertices as $e_1$ that follows $e_1$ in right-to-left order around $s$ is \NULL. In the negative case, we are in {\em Case 4}. 
\item Finally, we access via $S$ the type-\Edge record $\varepsilon(c)$ referenced by $S$ for the rightmost chord $c$ incident to $s$ and the type-\Edge record $\varepsilon(r)$ referenced by $S$ for the rightmost edge $r$ belonging to a multilens. If $\varepsilon(c)$ is \NULL, or if the integer in $\varepsilon(c)$ representing the position of $c$ in the left-to-right order around $s$ is smaller than the integer in $\varepsilon(r)$ representing the position of $r$ in the left-to-right order around $s$, we are in {\em Case 2}, otherwise we are in {\em Case 1}.
\end{itemize}
Clearly, we can perform all these checks in $\bigoh(1)$ time.
\smallskip

\noindent{\sc Output.}
This procedure allows us to efficiently output an inner-canonical orientation of $G$. With this aim, it suffices to scan the array EDGES, printing for each edge its orientation. 
Clearly, this takes $\bigoh(\varphi)$ time.
\smallskip
\begin{figure}[htb]
    \centering
    \includegraphics[page=13,scale =.7]{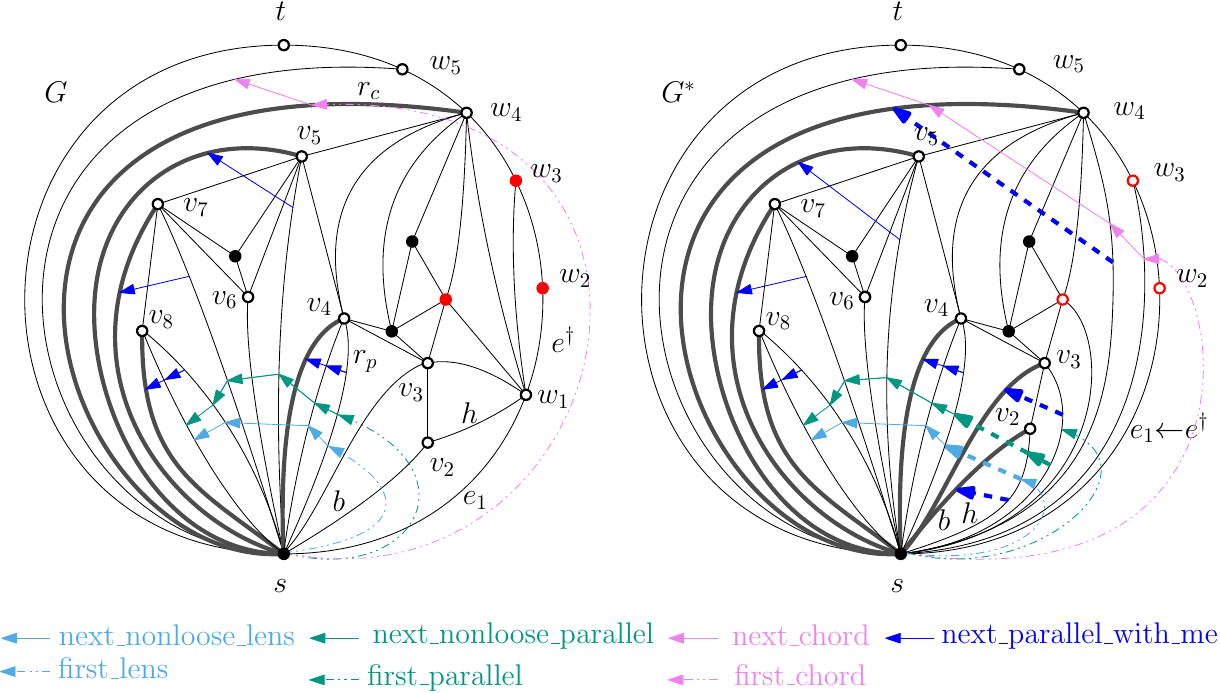}
\caption{\label{fig:contract-procedure}
Illustration for the {\sc Contract} procedure. The input graph $G$ is on the left, whereas the graph $G^*$ resulting from the contraction of $e_1$ in $G$ is on the right. Empty and filled circles represent vertices adjacent and not adjacent to $s$, respectively.
The neighbors of $w_1$ that are not incident to $S$ in $G$ are red. The arrowed curves illustrate the references from type-\Edge records to type-\Edge records that correspond to the edges incident to $s$ (solid or dashed) as well as the references from the type-\Vertex record $S$ to type-\Edge records that correspond to the edges incident to $s$ (dashed-dotted). To reduce clutter, the references \texttt{next\_around\_x} and \texttt{next\_around\_y} (see \cref{algo:datastructures})
used to encode the counter-clockwise order of the edges incident to $s$ are omitted.
The references on the right that are not present on the left are dashed and thick. Loose edges are solid and thick.
}
\end{figure}

\smallskip
\noindent{\sc Contract.}
This procedure allows us to efficiently perform the contraction of the edge $e_1$ in~$G$ in order to construct the graph $G^*$ and to update the data structures in such a way as to support the recursive calls of the {\sc ICE} algorithm; refer also
to the pseudocode description of \cref{alg:Contract} and to \cref{fig:contract-procedure}. 
The procedure works as follows.

First, we access the reference to the rightmost edge $e_1$ incident to $s$ via $S$ and set the Boolean value in $\varepsilon(e_1)$ representing the orientation of such an edge to \texttt{True} if the end-vertex $x$ of $e_1$  is $s$ and to \texttt{False} otherwise. Also, we set the integer value of $S$ representing the degree $\deg_{G^*}(s)$ of $s$ in $G^*$ to $\deg_G(s) + \deg_G(w_1) - 2$. Indeed, the degree of the vertex resulting from the contraction of $e_1$ is the sum of the degrees of the end-vertices $s$ and $w_1$ of $e_1$, minus two, as $e_1$ is incident to both $s$ and $w_1$ in $G$ and is not part of $G^*$.
Second, the reference to the rightmost edge incident to $s$ is set to point to the successor $(w_1,w_2)$ of $e_1$ in the counter-clockwise order of the edges incident to $w_1$. 
Third, we store for future use a reference to the rightmost parallel edge $r_p$ incident to $s$ in $G$ and a reference to the rightmost chord $r_c$ incident to $s$ in $G$; this information can be accessed via $S$.

Next, we perform a visit of the edges incident to $w_1$ in the counter-clockwise order in which they appear around $w_1$ in $G$ starting at $(w_1,w_2)$. Throughout the visit, we keep track of the lastly visited edge $l_c$ incident to $w_1$ that becomes a chord in~$G^*$ and of the lastly visited edge $l_p$ incident to $w_1$ that results in a nonloose parallel edge in $G^*$. We execute the following actions for each encountered edge $e$. 
\begin{description}
    \item[{Current edge's updates:}] We set the reference to the end-vertex of $e$ corresponding to $w_1$ to point to $S$ and the integer in $\varepsilon(e)$ representing the position of $e$ in left-to-right order around $s$ to be $\deg_{G^*}(s)-i+1$, if $e$ is the $i$-th edge considered in the visit. 
    \item[{Handling new chords:}] 
    We test whether the end-vertex $u$ of $e$ different from $w_1$ (in fact, different from $s$, after the previous update) is incident to the outer face (this information is stored in $\nu(u)$). If that is the case and if $e\neq (w_1,w_2)$, we have encountered a chord of $G^*$. If this is the first encountered edge incident to $w_1$ that is a chord in $G^*$, then it is also the rightmost chord of $G^*$, hence we set the reference in $S$ to the rightmost chord incident to $s$ to point to $\varepsilon(e)$. Otherwise, we have already encountered an edge incident to $w_1$ that is a chord in $G^*$, and the last encountered edge of this type is stored in $l_c$, hence we set the reference in $\varepsilon(l_c)$ to the chord that follows $l_c$ in the right-to-left order around $s$ to $\varepsilon(e)$. In either case, we update $l_c$ to $e$.
    \item[{Handling new lenses:}] We test whether the edge that follows $e$ in the counter-clockwise order of the edges incident to $w_1$ in $G$ is $e_1$. If this is the case, then the edge $e$ is the edge $(v_2,w_1)$, labeled $h$ in \cref{fig:contract-procedure}, and becomes the rightmost edge of a multilens of $G^*$ composed of parallel edges between $s$ and $v_2$.
    Therefore, we update the reference in $S$ to the rightmost nonloose edge of a multilens incident to $s$ to $\varepsilon(e)$. Notice that $e$ is the only nonloose edge that is involved in a multilens in $G^*$ and not in $G$. Also, if the reference in $S$ to the rightmost nonloose edge of a multilens incident to $s$ used to point to an edge $r_\ell$, then we update the reference in $\varepsilon(e)$ to the nonloose edge belonging to a multilens that follows $e$ in right-to-left order around $s$ to point to $\varepsilon(r_\ell)$.
    \item[{Handling new parallel edges:}] 
    We test if the end-vertex $x$ of $e$ different from $w_1$ is already adjacent to $s$, i.e., if $x$ is one of the vertices $v_2,\dots, v_m$; this information is stored in $\nu(x)$ as the reference to the rightmost edge $e_j$ between $x$ and $s$, for some $j \in \{2,\dots,m\}$. 
     If this is the case, then edges between $s$ and $x$ exist in $G$, thus the contraction of $e_1$ turns $e$ into an edge parallel to such edges; also, $e$ is a nonloose edge, as it is to the right of the edges that already exist between $s$ and $x$ in $G$. We set the reference in $\varepsilon(e)$ to the edge parallel to $e$ that follows $e$ in right-to-left order around $s$ to point to $\varepsilon(e_j)$.
    If $e$ is the first encountered edge incident to $w_1$ that is a parallel edge in $G^*$, then it is also the rightmost parallel edge of $G^*$, hence we set the reference in $S$ to the rightmost parallel edge incident to $s$ to point to $\varepsilon(e)$. Otherwise, we have already encountered an edge incident to $w_1$ that is a nonloose parallel edge in $G^*$, and the last encountered edge of this type is stored in $l_p$, hence we set the reference in $\varepsilon(l_p)$ to the nonloose parallel edge that follows $l_p$ in right-to-left order around $s$ to $\varepsilon(e)$. In either case, we update $l_p$ to $e$.
    
    Both if edges between $s$ and $x$ exist in $G$ and if they do not, the edge $e$ is the rightmost edge incident to $s$ and $x$ in $G^*$, hence we set the reference in $\nu(x)$ to the rightmost edge incident to $x$ and $s$ to point to $\varepsilon(e)$. 
\end{description}

\begin{algorithm}[h!]
\DontPrintSemicolon
\SetKwInOut{Input}{Input}
\SetKwInOut{Output}{Output}
\SetKwInOut{Global}{Global}
\SetKwInOut{SideEffects}{Side Effects}
\Global{\Vertex $S$; \codecomment{The pole $s$ of a well-formed graph $G$.}\\ 
\Edge{[}\ {]} EDGES;  \codecomment{The array of the edges of $G$.}}
\Output{A reference to the edge $e_1$ of $G$.}
\SideEffects{Sets the orientation of $e_1$ and contracts $e_1$ in $G$.}
\BlankLine
$e_1$ $\leftarrow$ $S$.$e_1$; \codecomment{The edge $e_1$ of $G$.}\\
\codecomment{\underline{\bf Orient $e_1$:}}\\
\If{$\bf e_1.x$ is $\bf S$}
{$e_1$.oriented\_from\_x\_to\_y $\leftarrow\; \texttt{True}$;}
\Else{$e_1$.oriented\_from\_x\_to\_y $\leftarrow\; \texttt{False}$;} 
\codecomment{In what follows, for simplicity, assume $e_1.x=s$ and $e_1.y=w_1$.}\\
$S.e_1 \leftarrow$ $e_1$.next\_around\_y; \codecomment{Update $S$'s $e_1$ to be $(w_1,w_2)$.}\\
$S$.degree $\leftarrow$ $S$.degree + $e_1$.y.degree $-\ 2$; \codecomment{Update $S$'s degree.}\\
\codecomment{\underline{\bf Counter-clockwise visit of the edges incident to $w_1$:}}\\
$l_c$, $l_p$  $\leftarrow \NULL$; $r_p$ $\leftarrow$ $S$.first\_parallel; $r_c$ $\leftarrow$ $S$.first\_chord; \codecomment{Auxiliary references used in the visit.}\\
$e \leftarrow S.e_1$;\\
\For{$\bf i \leftarrow 1$; $\bf i < e_1.y.degree$; $\bf i \leftarrow i+1$}{
    \codecomment{\underline{\bf Current edge $e$ udpates} (in what follows, for simplicity, assume $e.y=w_1$):}\\
    $e.y \leftarrow S$ ; \codecomment{Identify $s$ and $w_1$.}\\
    $e$.ord $\leftarrow$ $S$.degree $-\ i + 1$\\
    \codecomment{\underline{\bf Handling new chords:}}\\
    \If{$\bf e.x$.is\_outer is \texttt{True} and $\bf e$.isOuter is \texttt{False}}{
        \If{$\bf l_c$ is \NULL}{
            $S$.first\_chord  $\leftarrow e$; \codecomment{Update $S$.first\_chord.}
        }\Else{
            $l_c$.next\_chord $\leftarrow e$;\\
        }
        $l_c$ $\leftarrow e$;
    }
    \codecomment{\underline{\bf Handling new lenses:}}\\

    \If{$\bf e$.next\_around\_y is $\bf e_1$}{
        $r_\ell \leftarrow $ $S$.first\_lens;\\
        $S$.first\_lens $\leftarrow e$;\\
        $e$.next\_nonloose\_lens $\leftarrow$ $r_\ell$;\\
    }

    \codecomment{\underline{\bf Handling new parallel edges:}}\\
    \If{$\bf e.x$.first\_incident\_to\_s is not \NULL}{
        $e$.next\_parallel\_with\_me $\leftarrow e.x$.first\_incident\_to\_s;\\

    \If{$\bf l_p$ is \NULL}{
            $S$.first\_parallel  $\leftarrow e$; \codecomment{Update $S$.first\_parallel.}
        }\Else{
        $l_p$.next\_nonloose\_parallel $\leftarrow e$;\\
        }
        $l_p \leftarrow e$;
    }
    $e.x$.first\_incident\_to\_s $\leftarrow e$; \\
    \If{$\bf i < e_1.y.degree-1$}{
    $e \leftarrow$ $e$.next\_around\_y; \codecomment{The next edge to consider in the visit.}}
}
\codecomment{At this point, $e$ is the predecessor $h$ of $e_1$ in counter-clockwise order around $w_1$ in $G$.}\\
\If{$\bf l_c$ is not \NULL}{
$l_c$.next\_chord $\leftarrow r_c$;}
\If{$\bf l_p$ is not \NULL}{
$l_p$.next\_nonloose\_parallel $\leftarrow r_p$;}
e.next\_around\_y $\leftarrow$ $e_1$.next\_around\_x;\\ 
\Return{$e_1$};
\BlankLine
\caption{\textbf{Procedure {\sc Contract}}}
\label{alg:Contract}\vspace{-0.2cm}
\end{algorithm}


Finally, when all the edges incident to $w_1$ have been visited, three more actions are performed.

First, suppose that new chords have been introduced by the contraction of $e_1$. Recall that $l_c$ is the leftmost among such chords. We update the reference in $\varepsilon(l_c)$ to the next chord in right-to-left order around $s$ to point to $\varepsilon(r_c)$, which is the type-\Edge record corresponding to the rightmost chord incident to $s$ in $G$. This was stored before visiting the edges incident to $w_1$. In this way, we link together all the chords incident to $s$.

Second, suppose that new nonloose parallel edges have been introduced by the contraction of $e_1$. Recall that $l_p$ is the leftmost among such edges. We update the reference in $\varepsilon(l_p)$ to the next nonloose parallel edge in right-to-left order around $s$ to point to $\varepsilon(r_p)$, which is the type-\Edge record corresponding to rightmost parallel edge incident to $s$ in $G$. This was stored before visiting the edges incident to $w_1$. In this way, we link together all non-loose parallel edges incident to $s$.

Third, let $h$ be the predecessor of $e_1$ in counter-clockwise order around $w_1$ in $G$ and let $b$ be the successor of $e_1$ in $G$ in counter-clockwise order around $s$; refer to \cref{fig:contract-procedure}. We set the reference in $\varepsilon(h)$ to the edge incident to $s$ that follows $h$ in counter-clockwise order around $s$ in $G^*$ to point to $\varepsilon(b)$. In this way, we restore the rotation system around $s$. This concludes the description of the {\sc Contract} procedure. 
Since we visit all the edges incident to $w_1$ once and perform for each of them only checks and updates that take $\bigoh{(1)}$ time, the overall procedure runs in $\bigoh{(\deg_G(w_1))}$ time.

\smallskip
\noindent{\sc Decontract.} This procedure allows us to efficiently perform the ``decontraction'' of the edge $e_1$ in $G^*$, in order to obtain the graph $G$ and the record $S$ back. The corresponding data structures need to be updated accordingly. We omit the description of the steps of such procedure, as they can be easily deduced from the ones of the {\sc Contract} procedure\remove{, and include its pseudocode in \cref{alg:decontract}}. In particular, the edges incident to $s$ in $G^*$ that are incident to $w_1$ in $G$ consist of the rightmost edge belonging to a multilens of $G^*$ (whose reference is stored in $S$) and of all the edges that precede such an edge in the right-to-left order around $s$.
Analogously as for the {\sc Contract} procedure, the {\sc Decontract} procedure can be implemented to run in $\bigoh{(\deg_G(w_1))}$ time. A pseudocode description of the procedure is provided in \cref{alg:decontract}.

	\begin{algorithm}[h!]
	\caption{\textbf{Procedure {\sc Decontract}}}\label{alg:decontract}
	\SetKwInOut{Input}{Input}
	\SetKwInOut{Output}{Output}
	\SetKwInOut{SideEffects}{Side Effects}
	\SetKwInOut{Global}{Global}
	\Global{\Vertex $S$; \codecomment{The pole $s$ of a well-formed graph $G$.}\\ 
		\Edge{[}\ {]} EDGES;  \codecomment{The array of the edges of $G$.}}
	\Input{An edge $e_1$}
	\SideEffects{Sets the orientation of $e_1$ to \NULL and decontracts $e_1$ in $G^*$ in order to obtain $G$}
	\BlankLine
	$e_1$.oriented\_from\_x\_to\_y $\leftarrow$ \NULL;\\
	\codecomment{$S$.first\_lens, that is $h$, is not part of a multilens in $G$.}\\
	$h \leftarrow S$.first\_lens;\\
	$b \leftarrow$ $h$.next\_parallel\_with\_me;\\
	leftmost\_chord, leftmost\_parallel;\\
	\codecomment{Visit the edges incident to $s$ in $G^*$ that are not incident to $s$ in $G$.}\\
	$e\leftarrow S.e_1$;\\
	\While{$e\neq b$}{
		\codecomment{We assume, for simplicity, $e.x=s$, $e_1.x=s$, and $e_1.y=w_1$.}\\
		$e.x \leftarrow e_1.y$;\\
		$e.y$.first\_incident\_to\_s $\leftarrow$ $e$.next\_parallel\_with\_me; \codecomment{$e$ is not incident to $s$ in $G$, edges parallel to $e$ in $G^*$ are.}\\
		\codecomment{Find the leftmost chord and the leftmost parallel edge that are not incident to $s$ in $G$.}\\
		\If{$e$.is\_outer is \texttt{False} and $e.y$.is\_outer is \texttt{True}}{
		leftmost\_chord $\leftarrow$ $e$;}
		\If{$e$.next\_parallel\_with\_me is not \NULL}{
		leftmost\_parallel $\leftarrow e$;}
		\codecomment{Current edge $e$ is not incident to $s$ in $G$.}\\
		$e$.ord $\leftarrow$ \NULL;\\
		$e$.next\_chord $\leftarrow$ \NULL;\\
		$e$.next\_parallel\_with\_me $\leftarrow$ \NULL;\\
		$e$.next\_nonloose\_parallel $\leftarrow$ \NULL;\\
		$e$.next\_nonloose\_lens $\leftarrow$ \NULL;\\
		$S$.degree $\leftarrow$ $S$.degree$-1$;\\
		$e\leftarrow e.$next\_around\_x;\\
	}
	$h$.next\_around\_x $\leftarrow$ $e_1$;\\
	$S$.first\_lens $\leftarrow h$.next\_nonloose\_lens;\\
	$S$.first\_chord $\leftarrow$ leftmost\_chord.next\_chord;\\
	$S$.first\_parallel $\leftarrow$ leftmost\_parallel.next\_nonloose\_parallel;\\
	$S.e_1 \leftarrow e_1$;\\
\end{algorithm}

\smallskip
\noindent{\sc Remove.} This procedure allows us to efficiently perform the removal of the edges $e_1,\dots,e_j$  (as defined in \cref{le:remove-well-formed}) in $G$ in order to construct the graph $G^-$ and to update the data structures in such a way as to support the recursive calls of the {\sc ICE} algorithm; refer 
to the pseudocode description of \cref{alg:Remove} and to \cref{fig:remove-procedure}.
The procedure works as described next and returns an array \texttt{REMOVED}, whose entry with index $i$ points to the type-\Edge record $\varepsilon(e_{i+1})$ of the removed edge $e_{i+1}$, for $i=0,\dots,j-1$. 

First, pointers to the type-\Edge records corresponding to $e_j$ and $e_{j+1}$ are retrieved. The first one is indeed stored in the record $S$ as the reference to the rightmost edge belonging to a multilens. The second one is instead stored in the record $\varepsilon(e_j)$ as the reference to the edge incident to $s$ that follows $e_j$ in counter-clockwise order around $s$.
Second, the following updates are performed on $S$.
\begin{enumerate}[{\bf (i)}] 
\item The degree of $s$ is updated in $S$ to be the integer stored in $\varepsilon(e_{j+1})$ representing the position of $e_{j+1}$ in the left-to-right order around $s$.
\item The reference in $S$ to the rightmost edge belonging to a multilens is updated to the reference stored in $\varepsilon(e_j)$ to the nonloose edge belonging to a multilens that follows $e_j$ in counter-clockwise order around $s$. Observe that, if multilenses exist in $G^-$, then such an edge is $e_{j+1}$ if $e_j$ belongs to a multilens consisting of more than two parallel edges (see \cref{fig:remove-multilens}), whereas it is different from $e_{j+1}$ if $e_j$ belongs to a multilens consisting of two parallel edges (see \cref{fig:remove-lens}).
\item The reference in $S$ to the rightmost parallel edge incident to $s$ is updated to the reference stored in $\varepsilon(e_j)$ to the nonloose parallel edge that follows $e_j$ in counter-clockwise order around $s$.
If parallel edges exist in $G^-$, then an analogous observation to the one given for (ii) holds also in this case; refer again to \cref{fig:remove-multilens,fig:remove-lens}.
\item The reference in $S$ to the rightmost edge incident to $s$ is updated to point to $\varepsilon(e_{j+1})$.
\end{enumerate}

Third, we perform a counter-clockwise visit of the edges incident to $s$ to be removed, starting from $e_1$ and ending at $e_j$ (both extremes are considered in the visit), and execute the following actions for each encountered edge $e_i$. Let $x$ and $y$ be the end-vertices of $e_i$, where $x=s$. For future use, we store the reference $r_c$ to the rightmost chord incident to $s$ in $G$, if any. 

%
%
\begin{figure}[!]
\centering
    \begin{subfigure}{\textwidth}
    \centering
        \includegraphics[page=14,scale =.7]{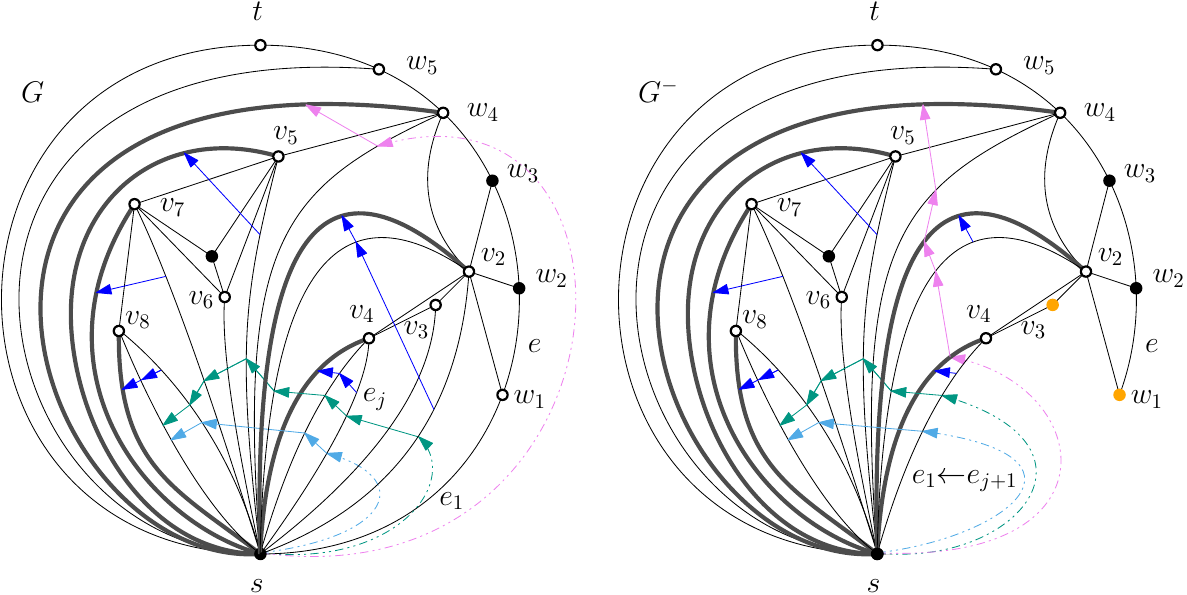}
        \caption{\label{fig:remove-multilens}$e_j$ belongs to a multilens  consisting of more than two parallel edges.}
    \end{subfigure}
    \begin{subfigure}{\textwidth}
    \centering
        \includegraphics[page=15,scale =.7]{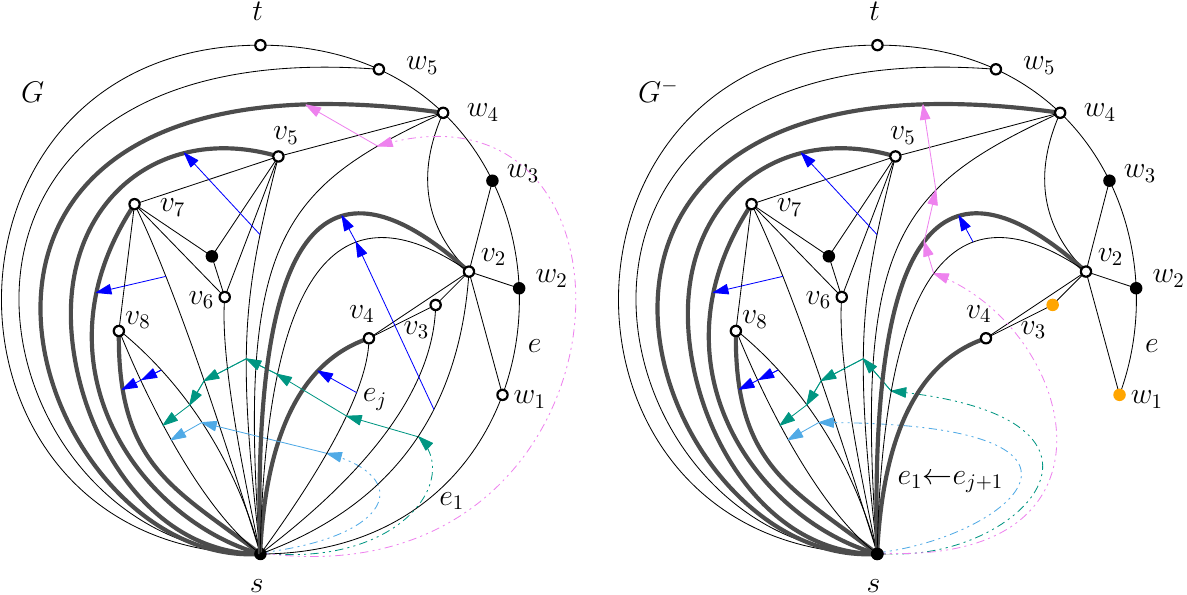}
        \caption{\label{fig:remove-lens}$e_j$ belongs to a multiles consisting of two parallel edges.}
    \end{subfigure}
\caption{\label{fig:remove-procedure}
Illustrations for the {\sc Remove} procedure. The input graph $G$ is on the left, whereas the graph $G^-$ resulting from the removal of $e_1,\dots,e_j$ in $G$ is on the right. 
The meaning of arrowed curves, of the thickness of the edges, and of the shape of the vertices is as in \cref{fig:contract-procedure}.
The vertices $v_1=w_1$ and $v_3$ that are neighbors of $s$ in $G$ but not in $G^-$ are orange.}
\end{figure}
%
\begin{description}
\item[Current edge's updates:]  We set the Boolean value in $\varepsilon(e_i)$ representing the orientation of such an edge to \texttt{True} (i.e., $e_i$ is oriented from $x=s$ to $y=v_i$), the Boolean value in $\nu(v_i)$ representing the fact that this vertex is incident to the outer face of $G^-$ to \texttt{True}, and the integer value in $\nu(v_i)$ representing the degree $\deg_{G^-}(v_i)$ of $v_i$ in $G^-$ to $\deg_G(v_i) - 1$. Furthermore, we add $e_i$ to the array \texttt{REMOVED} in the position with index $i-1$. 


\item[Updating new outer edges and outer vertices:] 
For $h=1,\dots,j-1$, denote by $e_{h,h+1}$ the edge $(v_h,v_{h+1})$ and by $e_{0,1}$ the edge $(v_1 = w_1,w_2)$ of $G$.
We consider two cases based on whether $i \neq j$ or $i=j$.

If $i \neq j$, then $e_i$, $e_{i+1}$, and $e_{i,i+1}$ bound a triangular face of $G$, and $(v_i,v_{i+1})$ is an internal edge of $G$. The edges $e_i$ and $e_{i+1}$ do not belong to $G^-$, while the edge $e_{i,i+1}$ does and is an outer edge of $G^-$. We, therefore, set the Boolean value in $\varepsilon(e_{i,i+1})$ that represents the fact that $e_{i,i+1}$ is incident to the outer face of $G^-$ to \texttt{True}. Also, we set the reference in $\varepsilon(e_{i,i+1})$ to the edge that follows it in counter-clockwise order around $v_i$~to~point~to~$\varepsilon(e_{i-1,i})$. Note that $\varepsilon(e_{i+1})$ and $\varepsilon(e_{i-1,i})$ can be accessed via $\varepsilon(e_i)$ as the record of the edge that follows $e_i$ in counter-clockwise order around $s$ and around $v_i$, respectively, and $\varepsilon(e_{i,i+1})$ can be accessed via $\varepsilon(e_{i+1})$ as the record of the edge that follows $e_{i+1}$ in counter-clockwise order around $v_{i+1}$.

If $i=j$, then $e_j$ and $e_{j+1}$ bound an internal face of $G$, and in particular the edge $e_{j+1}$ is an internal edge of $G$. The edge $e_j$ does not belong to $G^-$, while the edge $e_{j+1}$ does and is an outer edge of $G^-$. We, therefore, set the Boolean value in $\varepsilon(e_{j+1})$ that represents the fact that $e_{j+1}$ is incident to the outer face of $G^-$ to \texttt{True}. Also, we set the reference in $\varepsilon(e_{j+1})$ to the edge that follows it in counter-clockwise order around $v_i$~to~point~to~$\varepsilon(e_{j-1,j})$. Note that $\varepsilon(e_{j+1})$ can be accessed via $\varepsilon(e_i)$ as above, and $\varepsilon(e_{j-1,j})$ can be accessed via $\varepsilon(e_{j})$ as the record of the edge that follows $e_{j}$ in counter-clockwise order around $v_{j}$.

\item[Updating chords:] Note that the applicability of the removal operations requires that there exists no chord $(s,v_h)$ in $G$ with $h\leq j$. Therefore, each chord of $G$ is also a chord of $G^-$. However, there may exist chords in $G^-$ that do not belong to $G$. These chords, if any, are exactly the edges parallel to the removed edges $e_1,\dots,e_j$; refer to \cref{fig:remove-procedure}. If neither of $e_1,\dots,e_{j-1}$ has parallel edges and the unique edge parallel to $e_j$ is $e_{j+1}$, then $G^-$ and $G$ have the same chords. To account for the chords in $G^-$ that do not belong to $G$, we take the steps described below for each visited edge $e_i$.  

For any $1\leq a<b\leq j$, we detect all the chords of $G^-$ stemming from edges parallel to $e_a$ before the chords stemming from $e_b$. All the chords stemming from $e_a$ need to be linked together. Also, if $e_a$ and $e_b$ each have at least one parallel edge in $G$ and there exists no edge $e_c$ with $a<c<b$ such that $e_c$ has at least one parallel edge in $G$, then the leftmost chord of $G^-$ stemming from $e_b$ needs to be linked to the rightmost chord of $G^-$ stemming from $e_a$. Moreover, let $a'$ and $b'$ be the minimum and maximum indices of edges in $e_1,\dots,e_j$ having parallel edges. Then the leftmost chord of $G^-$ stemming from $e_{a'}$ needs to be linked to the rightmost chord $r_c$ of $G$, whereas the rightmost chord stemming from $e_{b'}$ needs to be set at the rightmost chord incident to $s$. 

Below, we provide the details of the actions needed to implement the above updates in the data structures, when processing an edge $e_i$. Throughout, we maintain the record $R$ of the rightmost encountered chord; initially, we set $R=r_c$. 

\begin{itemize}
	\item First, we access the reference to the rightmost edge $r^i_p$ of $G$, if any, different from $e_i$ that is parallel to $e_i$ (i.e., this information is stored in $\varepsilon(e_i)$ as the reference to the edge parallel to $e_i$ that follows $e_i$ in right-to-left order around $s$ in $G$). 
	\item Second, we set the reference in $\nu(v_i)$ to the rightmost edge incident to $v_i$ that is also incident to $s$ to point to $\varepsilon(r^i_p)$. If $\varepsilon(r^i_p)$ is \NULL, we proceed to consider the next edge $e_{i+1}$, otherwise we continue to process $e_i$ as below.
	\item Third, we perform a counter-clockwise visit around $s$ of the list of edges parallel to $e_i$ starting from $r^i_p$ and ending at the leftmost edge $l^i_p$ of such a list. Let $\alpha$ be the currently considered edge in this visit;  initially, $\alpha = r^i_p$.  We set the reference in $\varepsilon(\alpha)$ to the chord that follows $\alpha$ in right-to-left order to point to the record of the edge $\beta$ parallel to $\alpha$ that follows $\alpha$ in right-to-left order, if any. This links together all the chords stemming from $e_i$. When all the edges parallel to $e_i$ have been visited, we set the reference in $\varepsilon(l^i_p)$ to the chord that follows $l^i_p$ in right-to-left order to point to $R$, and then we update $R=r^i_p$. This links the chords stemming from $e_i$ to the chords stemming from $e_1,\dots,e_{i-1}$ and to the chords in $G$.
\end{itemize}
When all the edges $e_1,\dots,e_j$ have been visited, we update the reference in $S$ to the rightmost chord incident to $S$ to point to $R$. Note that, if none of $e_1,\dots,e_{j-1}$ has parallel edges and the unique edge parallel to $e_j$ is $e_{j+1}$, then $G^-$ and $G$ have the same chords and $R = r_c$. 
\end{description}



As a final step, the procedure returns the array \texttt{REMOVED}.

Since we visit all the edges $e_1,\dots,e_j$ once and all the edges parallel to $e_2,\dots,e_j$ once, and since we perform for each of these edges only checks and updates that take $\bigoh{(1)}$ time, the overall procedure runs in $\bigoh({j+\sum^j_{i=2}\pi(e_i))}$, where $\pi(e_i)$ denotes the number of edges parallel to $e_i$ in $G$.

\begin{algorithm}[!]{
\DontPrintSemicolon
\SetKwInOut{Input}{Input}
\SetKwInOut{Output}{Output}
\SetKwInOut{Global}{Global}
\SetKwInOut{SideEffects}{Side Effects}
\Global{\Vertex $S$; \codecomment{The pole $s$ of a well-formed graph $G$.}\\ 
\Edge{[}\ {]} EDGES;  \codecomment{The array of the edges of $G$.}}
\Output{A length-$j$ array REMOVED whose entries point to the type-\Edge records corresponding to the edges $e_1,\dots,e_j$ of $G$ removed as in \cref{le:remove-well-formed}.}
\SideEffects{Sets the orientation of $e_1,\dots,e_j$ and removes these edges from $G$.}
\BlankLine
\codecomment{For simplicity of description, assume $e.x$ points to $S$ for any edge $e$ incident to $s$.}\\
$e_1$ $\leftarrow$ $S$.$e_1$; \\
$e_j$ $\leftarrow$ $S$.first\_lens;\\
$e_{j+1}$ $\leftarrow$ $e_j$.next\_around\_x;\\


\codecomment{\underline{\bf{Updating S:}}}\\

$S$.degree $\leftarrow$ $e_{j+1}$.ord; \codecomment{setting the degree of $S$}\\
$S$.first\_lens $\leftarrow$ $e_j$.next\_nonloose\_lens; \codecomment{setting the pointer to the rightmost  edge of a multilens}\\
$S$.first\_parallel $\leftarrow$ $e_j$.next\_nonloose\_parallel; \codecomment{setting the pointer to the rightmost parallel edge}\\
$S.e_1$ $\leftarrow$ $e_{j+1}$; \codecomment{setting the pointer to the rightmost edge incident to $s$}\\
\codecomment{\underline{\bf{Right-to-left visit of the edges incident to $s$ to be removed:}}}\\
$e \leftarrow e_1$; 
$r_c$ $\leftarrow$ $S$.first\_chord; $R\leftarrow r_c$;\\
\While{$e$ is not $e_{j+1}$}{
    \codecomment{\bf \underline{Setting the orientation of $e = (s,v_i)$ and updating the fields of $\nu(v_i)$.}}\\
    $e$.oriented\_from\_x\_to\_y $\leftarrow\; \texttt{True}$; \codecomment{Orient  $e$}\\
    $e.y$.isOuter $\leftarrow$ \texttt{True};\\
    $e.y$.degree $\leftarrow$ $e.y$.degree $- 1$;\\
    REMOVED[$e_1$.ord-$e$.ord] $\leftarrow$ e; \codecomment{Add $e$ to the array of removed edges}\\
\If{$e$ is not $e_j$}{
	\codecomment{\bf \underline{Updating the incidence list of $v_i$ and making $(v_i,v_{i+1})$ an outer edge.}}\\
	$e_{i,i+1} \leftarrow$ $e$.next\_around\_x.next\_around\_y;\\
	$e_{i,i+1}$.isOuter $\leftarrow$ \texttt{True};\\
	$e_{i,i+1}$.next\_around\_y $\leftarrow$ $e$.next\_around\_y; \codecomment{$e_{i-1,i}$ follows $e_{i,i+1}$ in counterclockwise order around $v_i$}\\
}
\Else{
	\codecomment{\underline{\bf Updating the incidence list of $v_j$ and making $e_{j+1}$ an outer edge.}}\\
	$e_{j+1}$.isOuter $\leftarrow$ \texttt{True};\\
	$e_{j+1}$.next\_around\_y $\leftarrow$  $e_{j}$.next\_around\_y; 
}

\codecomment{\underline{\bf{Updating chords:}}}\\

$r^i_p \leftarrow$ $e.$next\_parallel\_with\_me;\\
$e.y$.first\_incident\_to\_s $\leftarrow r^i_p$;\\
\If{$e=e_j$}{
$r^i_p \leftarrow r^i_p$.next\_parallel\_with\_me; \codecomment{ignore $e_{j+1}$, which is not a chord}\\
}

\If{$e$ is not $e_1$ and $r^i_p$ is not \NULL} { 
\codecomment{there are new chords $(s,v_i)$}\\
$\alpha \leftarrow r^i_p$;\\
$\beta \leftarrow \alpha$.next\_parallel\_with\_me;\\
\While{$\beta$ is not \NULL}{
	$\alpha$.next\_chord $\leftarrow \beta$;\\
	$\alpha \leftarrow \beta$;\\
	$\beta \leftarrow \alpha$.next\_parallel\_with\_me;
}
$\alpha$.next\_chord $\leftarrow R$;\\
$R\leftarrow r^i_p$;\\
}

\If{$e=e_j$}{
	S.first\_chord $\leftarrow R$; \codecomment{Update the first chord}\\
}
$e \leftarrow e$.next\_around\_x;\\
}

\Return REMOVED;
}

\BlankLine
\caption{\textbf{Procedure {\sc Remove}}
}
\label{alg:Remove}\vspace{-0.2cm}
\end{algorithm}



\smallskip

\noindent{\sc Reinsert.} This procedure allows us to efficiently perform the ``reinsertion'' of the edges $e_1,\dots,e_j$ in~$G^-$ to reconstruct the graph $G$, and to accordingly update the data structures. For space reasons, we omit the description of the steps of such procedure as they can be easily deduced from the ones of the {\sc Remove} procedure. Analogously as for the {\sc Remove} procedure, the {\sc Reinsert} procedure can be implemented to run in $\bigoh({j+\sum^j_{i=2}\pi(e_i))}$. A pseudocode description of the procedure is provided in \cref{alg:reinsert}.

\begin{algorithm}[h!]
	\caption{\textbf{Procedure {\sc Reinsert }}
	}\label{alg:reinsert}
	\SetKwInOut{Input}{Input}
	\SetKwInOut{Output}{Output}
	\SetKwInOut{Global}{Global}
	\SetKwInOut{SideEffects}{Side Effects}
	\Global{\Vertex $S$; \codecomment{The pole $s$ of a well-formed graph $G$.}\\ 
	\Edge{[}\ {]} EDGES;  \codecomment{The array of the edges of $G$.}}	
	\Input{A length-$j$ array \texttt{REMOVED} whose entries point to the type-\Edge records corresponding to the edges $e_1, \dots, e_j$ removed as in \cref{le:remove-well-formed}.}
	\SideEffects{Reinserts the edges $e_1, \dots, e_j$ in $G^-$ to obtain $G$.}
	\BlankLine
	$S$.$e_1$ $\leftarrow$ \texttt{REMOVED}[0]; \codecomment{this is $e_1$}\\
	$S$.first\_lens $\leftarrow$ \texttt{REMOVED}[\texttt{REMOVED}.length-1]; \codecomment{this is $e_j$}\\
	new\_first\_chord $\leftarrow$ $S$.first\_chord;\\
    new\_first\_parallel $\leftarrow$ $S$.first\_parallel;\\
	\For{\Int{} $i$=\texttt{REMOVED}.length; $i\geq 1$; $i--$}{
		$e_i$ $\leftarrow$ \texttt{REMOVED}[i-1];\\
		\codecomment{We assume, for simplicity, $e_i.x=S$, $e_i.y=v_i$}\\
		\If{$e_i$ is not $e_1$}{
			$v_i$.isOuter $\leftarrow$ \texttt{False};
		}
		$v_i$.first\_incident\_to\_s $\leftarrow$ $e_i$;\\
		$v_i$.degree $\leftarrow$ $v_i$.degree$+1$\\
		\If{$e_i$ is not $e_j$}{
			\codecomment{\bf \underline{Updating the incidence list of $v_i$, and making the edge $(v_i,v_{i+1})$ an internal edge.}}\\
			\codecomment{We assume, for simplicity, $e_{i,i+1}=(v_i, v_{i+1})$, $e_{i,i+1}.x=v_i$, and $e_{i,i+1}.y=v_{i+1}$}\\
    		$e_{i+1}$ $\leftarrow$ \texttt{REMOVED}[i];\\
            $e_{i,i+1}$ $\leftarrow$ $e_{i+1}$.next\_around\_y;\\
			$e_{i,i+1}$.next\_around\_x $\leftarrow$ $e_i$;\\
			$e_{i,i+1}$.isOuter $\leftarrow$ \texttt{False};\\
		}
		\Else{
			\codecomment{\underline{\bf Updating the incidence list of $v_j$ and making $e_{j+1}$ an internal edge.}}\\
			$e_{j+1}$ $\leftarrow$ $e_{j}$.next\_around\_x;\\
			$e_{j+1}$.next\_around\_y $\leftarrow$ $e_{j}$;\\
            \If{ $e_{j+1}$.next\_around\_x is not \NULL}{$e_{j+1}$.isOuter $\leftarrow$ \texttt{False};\\}
			
		}
		\codecomment{\underline{\bf{Updating chords:}}}\\
		\If{$e_i$ is not $e_1$ and $e_i$.next\_parallel\_with\_me is not \NULL}{
			\codecomment{All the edges parallel with $e_i$ are not chords in $G$.}\\
			$e$ $\leftarrow$ $e_i$.next\_parallel\_with\_me;\\
			\While{$e$.next\_parallel\_with\_me is not \NULL}{
				\codecomment{Both $e$.next\_parallel\_with\_me and  $e$.next\_chord reference the next edge $(s,v_i)$.}\\
				$e$.next\_chord $\leftarrow$ \NULL;\\
				$e \leftarrow e$.next\_parallel\_with\_me;\\
			}
		new\_first\_chord $\leftarrow e$.next\_chord;\\
		$e$.next\_chord $\leftarrow$ \NULL;}
        \codecomment{\underline{\bf{Updating first\_parallel:}}}\\
        \If{$e_i$.next\_parallel\_with\_me is not \NULL}{
            new\_first\_parallel $\leftarrow e_i$;\\}
	}
    \codecomment{\underline{\bf Updating $S$'s first chord and first parallel}}\\
	$S$.first\_chord $\leftarrow$ new\_first\_chord;\\
	$S$.first\_parallel $\leftarrow$ new\_first\_parallel;\\
        $S$.degree $\leftarrow$ $S$.degree+ \texttt{REMOVED}.length
\end{algorithm}

\medskip

We are finally ready to prove the bounds stated in \cref{th:innercanonical-orientation-plane-uv}. 
\begin{description}
    \item[Setup time:] Recall that $\varphi$ denotes the number of edges of $G$. Initializing the type-\Vertex and the type-\Edge records requires $\bigoh(\varphi)$ time, assuming that: (i) for each vertex of $G$, a circularly-linked list is provided encoding the counter-clockwise order of the edges incident to $v$ in the planar embedding of $G$, and that (ii) the edge $(s,v_1)$ incident to the outer face of $G$ is specified. Indeed, setting the type-\Vertex and the type-\Edge records up can be easily accomplished by suitably traversing the above lists. In particular, a first visit starting from the edge $(s,v_1)$ allows us to determine the outer vertices and edges of $G$, from which the chords of $G$ can be determined. Parallel edges can be detected easily since they are incident to $s$; indeed, while traversing the list of the edges incident to $s$, one can mark each end-vertex different from $s$ the first time an edge incident to it is encountered, and also keep track of a reference to that edge. Edges incident to an already marked vertex and to $s$ are parallel edges and also make the first edge incident to those two vertices a parallel edge. The array \texttt{EDGES} can clearly be constructed in $\bigoh(\varphi)$ time.


    \item[Space usage:] At any step, the graph considered by the {\sc ICE} algorithm has at most $\varphi$ edges. Thus, the space used to represent such a structure is $\bigoh(\varphi)$. Also, the number of recursive calls to the {\sc ICE} algorithm is $\bigoh(\varphi)$. Therefore, in order to show that the overall space usage of the {\sc ICE} algorithm is also $\bigoh(\varphi)$, we only need to account for the amount of information that needs to be stored, at any moment, in the call stack, i.e., for the size of the activation records of all the calls. The top of the stack either contains the activation record of a call to the {\sc ICE} algorithm on the current graph or of a call to the {\sc Output}, {\sc Contract},  {\sc Decontract},  {\sc Remove},  or {\sc Reinsert} auxiliary procedures. The interior of the stack only contains the activation records of calls to the {\sc ICE} algorithm. Whereas the activation records for each of the five auxiliary procedures are of $\bigoh(1)$ size, the size of the activation record of a call to the {\sc ICE} algorithm is $\bigoh(1)$, if the considered call does not invoke the {\sc Remove} procedure, or is $\bigoh(j)$, if the considered call invokes the {\sc Remove} procedure in order to remove the edges $e_1,\dots,e_j$. In particular, the activation record of each call to the {\sc ICE} algorithm contains either a reference to the contracted edge $e_1$, if the considered call invokes the {\sc Contract} procedure, or references to the removed edges  $e_1,\dots,e_j$, if it invokes the {\sc Remove} procedure. The key observation here is that a reference to an edge of $G$ may appear only once over all the activation records that are simultaneously on the stack during the execution of the {\sc ICE} algorithm, as a contracted or removed edge is not part of the graph considered in the recursive calls. Therefore, the overall space usage of the stack is $\bigoh(\varphi)$. 
    \item[Delay:] We show that the time spent by the {\sc ICE} algorithm to output the first inner-canonical orientation of $G$ is $\bigoh(\varphi)$, and that the time between any two inner-canonical orientations of $G$ that are consecutively listed by the {\sc ICE} algorithm is also $\bigoh(\varphi)$.
    The recursive calls to the {\sc ICE} algorithm determine a rooted binary tree $\cal T$, which we refer to as the \emph{call tree}, defined as follows; see \cref{fig:computation}. 
    The root $\rho$ of $\cal T$ corresponds to the first call on the input graph $G$, each non-root node of $\cal T$ corresponds to the call on a graph obtained starting from $G$ by applying a sequence of {\sc Contract} or {\sc Remove} procedures.
    Let $\nu$ be the parent node of a node $\mu$ of $\cal T$, and let $G_\nu$ and $G_\mu$ be the graphs associated with $\nu$ and $\mu$, respectively.
    The edge $(\nu,\mu)$ either
    corresponds to a {\sc Contract} (and the symmetric {\sc Decontract}) procedure if $G_\mu = G^*_\nu$ or corresponds to a {\sc Remove}  (and the symmetric {\sc Reinsert}) procedure if $G_\mu = G^-_\nu$.

    By \cref{le:at-least-one-inner-canonical}, we have that the leaves of $\cal T$ correspond to calls to the {\sc ICE} algorithm on the single edge $(s,t)$, which is the base case of the {\sc ICE} algorithm that results in a call to the {\sc Output} procedure. We consider the leaves of $\cal T$ as ordered according to their order of creation in the construction of $\cal T$. Therefore, we can refer to the \emph{first leaf} of $\cal T$ and, given a leaf $\lambda$ of $\cal T$, to the leaf of $\cal T$ that \emph{follows} $\lambda$.
    For each edge $e$ of $\cal T$, the \emph{cost} of $e$, denoted by $c(e)$, is the time spent to perform the procedure corresponding to $e$.
    
    The time spent to output the first inner-canonical orientation of $G$ coincides with the sum of the costs of all the edges of the root-to-leaf path $p_\alpha$ in $\cal T$ that connects $\rho$ and the first leaf $\alpha$ of $\cal T$, i.e., $\sum_{e \in p_\alpha} c(e)$, plus the time spent by the {\sc Output} procedure. As the latter is $\bigoh(\varphi)$, we only need to show that the former is also $\bigoh(\varphi)$.
    As already shown in the description of the procedures, for any edge $e=(\nu,\mu)$ of $\cal T$, the cost $c(e)$ is at most $k_1 \cdot \deg_{G_{\nu}}(w_1)$, if $e$ corresponds to a {\sc Contract/Decontract} procedure, and is at most $k_2 \cdot (j + \sum^j_{i=2}\pi(e_i))$, if $e$ corresponds to a {\sc Remove/Reinsert} procedure that removes/reinsert the $j$ rightmost edges incident to the source of $G_\nu$, for suitable constants $k_1,k_2>0$. 
    We have that $\sum_{e \in p_\alpha} c(e) \leq \max(k_1,k_2) \cdot 3 \varphi$. In fact, each edge of $G$: (1) may contribute at most twice to the degree of a vertex $w_1$ that has an incident edge involved in a {\sc Contract/Decontract} procedure; (2) may appear at most once as one of the $j$ edges removed/reinserted by a {\sc Remove/Reinsert} procedure; an (3) might appear at most once as one of the edges parallel to some edge removed/reinserted by a {\sc Remove/Reinsert} procedure. Indeed: (1) an edge $e$ that is incident to $w_1$ in $G$ is incident to $s$ in $G^*$, hence $e$ might only be incident to a ``new'' vertex $w_1$ involved in a second {\sc Contract/Decontract} procedure if it is itself the edge to be contracted; after that, the edge is not part of the resulting graph $G^*$; (2) an edge that appears as one of the $j$ edges removed from $G$ is not part of the resulting graph $G^-$; and (3) if an edge $e$ is parallel to some removed edge in $e_2,\dots,e_j$, then $e$ is a chord in $G^-$, hence it is not a parallel edge of a removed edge later, as that would imply that a removed edge is also a chord, which does not happen in a {\sc Remove/Reinsert} procedure.    
  
    Finally, let $\lambda$ be a leaf of $\cal T$, let $\eta$ be the leaf of $\cal T$ that follows $\lambda$, and let $\xi$ be the lowest common ancestor of $\lambda$ and $\eta$. The time between the output of the inner-canonical orientation of $G$ corresponding to $\lambda$ and the output of the inner-canonical orientation of $G$ corresponding to $\eta$ coincides with the sum of the costs of all the edges of the path $q_\lambda$ in $\cal T$ between $\lambda$ and $\xi$ and of the costs of all the edges of the path $q_\eta$ in $\cal T$ between $\xi$ and $\eta$, i.e., $\sum_{e \in q_\lambda} c(e) + \sum_{e \in q_\eta} c(e)$, plus the time spent by the {\sc Output} procedure. As the latter is $\bigoh(\varphi)$, we only need to show that the former is also $\bigoh(\varphi)$. By the same arguments used above, we have that $\sum_{e \in q_\lambda} c(e) + \sum_{e \in q_\eta} c(e) \leq \max(k_1,k_2) \cdot 6 \varphi$, which is $O(\varphi)$. 
\end{description}
This concludes the proof of \cref{th:innercanonical-orientation-plane-uv}.

\section{Enumeration of Canonical Orderings and Canonical Drawings} \label{se:canonical-drawings}

In this section, we show how the enumeration algorithm for canonical orientations from \cref{se:canonical-orientation} can be used in order to provide efficient algorithms for the enumeration of canonical orderings and canonical drawings. We start with the former.

\begin{lemma} \label{th:canonical-ordering-plane-uv}
Let $G$ be an $n$-vertex maximal plane graph and let $(u,v,z)$ be the cycle delimiting its outer face. There exists an algorithm with $\bigoh(n)$ setup time and $\bigoh(n)$ space usage that lists all canonical orderings of $G$ with first vertex $u$ with $\bigoh(n)$ delay.
\end{lemma}

\begin{proof}
The algorithm uses the algorithm in the proof of \cref{th:canonical-orientation-plane-uv} for the enumeration of the canonical orientations of $G$ with first vertex $u$. Indeed, for every canonical orientation $\mathcal D$ of $G$ listed by the latter algorithm, all canonical orderings $\pi$ of $G$ such that $\mathcal D$ is the canonical orientation of $G$ with respect to $\pi$ can be generated as the topological sortings of $\mathcal D$, by \cref{le:from-orientation-to-ordering}. Algorithms exist for listing all such topological sortings with $\bigoh(n)$ setup time, $\bigoh(n)$ space usage, and even just $\bigoh(1)$ delay, given that $G$ has $\bigoh(n)$ edges; see~\cite{on-ctg-05,pr-glef-94}. 

All generated canonical orderings have $u$ as first vertex since all the canonical orientations produced by the algorithms in~\cite{on-ctg-05,pr-glef-94} have $u$ as first vertex. Furthermore, any two canonical orderings generated from the same canonical orientation $\mathcal D$ differ as any two topological sortings listed by the algorithms in~\cite{on-ctg-05,pr-glef-94} are different from one another. Moreover, any two canonical orderings $\pi$ and $\pi'$ generated from different canonical orientations $\mathcal D$ and $\mathcal D'$, respectively, differ as there exists an edge $(w,w')$ in $G$ which is oriented from $w$ to $w'$ in $\mathcal D$ and from $w'$ to $w$ in $\mathcal D'$; this implies that $w$ precedes $w'$ in $\pi$ and follows $w'$ in $\pi'$.
\end{proof}

\begin{theorem} \label{th:canonical-ordering-plane}
Let $G$ be an $n$-vertex maximal plane (planar) graph. There exists an algorithm $\mathcal B_1$ (resp.\ $\mathcal B_2$) with $\bigoh(n)$ setup time and $\bigoh(n)$ space usage that lists all canonical orderings of $G$ with $\bigoh(n)$ delay.
\end{theorem}

\begin{proof}
Algorithm $\mathcal B_1$ uses the algorithm  in the proof of \cref{th:canonical-ordering-plane-uv} applied three times, namely once for each choice of the first vertex among the three vertices incident to the outer face of the given maximal plane graph. Algorithm $\mathcal B_2$ uses the algorithm $\mathcal B_1$ applied $4n-8$ times, since there are $4n-8$ maximal plane graphs which are isomorphic to a given maximal planar graph (see the proof of \cref{th:canonical-orientation-plane-and-planar}). Note that any two canonical orderings produced by different applications of algorithm $\mathcal B_2$ differ on the first, or on the second, or on the last vertex in the ordering. 
\end{proof}

We now turn our attention to the enumeration of the planar straight-line drawings produced by the algorithm by de Fraysseix, Pach, and Pollack~\cite{dpp-hdpgg-90}, known as \emph{canonical drawings}. We begin by reviewing such an algorithm, which in the following is called \emph{FPP algorithm}. The algorithm takes as input:
\begin{itemize}
\item an $n$-vertex maximal plane graph $G$, whose outer face is delimited by a cycle $(u,v,z)$, where $u$, $v$, and $z$ appear in this counter-clockwise order along the outer face of $G$; and
\item a canonical ordering $\pi=(v_1=u,v_2=v,v_3,\dots,v_n=z)$ of $G$; recall that $G_k$ denotes the subgraph of $G$ induced by the first $k$ vertices of $\pi$. 
\end{itemize}

The FPP algorithm works in $n-2$ steps. At the first step, the FPP algorithm constructs a planar straight-line drawing $\Gamma_3$ of $G_3$ so that $v_1$ is placed at $(0,0)$, $v_2$ at $(2,0)$, and $v_3$ at $(1,1)$, and defines the sets $M_3(v_1)=\{v_1,v_2,v_3\}$, $M_3(v_3)=\{v_2,v_3\}$, and $M_3(v_2)=\{v_2\}$.

\begin{figure}[t]
    \centering
    \begin{subfigure}{.49\textwidth}
    \centering
    	\includegraphics[page=1, scale =.75]{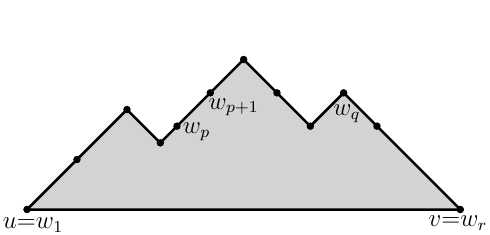}
        \caption{The drawing $\Gamma_k$ of $G_k$.}
    	\label{fig:canonical-drawing-a}
    \end{subfigure}
    \hfill
    \begin{subfigure}{.49\textwidth}
    \centering
    	\includegraphics[page=2, scale =.75]{img/FPP.pdf}
    	\caption{The drawing $\Gamma_{k+1}$ of $G_{k+1}$.}
    	\label{fig:canonical-drawing-b}
    \end{subfigure}
\caption{Illustrations for the FPP algorithm.}
\label{fig:canonical-drawing}
\end{figure}
For any $k=3,\dots,n-1$, at step $k-1$, the FPP algorithm constructs a planar straight-line drawing $\Gamma_{k+1}$ of $G_{k+1}$ by suitably modifying $\Gamma_{k}$, as follows; refer to \cref{fig:canonical-drawing}. Denote by $w_1=u,w_2,\dots,w_{r}=v$ the clockwise order of the vertices of $G_k$ along its outer face. Assume that, during step $k-2$, the algorithm has defined, for $i=1,\dots,r$, a subset $M_k(w_i)$ of the vertices of $G_k$ such that $M_k(w_1)\supset M_k(w_2)\supset \dots \supset M_k(w_{r})$. Let $w_p,w_{p+1},\dots,w_q$ be the neighbors of $v_{k+1}$ in $G_k$, for some $1\leq p <q\leq r$. Then $\Gamma_{k+1}$ is obtained from $\Gamma_k$ by increasing the $x$-coordinate of each vertex in $M_k(w_{p+1})$ by one unit, by increasing the $x$-coordinate of each vertex in $M_k(w_q)$ by one additional unit, and by placing $v_{k+1}$ at the intersection point of the line through $w_p$ with slope $+1$ and of the line through $w_q$ with slope $-1$. The key point for the proof of planarity of $\Gamma_{k+1}$ is that the vertices $w_1,\dots,w_{r}$ define in $\Gamma_k$ an $x$-monotone path whose edges have slope either $+1$ or $-1$. The ``shift'' of the vertices in $M_k(w_{p+1})$ makes room for drawing the edge $(w_p,v_{k+1})$ with slope $+1$, and the shift of the vertices in $M_k(w_q)$ makes room for drawing the edge $(w_q,v_{k+1})$ with slope $-1$, thus maintaining the invariant on the shape of the boundary of $\Gamma_{k+1}$. Step $k-1$ is completed by defining the sets:
\begin{itemize}
    \item $M_{k+1}(w_i)=M_k(w_i)\cup \{v_{k+1}\}$, for $i=1,\dots,p$;
    \item $M_{k+1}(v_{k+1})=M_k(w_{p+1})\cup \{v_{k+1}\}$; and
    \item $M_{k+1}(w_i)=M_{k}(w_i)$, for $i=q,\dots,r$.
\end{itemize}
We call \emph{canonical drawing with base edge $(u,v)$} the drawing $\Gamma_n$ of $G$ constructed by the FPP algorithm. We now prove the following main ingredient of our enumeration algorithm for canonical drawings.

\begin{theorem} \label{th:bijection-FPP}
Let $G$ be an $n$-vertex maximal plane graph and let $(u,v,z)$ be the cycle delimiting the outer face of $G$, where $u$, $v$, and $z$ appear in this counter-clockwise order along the outer face of $G$.
There exists a bijective function from the canonical orientations of $G$ with first vertex $u$ to the canonical drawings of $G$ with base edge $(u,v)$. Also, given a canonical orientation of $G$ with first vertex $u$, the corresponding canonical drawing of $G$ with base edge $(u,v)$ can be constructed in $\bigoh(n)$ time.
\end{theorem}

In the following we prove \cref{th:bijection-FPP}. We do this outside of a proof environment as the proof contains some statements which might be of independent interest. 

We introduce some definitions. If $\mathcal D$ is the canonical orientation of $G$ with respect to $\pi$, we say that $\pi$ \emph{extends} $\mathcal D$ and that $\pi$ \emph{defines} $\mathcal D$, depending on whether $\pi$ is constructed from $\mathcal D$ or vice versa. Also, we say that the canonical drawing $\Gamma$ of $G$ obtained by applying the FPP algorithm with a canonical ordering $\pi$ of $G$ \emph{corresponds to} $\pi$.

The bijective function $f$ that proves the statement of the theorem is defined as follows. Consider any canonical orientation $\mathcal D$ of $G$ with first vertex $u$ and let $\pi$ be a canonical ordering with first vertex $u$ that extends $\mathcal D$. Then the function $f$ maps $\mathcal D$ to the canonical drawing that corresponds to $\pi$. Note that this canonical drawing has $(u,v)$ as the base edge, since $u$ is the first vertex of $\mathcal D$ (and hence the first vertex of $\pi$) and since $u$, $v$, and $z$ appear in this counter-clockwise order along the boundary of the outer face of $G$. The second part of the statement of \cref{th:bijection-FPP} then follows from the fact that a canonical ordering that extends $\mathcal D$ can be computed in $\bigoh(n)$ time as any topological sorting of $\mathcal D$~\cite{DBLP:journals/cacm/Kahn62}, and that the FPP algorithm can be implemented in $\bigoh(n)$ time~\cite{chrobak1995linear}.

In order to prove that $f$ is bijective, we prove that it is injective (that is, for any two distinct canonical orientations $\mathcal D_1$ and $\mathcal D_2$ of $G$ with first vertex $u$, we have that  $f(\mathcal D_1)$ and $f(\mathcal D_2)$ are not the same drawing) and that it is surjective (that is, for any canonical drawing $\Gamma$ with base edge $(u,v)$, there exists a canonical orientation $\mathcal D$ such that $f(\mathcal D)$ is $\Gamma$).

We first prove that $f$ is injective. Let $\mathcal D_1$ and $\mathcal D_2$ be any two distinct canonical orientations  of $G$ with first vertex $u$. For $i=1,2$, let $\pi_i$ be any canonical ordering that extends~$\mathcal D_i$. Since $\mathcal D_1$ and $\mathcal D_2$ are distinct, they differ on the orientation of some edge $(a,b)$ different from $(u,v)$, say that $(a,b)$ is directed towards $b$ in $\mathcal D_1$ and towards $a$ in $\mathcal D_2$. This implies that $b$ follows $a$ in $\pi_1$ and precedes $a$ in $\pi_2$. Hence, the $y$-coordinate of $b$ is larger than the one of $a$ in $f(\mathcal D_1)$ and smaller than the one of $a$ in $f(\mathcal D_2)$, thus $f(\mathcal D_1)$ and $f(\mathcal D_2)$ are not the same drawing.

We now prove that $f$ is surjective. Consider any canonical drawing $\Gamma$ of $G$ with base edge $(u,v)$ and let $\pi$ be a canonical ordering of $G$ with first vertex $u$ such that the canonical drawing corresponding to $\pi$ is $\Gamma$. Let $\mathcal D$ be the canonical orientation of $G$ with respect to $\pi$. The existence of the canonical orientation $\mathcal D$ is not enough to prove that $f$ is surjective. Indeed, given the canonical orientation $\mathcal D$, the function $f$ considers {\em some} canonical ordering $\tau$, possibly different from $\pi$, that extends $\mathcal D$, hence $f(\mathcal D)$ is the canonical drawing corresponding to $\tau$ and it is not guaranteed that $f(\mathcal D)=\Gamma$. However, we have the following claim.

\begin{claimx} \label{cl:orientation-determines-drawing}
Any two canonical orderings $\pi$ and $\tau$ that extend $\mathcal D$ are such that the canonical drawings of $G$ corresponding to $\pi$ and $\tau$ are the same drawing.    
\end{claimx}

\cref{cl:orientation-determines-drawing} implies that $f$ is surjective. Indeed, the function $f$ considers some canonical ordering $\tau$, possibly different from $\pi$, that extends $\mathcal D$; by \cref{cl:orientation-determines-drawing}, the drawing corresponding to $\tau$ is the same drawing as the one corresponding to $\pi$, that is, $\Gamma$.  It remains to prove \cref{cl:orientation-determines-drawing}, which we do next. We start by extending the notions of canonical ordering, orientation, and drawing to biconnected plane graphs that are not necessarily maximal. 

Let $H$ be an $m$-vertex biconnected plane graph, with $m\geq 3$, whose internal faces are delimited by cycles with $3$ vertices, and let $u$ and $v$ be two vertices incident to the outer face of $H$ such that $u$ immediately precedes $v$ in the counter-clockwise order of the vertices along the boundary of the outer face of $H$. A \emph{canonical ordering of $H$ with first vertex $u$} is a labeling of the vertices $(v_1=u, v_2=v, v_3,  \dots, v_{m-1}, v_m)$ such that, for $k=3,\dots,m$, the plane subgraph $H_{k}$ of $H$ induced by $v_1,v_2,\dots,v_k$ satisfies conditions (CO-1) and (CO-2) in \cref{se:preliminaries} (with $H_k$ and $H$ replacing $G_k$ and $G$, respectively). Then a \emph{canonical orientation $\mathcal D_H$ of $H$ with first vertex $u$} is obtained from a canonical ordering $\pi$ of $H$ with first vertex $u$ by orienting each edge of $H$ so that it is outgoing at the end-vertex that comes first in $\pi$. We say that $\pi$ \emph{extends} and \emph{defines} $\mathcal D_H$. A \emph{canonical drawing $\Gamma$ of $H$ with base edge $(u,v)$} is a drawing obtained by applying the FPP algorithm with a canonical ordering $\pi$ of $H$. We say that $\Gamma$ \emph{corresponds to} $\pi$.


We now state the following claim, which is more general than, and hence implies, \cref{cl:orientation-determines-drawing}.

\begin{claimx} \label{cl:biconnected-orientation-determines-drawing}
Any two canonical orderings $\pi$ and $\tau$ that extend $\mathcal D_H$ are such that the canonical drawings of $H$ corresponding to $\pi$ and $\tau$ are the same drawing.    
\end{claimx}

The proof of \cref{cl:biconnected-orientation-determines-drawing} is by induction on $m$. The proof uses \cref{pr:inductive-sets} below, which is proved inductively together with \cref{cl:biconnected-orientation-determines-drawing}. Let $z_1,z_2,\dots,z_{r}$ be the clockwise order of the vertices along the outer face of $H$. Furthermore, let $M^{\pi}_m(z_1)$, $M^{\pi}_m(z_2)$, $\dots$, $M^{\pi}_m(z_{r})$ (let $M^{\tau}_m(z_1)$, $M^{\tau}_m(z_2)$, $\dots$, $M^{\tau}_m(z_{r})$) be the sets that are associated to the vertices $z_1,z_2,\dots,z_{r}$, respectively, by the FPP algorithm, when applied with canonical ordering $\pi$ (resp.\ $\tau$). 

\begin{property} \label{pr:inductive-sets}
For $i=1,2,\dots,r$, the sets $M^{\pi}_m(z_i)$ and $M^{\tau}_m(z_i)$ coincide.   
\end{property}

Let $\pi=(u_1=u,u_2=v,\dots,u_m)$ and $\tau=(v_1=u,v_2=v,\dots,v_m)$. Also, let $\Gamma$ and $\Phi$ be the canonical drawings of $H$ that correspond to $\pi$ and $\tau$, respectively. In the base case we have $m=3$. Then the statements of \cref{cl:biconnected-orientation-determines-drawing} and \cref{pr:inductive-sets} are trivial. Indeed, the first two vertices in any canonical ordering of $H$ with first vertex $u$ are respectively $u$ and $v$, hence there is a unique canonical ordering that extends $\mathcal D_H$, there is a unique canonical drawing with base edge $(u,v)$, and we have $M^{\pi}_3(z_1)=M^{\tau}_3(z_1)=\{u_1,u_2,u_3\}$, $M^{\pi}_3(z_2)=M^{\tau}_3(z_2)=\{u_2,u_3\}$, and $M^{\pi}_3(z_3)=M^{\tau}_3(z_3)=\{u_2\}$.

Suppose now that $m>3$. Assume that the statements of \cref{cl:biconnected-orientation-determines-drawing} and \cref{pr:inductive-sets} hold if $H$ has $m-1$ vertices. We prove that they also hold if $H$ has $m$ vertices. Let $\ell$ be the index such that $v_{\ell}$ is the same vertex as $u_m$. We observe the following simple facts.

\begin{property} \label{pr:last-elsewhere}
We have $\ell\in \{4,\dots,m\}$.   
\end{property}

\begin{proof}
The first three vertices are the same in any canonical drawing with first vertex $u$, hence $v_1$, $v_2$, and $v_3$ respectively coincide with $u_1$, $u_2$, and $u_3$ and are different from $u_m$.     
\end{proof}

\begin{property} \label{pr:neighbors-before}
All the neighbors of $v_{\ell}$ in $H$ precede $v_{\ell}$ in $\tau$.  
\end{property}
\begin{proof}
Since $u_m$ is the last vertex of $\pi$, all the neighbors of $u_m$ in $H$ precede $u_m$ in $\pi$, hence they also precede $v_{\ell}=u_m$ in $\tau$, as $\pi$ and $\tau$ define the same canonical orientation of $H$.    
\end{proof}

The proof now distinguishes two cases, depending on whether $\ell=m$ or not.

{\bf Case 1: $u_m$ and $v_m$ are the same vertex.} Let $L$ be the $(m-1)$ vertex plane graph obtained from $H$ by removing the vertex $u_m$ and its incident edges. Let $w_1,w_2,\dots,w_{s}$ be the clockwise order of the vertices along the outer face of $H$ and let $w_p,w_{p+1},\dots,w_q$ be the neighbors of $u_m$ in $H$, for some $1\leq p<q\leq s$. Also, let $\lambda$ and $\xi$ be the vertex orderings of $L$ obtained from $\pi$ and $\tau$, respectively, by removing the vertex $u_m$. Finally, let $M^{\lambda}_{m-1}(w_1), M^{\lambda}_{m-1}(w_2), \dots,M^{\lambda}_{m-1}(w_s)$ (let $M^{\xi}_{m-1}(w_1)$, $M^{\xi}_{m-1}(w_2)$, $\dots$, $M^{\xi}_{m-1}(w_s)$) be the sets that are associated to the vertices $w_1,w_2,\dots,w_{s}$ by the FPP algorithm, when applied with canonical ordering $\lambda$ (resp.\ $\xi$). 
We next prove the following.

\begin{lemma} \label{le:one-less-vertex-canonical}
$\lambda$ and $\xi$ are canonical orderings of $L$; furthermore, $\lambda$ and $\xi$ define the same canonical orientation of $L$.    
\end{lemma}

\begin{proof}
First, $L$ is biconnected, as it coincides with the subgraph $H_{m-1}$ of $H$ induced by the first $m-1$ vertices of $\pi$, which is biconnected by Condition (CO-1) of $\pi$. We show that $\lambda$ is a canonical ordering of $L$; the proof that $\xi$ is also a canonical ordering of $L$ is analogous. The first and second vertex of $\lambda$ are $u$ and $v$, since the same is true for $\pi$ and since $m>3$. Also, for $k=3,\dots,m-2$, we have that conditions (CO-1) and (CO-2) hold for the subgraph $L_k$ of $L$ induced by the first $k$ vertices in $\lambda$, given that this coincides with the subgraph $H_k$ of $H$ induced by the first $k$ vertices in $\pi$ and given that $\pi$ is a canonical ordering of $H$. Finally, $\lambda$ and $\xi$ define the same canonical orientation of $L$, since  the orientations of $L$ defined by $\lambda$ and $\xi$ both coincide with the orientation obtained from $\mathcal D_H$ by removing $u_m$ and its incident edges.
\end{proof}

By \cref{le:one-less-vertex-canonical}, we have that $\lambda$ and $\xi$ are canonical orderings of $L$. Let $\Lambda$ and $\Xi$ be the canonical drawings of $L$ corresponding to $\lambda$ and $\xi$, respectively. By induction, $\Lambda$ and $\Xi$ are the same drawing and, for $i=1,2,\dots,s$, we have $M^{\lambda}_{m-1}(w_i)=M^{\xi}_{m-1}(w_i)$. 

The FPP algorithm constructs $\Gamma$ from $\Lambda$ by shifting each vertex in $M^{\lambda}_{m-1}(w_{p+1})$ by one unit to the right, by shifting each vertex in $M^{\lambda}_{m-1}(w_q)$ by one additional unit to the right, and by placing $u_m$ at the intersection point of the line through $w_p$ with slope $+1$, and of the line through $w_q$ with slope $-1$. Also, the FPP algorithm constructs $\Phi$ from $\Xi$ by shifting each vertex in $M^{\xi}_{m-1}(w_{p+1})$ by one unit to the right, by shifting each vertex in $M^{\xi}_{m-1}(w_q)$ by one additional unit to the right, and by placing $u_m$ at the intersection point of the line through $w_p$ with slope $+1$, and of the line through $w_q$ with slope $-1$. Since $\Xi$ coincides with $\Lambda$, since $M^{\xi}_{m-1}(w_{p+1})=M^{\lambda}_{m-1}(w_{p+1})$, and since $M^{\xi}_{m-1}(w_q)=M^{\lambda}_{m-1}(w_q)$, we have that $\Phi$ and $\Gamma$ are the same drawing, as required.

By construction, the FPP algorithm defines: {\bf(i)} $M^{\pi}_{m}(w_i)=M^{\lambda}_{m-1}(w_i)\cup \{u_m\}$, for $i=1,\dots,p$; {\bf(ii)} $M^{\pi}_{m}(u_m)=M^{\lambda}_{m-1}(w_{p+1})\cup \{u_m\}$; {\bf(iii)} $M^{\pi}_{m}(w_i)=M^{\lambda}_{m-1}(w_i)$, for $i=q,\dots,s$; {\bf(iv)} $M^{\tau}_m(w_i)=M^{\xi}_{m-1}(w_i)\cup \{u_m\}$, for $i=1,\dots,p$; {\bf(v)} $M^{\tau}_m(u_m)=M^{\xi}_{m-1}(w_{p+1})\cup \{u_m\}$; and {\bf(vi)} $M^{\tau}_m(w_i)=M^{\xi}_{m-1}(w_i)$, for $i=q,\dots,s$. Since $M^{\xi}_{m-1}(w_i)=M^{\lambda}_{m-1}(w_i)$, for $i=1,\dots,s$, it follows that $M^{\pi}_m(z_i)=M^{\tau}_m(z_i)$, for $i=1,\dots,r$, as required. 

{\bf Case 2: $u_m$ and $v_m$ are not the same vertex.} In this case, by \cref{pr:last-elsewhere}, we have that $u_m$ coincides with $v_{\ell}$, for some $\ell\in\{4,\dots,m-1\}$. Our proof uses a sequence of canonical orderings $\sigma_{\ell+1}$, $\sigma_{\ell+2}$, $\dots$, $\sigma_{m}$ of $H$, where for $j=\ell+1,\dots,m$, the ordering $\sigma_{j}$ is defined as $(v_1$, $\dots$, $v_{\ell-1}$, $v_{\ell+1}$, $\dots$,$v_{j}$,$v_{\ell}$,$v_{j+1},\dots,v_m)$. That is, the order of the vertices of $H$ in $\sigma_j$ coincides with the one in $\tau$, except that $v_{\ell}$ is shifted to the $j$-th position. See \cref{fig:canonical-ordering} for an example. We first prove that $\sigma_{\ell+1},\sigma_{\ell+2},\dots,\sigma_{m}$ are indeed canonical orderings of $H$.

\begin{figure}[t]
    \centering
    \begin{subfigure}{.49\textwidth}
    \centering
    	\includegraphics[page=3, scale =.75]{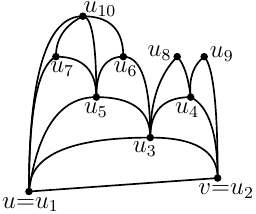}
        \caption{The canonical ordering $\pi$ of $H$.}
    	\label{fig:canonical-ordering-a}
    \end{subfigure}
    \hfill
    \begin{subfigure}{.49\textwidth}
    \centering
    	\includegraphics[page=4, scale =.75]{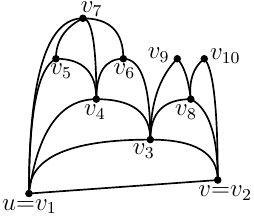}
    	\caption{The canonical ordering $\tau$ of $H$.}
    	\label{fig:canonical-ordering-b}
    \end{subfigure}
\caption{The case in which $u_m$ and $v_m$ are not the same vertex. In this example, we have $\ell=7$. Then $\sigma_7=(v_1,v_2,v_3,v_4,v_5,v_6,v_7,v_8,v_9,v_{10})$, $\sigma_8=(v_1,v_2,v_3,v_4,v_5,v_6,v_8,v_7,v_9,v_{10})$, $\sigma_9=(v_1,v_2,v_3,v_4,v_5,v_6,v_8,v_9,v_7,v_{10})$, and $\sigma_{10}=(v_1,v_2,v_3,v_4,v_5,v_6,v_8,v_9,v_{10},v_7)$.}
\label{fig:canonical-ordering}
\end{figure}

\begin{lemma} \label{le:third-canonical}
For $j=\ell+1,\dots,m$, we have that $\sigma_j$ is a canonical ordering of $H$; furthermore, the canonical orientation of $H$ defined by $\sigma_j$ is $\mathcal D_H$.
\end{lemma}

\begin{proof}
First, since $u_m$ coincides with $v_{\ell}$, for some $\ell\in \{4,\dots,m-1\}$, the first and second vertex of $\sigma_j$ are $u$ and $v$, since the same is true for $\tau$. For $k=3,\dots,m$, let $H^{\tau}_k$ and $H^j_k$ be the plane subgraphs induced by the first $k$ vertices of $H$ in $\tau$ and in $\sigma_j$, respectively. We now prove that $H^j_k$ satisfies conditions (CO-1) and (CO-2), for $k=3,\dots,m-1$.

\begin{itemize}
\item {\em Condition (CO-1).} For $k=3,\dots,\ell-1,j,j+1,\dots,m-1$, we have that $H^j_k$ is biconnected, because it coincides with $H^{\tau}_k$, and because $H^{\tau}_k$ satisfies condition (CO-1). We next prove that the graphs $H^j_{\ell}$, $H^j_{\ell+1}$, $\dots$, $H^j_{j-1}$ are biconnected. The vertices $v_{\ell+1},v_{\ell+2},\dots,v_{j}$ are not neighbors of $v_{\ell}$ in $H$, since by \cref{pr:neighbors-before} all the neighbors of $v_{\ell}$ in $H$ precede $v_{\ell}$ in $\tau$ and hence are among $v_1,v_2,\dots,v_{\ell-1}$. Also, for $k=\ell,\ell+1,\dots,j-1$, we have that $H^j_k$ is the graph obtained from $H^{\tau}_{k+1}$ by removing $v_{\ell}$ and its incident edges. Hence, for $k=\ell-1,\ell,\dots,j-2$, the $(k+1)$-th vertex of $\sigma_j$ has at least two neighbors in $H^j_k$, given that it coincides with the $(k+2)$-th vertex of $\tau$, which has at least two neighbors in $H^{\tau}_{k+1}$ since $H^{\tau}_{k+1}$ satisfies condition (CO-2), and given that it is not a neighbor of $v_{\ell}$. This, together with the fact that $H^j_{\ell-1}$ is biconnected, implies the biconnectivity of $H^j_{\ell}$, $H^j_{\ell+1}$, $\dots$, $H^j_{j-1}$. 


\item {\em Condition (CO-2).} For $k=3,\dots,\ell-2,j,j+1,\dots,m-1$, the $(k+1)$-th vertex of $\sigma_j$ is in the outer face of $H^j_k$ because it is the same vertex as the $(k+1)$-th vertex of $\tau$, because $H^j_k$ coincides with  $H^{\tau}_k$, and because the $(k+1)$-th vertex of $\tau$ is in the outer face of $H^{\tau}_k$, given that $H^{\tau}_k$ satisfies condition (CO-2). For $k=\ell-1,\ell,\dots,j-2$, the $(k+1)$-th vertex of $\sigma_j$ is in the outer face of $H^j_k$ because it is the same vertex as the $(k+2)$-th vertex of $\tau$, because $H^j_k$ is a subgraph of  $H^{\tau}_{k+1}$, and because the $(k+2)$-th vertex of $\tau$ is in the outer face of $H^{\tau}_{k+1}$, given that $H^{\tau}_{k+1}$ satisfies condition (CO-2). Finally, the $j$-th vertex of $\sigma_j$, that is $v_{\ell}=u_m$, is in the outer face of $H^j_{j-1}$ because it is in the outer face of the plane subgraph $H_{m-1}$ of $H$ induced by the first $m-1$ vertices of $\pi$, given that $H_{m-1}$ satisfies condition (CO-2). 
\end{itemize}

This proves that $\sigma_j$ is a canonical ordering. The fact that the canonical orientation $\mathcal D_j$ defined by $\sigma_j$ is $\mathcal D_H$ follows from the following facts: (i) the canonical orientation defined by $\tau$ is $\mathcal D_H$; (ii) the orderings $\tau$ and $\sigma_j$ coincide on the vertices of $H$ different from $v_{\ell}$ -- this implies that the orientation of every edge that is not incident to $v_{\ell}$ is the same in $\mathcal D_j$ and in $\mathcal D_H$; and (iii) the vertex $v_{\ell}$ follows its neighbors in $H$ both in $\tau$ and in $\sigma_j$ -- this implies that the orientation of every edge incident to $v_{\ell}$ is the same in $\mathcal D_j$ and in $\mathcal D_H$.
\end{proof}

In order to simplify the description, let $\sigma_{\ell}$ coincide with $\tau$. For $j=\ell,\dots,m$, let $\Sigma^j$ be the canonical drawing of $H$ that corresponds to $\sigma_j$. Also, let $M^j_m(z_1)$, $M^j_m(z_2)$, $\dots$, $M^j_m(z_{r})$ be the sets that are associated with the vertices $z_1,z_2,\dots,z_{r}$, respectively, by the FPP algorithm, when applied with canonical ordering $\sigma_j$. In order to prove that $\Gamma$ and $\Phi=\Sigma^{\ell}$ are the same drawing and that $M^{\pi}_m(z_i)=M^{\tau}_m(z_i)$, for $i=1,\dots,r$, it suffices to prove that: 
\begin{enumerate}[\bf (i)]
    \item for $j=\ell,\dots,m-1$, it holds that $\Sigma^j$ and $\Sigma^{j+1}$ are the same drawing and, for $j=\ell,\dots,m-1$ and for $i=1,\dots,r$, it holds that $M^j_m(z_i)=M^{j+1}_m(z_i)$.
    \item $\Sigma^m$ and $\Gamma$ are the same drawing and, for $i=1,\dots,r$, it holds that $M^m_m(z_i)=M^{\pi}_m(z_i)$.
\end{enumerate}
We have that {\bf (ii)} follows by Case~1. Indeed, $\Sigma^m$ and $\Gamma$ define the same canonical orientation of $H$, by \cref{le:third-canonical}, and the last vertex of both $\sigma_m$ and $\pi$ is $v_\ell=u_m$. Hence, it only remains to prove {\bf (i)}. Thus, consider any $j\in \{\ell,\dots,m-1\}$. We need to prove that $\Sigma^j$ and $\Sigma^{j+1}$ are the same drawing and that, for $i=1,\dots,r$, it holds that $M^j_m(z_i)=M^{j+1}_m(z_i)$. 

We introduce some notation. For $k=3,\dots,m$, let $H^j_k$ (let $H^{j+1}_k$) be the plane subgraph of $H$ induced by the first $k$ vertices of $\sigma_j$ (resp.\ of $\sigma_{j+1}$). Also, let $\Sigma^j_k$  (let $\Sigma^{j+1}_k$) be the drawing of $H^j_k$ (resp.\ of $H^{j+1}_k$) that is constructed by the FPP algorithm when applied with canonical ordering $\sigma_j$ (resp.\ $\sigma_{j+1}$), on the way of constructing $\Sigma^j$ (resp.\ $\Sigma^{j+1}$). Furthermore, let $w^k_1=u,w^k_2,\dots,w^k_{r_k}=v$ (let $z^k_1=u,z^k_2,\dots,z^k_{s_k}=v$) be the clockwise order of the vertices along the outer face of $H^j_k$ (resp.\ of $H^{j+1}_k$). Finally, for $i=1,\dots,r_k$ (for $i=1,\dots,s_k$), let $M^j_k(w^k_i)$ (resp.\ let $M^{j+1}_k(z^k_i)$) be the set that is associated to the vertex $w^k_i$ (resp.\ to the vertex $z^k_i$) by the FPP algorithm, when applied with canonical ordering $\sigma_j$ (resp.\ $\sigma_{j+1}$).

The whole reason for introducing the sequence of canonical orderings $\sigma_{\ell}$, $\sigma_{\ell+1}$, $\dots$, $\sigma_{m}$ of $H$ is that canonical orderings that are consecutive in the sequence are very similar to one another. Indeed, $\sigma_j$ and $\sigma_{j+1}$ coincide, except for the $j$-th and $(j+1)$-th vertex, which are swapped. Specifically, the $j$-th and $(j+1)$-th vertex of $\sigma_j$ are $v_{\ell}$ and $v_{j+1}$, respectively, while the $j$-th and $(j+1)$-th vertex of $\sigma_{j+1}$ are $v_{j+1}$ and $v_{\ell}$, respectively. This implies that, for $k=3,\dots,j-1,j+1,\dots,m$, the graphs $H^j_k$ and $H^{j+1}_k$ coincide. Further, since $\sigma_j$ and $\sigma_{j+1}$ coincide on the first $j-1$ vertices, the drawings $\Sigma^j_{j-1}$ and $\Sigma^{j+1}_{j-1}$ of $H^j_{j-1}=H^{j+1}_{j-1}$ coincide and, for $i=1,\dots,r_{j-1}=s_{j-1}$, the set $M^j_{j-1}(w^{j-1}_i)$ coincides with $M^{j+1}_{j-1}(w^{j-1}_i)$. 

Note that $H^j_j$ and $H^{j+1}_j$ are not the same graph, as the vertex set of the former is $V(H^j_{j-1})\cup \{v_{\ell}\}$, while the one of the latter is $V(H^j_{j-1})\cup \{v_{j+1}\}$. Thus, obviously, $\Sigma^j_j$ and $\Sigma^{j+1}_j$ are not the same drawing, and the sets $M^j_j(w^j_i)$ and $M^{j+1}_j(z^j_i)$ associated to the vertices on the boundary of $H^j_j$ and $H^{j+1}_j$, respectively, do not coincide. However, we are going to prove that all the required equalities are recovered at the next step of the FPP algorithm, when $v_{j+1}$ and $v_{\ell}$ are inserted into $H^j_j$ and $H^{j+1}_j$ to form $H^j_{j+1}=H^{j+1}_{j+1}$. 

Recall that the boundary of $H^j_{j-1}=H^{j+1}_{j-1}$ is $w^{j-1}_1,w^{j-1}_2,\dots,w^{j-1}_{r_{j-1}}$. In the discussion that follows, for ease of notation, we drop the apex $^{j-1}$ from these vertices, which are then denoted by $w_1,w_2,\dots,w_{r_{j-1}}$. Since all the neighbors of $v_{\ell}$ precede $v_{\ell}$ in $\tau$, by \cref{pr:neighbors-before}, and since $v_{j+1}$ follows $v_{\ell}$ in $\tau$, it follows that $v_{\ell}$ and $v_{j+1}$ are not neighbors. This, together with condition (CO-2) for $H^j_{j}$ and $H^j_{j+1}$, implies that the neighbors of $v_{\ell}$ in $H^j_{j+1}=H^{j+1}_{j+1}$ are $w_{p},w_{p+1},\dots,w_{q}$, for some $1\leq p<q\leq r_{j-1}$, that the neighbors of $v_{j+1}$ in $H^j_{j+1}=H^{j+1}_{j+1}$ are $w_{p'},w_{p'+1},\dots,w_{q'}$, for some $1\leq p'<q'\leq r_{j-1}$, and that either $q\leq p'$ or $q'\leq p$ (that is, the sequences of neighbors of $v_{\ell}$ and  $v_{j+1}$ along the boundary of $H^j_{j-1}=H^{j+1}_{j-1}$ are disjoint, except for the last vertex of one of them, which might coincide with the first vertex of the other one). Assume that $q\leq p'$, as the case $q'\leq p$ is symmetric. Then the graphs $H^j_j$, $H^{j+1}_j$, and $H^j_{j+1}=H^{j+1}_{j+1}$ have the following boundaries.
\begin{itemize}
\item $H^j_j$ has boundary $w_1,\dots,w_p,v_{\ell},w_q,\dots,w_{r_{j-1}}$;
\item $H^{j+1}_j$ has boundary $w_1,\dots,w_{p'},v_{j+1},w_{q'},\dots,w_{r_{j-1}}$; and
\item $H^j_{j+1}=H^{j+1}_{j+1}$ has boundary $w_1$, $\dots$, $w_p$, $v_{\ell}$, $w_q$, $\dots$, $w_{p'}$, $v_{j+1}$, $w_{q'}$, $\dots$, $w_{r_{j-1}}$.
\end{itemize}
Since $M^j_{j-1}(w_i)=M^{j+1}_{j-1}(w_i)$, for $i=1,\dots,r_{j-1}$, the FPP algorithm defines: 
\begin{enumerate}[(i)]
    \item $M^j_j(w_i)=M^j_{j-1}(w_i)\cup \{v_{\ell}\}$, for $i=1,\dots,p$;
    \item $M^j_{j}(v_{\ell})=M^j_{j-1}(w_{p+1})\cup \{v_{\ell}\}$; and
    \item $M^j_{j}(w_i)=M^j_{j-1}(w_i)$, for $i=q,\dots,r_{j-1}$.
    \item $M^{j+1}_{j}(w_i)=M^j_{j-1}(w_i)\cup \{v_{j+1}\}$, for $i=1,\dots,p'$;
    \item $M^{j+1}_{j}(v_{j+1})=M^j_{j-1}(w_{p'+1})\cup \{v_{j+1}\}$; and
    \item $M^{j+1}_{j}(w_i)=M^j_{j-1}(w_i)$, for $i=q',\dots,r_{j-1}$.
    \item $M^j_{j+1}(w_i)=M^j_{j}(w_i)\cup \{v_{j+1}\}$, for $i=1,\dots,p$;
    \item $M^j_{j+1}(v_{\ell})=M^j_{j}(v_{\ell})\cup \{v_{j+1}\}$;
    \item $M^j_{j+1}(w_i)=M^j_{j}(w_i)\cup \{v_{j+1}\}$, for $i=q,\dots,p'$;
    \item $M^j_{j+1}(v_{j+1})=M^j_{j}(w_{p'+1})\cup \{v_{j+1}\}$; and
    \item $M^j_{j+1}(w_i)=M^j_{j}(w_i)$, for $i=q',\dots,r_{j-1}$.
    \item $M^{j+1}_{j+1}(w_i)=M^{j+1}_{j}(w_i)\cup \{v_{\ell}\}$, for $i=1,\dots,p$;
    \item $M^{j+1}_{j+1}(v_{\ell})=M^{j+1}_{j}(w_{p+1})\cup \{v_{\ell}\}$;
    \item $M^{j+1}_{j+1}(w_i)=M^{j+1}_{j}(w_i)$, for $i=q,\dots,p'$;
    \item $M^{j+1}_{j+1}(v_{j+1})=M^{j+1}_{j}(v_{j+1})$; and
    \item $M^{j+1}_{j+1}(w_i)=M^{j+1}_{j}(w_i)$, for $i=q',\dots,r_{j-1}$.
\end{enumerate}
We now prove that the required equalities for the sets associated by the FPP algorithm to the vertices along the boundary of $H^j_{j+1}=H^{j+1}_{j+1}$ are indeed satisfied. 

\begin{itemize}
\item For $i=1,\dots,p$, by equalities (vii) and (i), we have $M^j_{j+1}(w_i)=M^j_{j-1}(w_i)\cup \{v_{\ell},v_{j+1}\}$. By equalities (xii) and (iv), we have $M^{j+1}_{j+1}(w_i)=M^j_{j-1}(w_i)\cup \{v_{\ell},v_{j+1}\}$. Hence, $M^j_{j+1}(w_i)=M^{j+1}_{j+1}(w_i)$, as required.
\item By equalities (viii) and (ii), we have $M^j_{j+1}(v_{\ell})=M^j_{j-1}(w_{p+1})\cup \{v_{\ell},v_{j+1}\}$. By equalities (xiii) and (iv), we have $M^{j+1}_{j+1}(v_{\ell})=M^j_{j-1}(w_{p+1})\cup \{v_{\ell},v_{j+1}\}$. Hence, $M^j_{j+1}(v_{\ell})=M^{j+1}_{j+1}(v_{\ell})$, as required.
\item For $i=q,\dots,p'$, by equalities (ix) and (iii), we have $M^j_{j+1}(w_i)=M^j_{j-1}(w_i)\cup \{v_{j+1}\}$. By equalities (xiv) and (iv), we have $M^{j+1}_{j+1}(w_i)=M^j_{j-1}(w_i)\cup \{v_{j+1}\}$. Hence, $M^j_{j+1}(w_i)=M^{j+1}_{j+1}(w_i)$, as required.
\item By equalities (x) and (iii), we have $M^j_{j+1}(v_{j+1})=M^j_{j-1}(w_{p'+1})\cup \{v_{j+1}\}$. By equalities (xv) and (v), we have $M^{j+1}_{j+1}(v_{j+1})=M^j_{j-1}(w_{p'+1})\cup \{v_{j+1}\}$. Hence, $M^j_{j+1}(v_{j+1})=M^{j+1}_{j+1}(v_{j+1})$, as required.
\item Finally, for $i=q'+1,\dots,r_{j-1}$,  by equalities (xi) and (iii), we have $M^j_{j+1}(w_i)=M^j_{j-1}(w_i)$. By equalities (xvi) and (vi), we have $M^{j+1}_{j+1}(w_i)=M^j_{j-1}(w_i)$. Hence, $M^j_{j+1}(w_i)=M^{j+1}_{j+1}(w_i)$, as required.
\end{itemize}
In order to prove that $\Sigma^j_{j+1}$  and $\Sigma^{j+1}_{j+1}$ are the same drawing, we first observe that all the vertices of $H^j_{j-1}$ have in $\Sigma^j_{j+1}$  and $\Sigma^{j+1}_{j+1}$ the same $y$-coordinate as in $\Sigma^j_{j-1}$ and in $\Sigma^{j+1}_{j-1}$, respectively. This is because the $j$-th and $(j+1)$-th step of the FPP algorithm shift the vertices in the drawing of $H^j_{j-1}$ only horizontally and because $\Sigma^j_{j-1}$ and $\Sigma^{j+1}_{j-1}$ are the same drawing.

In order to deal with the $x$-coordinates of the vertices of $H^j_{j-1}$ in $\Sigma^j_{j+1}$  and $\Sigma^{j+1}_{j+1}$, we partition the vertices of $H^j_{j-1}$ into five sets $V_0,V_1,V_2,V_3, V_4$ defined as follows: 
\begin{itemize}
\item $V_0:=M^j_{j-1}(w_1)\setminus M^j_{j-1}(w_{p+1})$,
\item $V_1:=M^j_{j-1}(w_{p+1})\setminus M^j_{j-1}(w_q)$,
\item $V_2:=M^j_{j-1}(w_q)\setminus M^j_{j-1}(w_{p'+1})$,
\item $V_3:=M^j_{j-1}(w_{p'+1})\setminus M^j_{j-1}(w_{q'})$, and
\item $V_4:=M^j_{j-1}(w_{q'})$.
\end{itemize}
We claim that, both in $\Sigma^j_{j+1}$ and in $\Sigma^{j+1}_{j+1}$, for $i=0,\dots,4$, the $x$-coordinate of a vertex in $V_i$ coincides with the one in $\Sigma^j_{j-1}=\Sigma^{j+1}_{j-1}$ plus $i$. This, together with the fact that $\Sigma^j_{j-1}$ and $\Sigma^{j+1}_{j-1}$ are the same drawing, implies that every vertex of $H^j_{j-1}$ has the same $x$-coordinate in $\Sigma^j_{j+1}$  and in $\Sigma^{j+1}_{j+1}$. 

By the FPP algorithm, the $x$-coordinate of every vertex $z$ of $H^j_{j-1}$ in $\Sigma^j_{j+1}$ is the same as in $\Sigma^j_{j-1}$, plus one unit for each set $z$ belongs to among $M^j_{j-1}(w_{p+1})$, $M^j_{j-1}(w_q)$, $M^j_{j}(w_{p'+1})$, and $M^j_{j}(w_{q'})$. By equality (iii), the last two sets coincide with $M^j_{j-1}(w_{p'+1})$ and $M^j_{j-1}(w_{q'})$, respectively. Analogously, the $x$-coordinate of every vertex $z$ in $\Sigma^{j+1}_{j+1}$ is the same as in $\Sigma^j_{j-1}$, plus one unit for each set $z$ belongs to among  $M^{j+1}_{j-1}(w_{p'+1})$, $M^{j+1}_{j-1}(w_{q'})$, $M^{j+1}_{j}(w_{p+1})$, and $M^{j+1}_{j}(w_{q})$. The first two sets coincide with $M^j_{j-1}(w_{p'+1})$ and $M^j_{j-1}(w_{q'})$, respectively. Further, by equality (iv), the last two sets coincide with $M^j_{j-1}(w_{p+1})\cup \{v_{j+1}\}$ and $M^j_{j-1}(w_{q})\cup \{v_{j+1}\}$, respectively.
\begin{itemize}
    \item First, any vertex $z\in V_0$ belongs to neither of the sets $M^j_{j-1}(w_{p+1})\cup \{v_{j+1}\}$, $M^j_{j-1}(w_{q})\cup \{v_{j+1}\}$, $M^j_{j-1}(w_{p'+1})$, $M^j_{j-1}(w_{q'})$, given that $v_{j+1}$ does not belong to $H^j_{j-1}$, given that $V_0\cap M^j_{j-1}(w_{p+1})=\emptyset$, by definition, and given that $M^j_{j-1}(w_{p+1})\supseteq M^j_{j-1}(w_{q})\supset M^j_{j-1}(w_{p'+1})\supseteq M^j_{j-1}(w_{q'})$. Hence, the $x$-coordinate of $z$ in both $\Sigma^j_{j+1}$ and $\Sigma^{j+1}_{j+1}$ coincides with the one in $\Sigma^j_{j-1}=\Sigma^{j+1}_{j-1}$.
    \item Second, any vertex $z\in V_1$ belongs to the set $M^j_{j-1}(w_{p+1})$, by definition, and to neither of the sets $M^j_{j-1}(w_{q})\cup \{v_{j+1}\}$, $M^j_{j-1}(w_{p'+1})$, $M^j_{j-1}(w_{q'})$, given that $v_{j+1}$ does not belong to $H^j_{j-1}$, given that $V_1\cap M^j_{j-1}(w_{q})=\emptyset$, by definition, and given that $M^j_{j-1}(w_{q})\supset M^j_{j-1}(w_{p'+1})\supseteq M^j_{j-1}(w_{q'})$. Hence, the $x$-coordinate of $z$ in both $\Sigma^j_{j+1}$ and $\Sigma^{j+1}_{j+1}$ coincides with the one in $\Sigma^j_{j-1}=\Sigma^{j+1}_{j-1}$ plus one.   
    \item Third, any vertex $z\in V_2$ belongs to the sets $M^j_{j-1}(w_{p+1})$ and $M^j_{j-1}(w_{q})$, by definition and since $M^j_{j-1}(w_{p+1})\supseteq M^j_{j-1}(w_{q})$, and to neither of the sets $M^j_{j-1}(w_{p'+1})$ and $M^j_{j-1}(w_{q'})$, given that $V_2\cap M^j_{j-1}(w_{p'+1})=\emptyset$, by definition, and given that $M^j_{j-1}(w_{p'+1})\supseteq M^j_{j-1}(w_{q'})$. Hence, the $x$-coordinate of $z$ in both $\Sigma^j_{j+1}$ and $\Sigma^{j+1}_{j+1}$ coincides with the one in $\Sigma^j_{j-1}=\Sigma^{j+1}_{j-1}$ plus two. 
    \item Fourth, any vertex $z\in V_3$ belongs to the sets $M^j_{j-1}(w_{p+1})$, $M^j_{j-1}(w_{q})$, and $M^j_{j-1}(w_{p'+1})$, by definition and since $M^j_{j-1}(w_{p+1})\supseteq M^j_{j-1}(w_{q})\supset M^j_{j-1}(w_{p'+1})$, and does not belong to $M^j_{j-1}(w_{q'})$, given that $V_3\cap M^j_{j-1}(w_{q'})=\emptyset$, by definition. Hence, the $x$-coordinate of $z$ in both $\Sigma^j_{j+1}$ and $\Sigma^{j+1}_{j+1}$ coincides with the one in $\Sigma^j_{j-1}=\Sigma^{j+1}_{j-1}$ plus three. 
    \item Finally, any vertex $z\in V_4$ belongs to all of the sets $M^j_{j-1}(w_{p+1})$, $M^j_{j-1}(w_{q})$, $M^j_{j-1}(w_{p'+1})$, and $M^j_{j-1}(w_{q'})$, by definition and since $M^j_{j-1}(w_{p+1})\supseteq M^j_{j-1}(w_{q})\supset M^j_{j-1}(w_{p'+1})\supseteq M^j_{j-1}(w_{q'})$. Hence, the $x$-coordinate of $z$ in both $\Sigma^j_{j+1}$ and $\Sigma^{j+1}_{j+1}$ coincides with the one in $\Sigma^j_{j-1}=\Sigma^{j+1}_{j-1}$ plus four.  
\end{itemize}

It remains to prove that $v_{\ell}$ and $v_{j+1}$ have the same coordinates in $\Sigma^j_{j+1}$  and $\Sigma^{j+1}_{j+1}$. Denote by $\Omega(z)$ the position of a vertex $z$ in a drawing $\Omega$.

\begin{itemize}
    \item By construction, $\Sigma^j_j(v_{\ell})$ coincides with the intersection point of the line with slope $+1$ through $\Sigma^j_{j-1}(w_{p})$ and the line with slope $-1$ through the point two units to the right of $\Sigma^j_{j-1}(w_{q})$. Since $v_{\ell}$ belongs neither to $M^j_{j}(w_{p'+1})$ nor to $M^j_{j}(w_{q'})$, as by equality (iii) such sets coincide with $M^j_{j-1}(w_{p'+1})$ and $M^j_{j-1}(w_{q'})$, respectively, it follows that $\Sigma^j_{j+1}(v_{\ell})=\Sigma^j_j(v_{\ell})$.
    
    Note that $\Sigma^{j+1}_j(w_{p})=\Sigma^{j+1}_{j-1}(w_{p})=\Sigma^j_{j-1}(w_{p})$; indeed, the first equality holds true because $w_{p}$ belongs neither to $M^{j+1}_{j-1}(w_{p'+1})$ nor to $M^{j+1}_{j-1}(w_{q'})$, and the second equality holds true since $\Sigma^j_{j-1}=\Sigma^{j+1}_{j-1}$. Analogously, $\Sigma^{j+1}_j(w_{q})=\Sigma^{j+1}_{j-1}(w_{q})=\Sigma^j_{j-1}(w_{q})$. Hence, $\Sigma^{j+1}_{j+1}(v_{\ell})$ coincides with the intersection point of the line with slope $+1$ through $\Sigma^j_{j-1}(w_{p})$ and the line with slope $-1$ through the point two units to the right of $\Sigma^j_{j-1}(w_{q})$, thus $v_{\ell}$ has the same coordinates in $\Sigma^j_{j+1}$  and $\Sigma^{j+1}_{j+1}$. 

    \item By construction and since $\Sigma^j_{j-1}=\Sigma^{j+1}_{j-1}$, we have that $\Sigma^{j+1}_j(v_{j+1})$ coincides with the intersection point $p$ of the line with slope $+1$ through $\Sigma^j_{j-1}(w_{p'})$ and the line with slope $-1$ through the point two units to the right of $\Sigma^j_{j-1}(w_{q'})$. Since $v_{j+1}$ belongs both to $M^{j+1}_{j}(w_{p+1})$ and to $M^{j+1}_{j}(w_{q})$, as by equality (iv) such sets coincide with $M^j_{j-1}(w_{p+1})\cup \{v_{j+1}\}$ and with $M^j_{j-1}(w_{q})\cup \{v_{j+1}\}$, respectively, it follows that $\Sigma^{j+1}_{j+1}(v_{j+1})$ coincides with the point two units to the right of $p$.

    Note that $\Sigma^j_j(w_{p'})$ coincides with the point two units to the right of $\Sigma^j_{j-1}(w_{p'})$. Indeed, $w_{p'}$ belongs both to $M^j_{j-1}(w_{p+1})$ and to $M^j_{j-1}(w_{q})$. Analogously, $\Sigma^j_j(w_{q'})$ coincides with the point two units to the right of $\Sigma^j_{j-1}(w_{q'})$. Hence, $\Sigma^j_{j+1}(v_{j+1})$ coincides with the intersection point of the line with slope $+1$ through the point two units to the right of $\Sigma^j_{j-1}(w_{p'})$ and the line with slope $-1$ through the point two units to the right of $\Sigma^j_{j-1}(w_{q'})$, thus coincides with the point two units to the right of $p$. Hence, $v_{j+1}$ has the same coordinates in $\Sigma^j_{j+1}$  and $\Sigma^{j+1}_{j+1}$.     
\end{itemize} 

This concludes the proof that $\Sigma^j_{j+1}$  and $\Sigma^{j+1}_{j+1}$ are the same drawing.

Finally, since $\Sigma^j_{j+1}$  and $\Sigma^{j+1}_{j+1}$ are the same drawing, since the sets $M^j_{j+1}(w^{j+1}_i)$ and $M^{j+1}_{j+1}(z^{j+1}_i=w^{j+1}_i)$ coincide, for $i=1,\dots,r_{j+1}=s_{j+1}$, and since $\sigma_j$ and $\sigma_{j+1}$ coincide on the last $m-(j+1)$ vertices, it follows that $\Sigma^j_k$ and $\Sigma^{j+1}_k$ are the same drawing, for $k=j+2,\dots,m$, and that the sets $M^j_{k}(w^k_i)$ and $M^{j+1}_{k}(z^k_i=w^k_i)$ coincide, for $k=j+2,\dots,m$ and for $i=1,\dots,r_k=s_k$. Since the drawings $\Sigma^j_m$ and $\Sigma^{j+1}_m$ coincide with $\Sigma^j$ and $\Sigma^{j+1}$, respectively, it follows that $\Sigma^j$ and $\Sigma^{j+1}$ are the same drawing, as required. Also, since for $i=1,\dots,r=r_m=s_m$, the set $M^j_m(z_i)$ coincides with $M^j_m(w^m_i)$ and the set $M^{j+1}_m(z_i)$ coincides with $M^{j+1}_m(z^m_i)$, it follows that $M^j_m(z_i)=M^{j+1}_m(z_i)$, as required. 
This completes the induction and hence the proof of \cref{cl:biconnected-orientation-determines-drawing} and \cref{pr:inductive-sets}. It follows that the function $f$ is surjective, which concludes the proof of \cref{th:bijection-FPP}. 

\cref{th:canonical-ordering-plane-uv} and \cref{th:bijection-FPP} imply the following. 

\begin{lemma} \label{th:FPP-enumeration-plane-uv}
Let $G$ be an $n$-vertex maximal plane graph and let $(u,v,z)$ be the cycle delimiting its outer face, where $u$, $v$, and $z$ appear in this counter-clockwise order along the cycle. There exists an algorithm with $\bigoh(n)$ setup time and $\bigoh(n)$ space usage that lists all canonical drawings of $G$ with base edge $(u,v)$ with $\bigoh(n)$ delay.
\end{lemma}

\begin{theorem} \label{th:FPP-enumeration-plane}
Let $G$ be an $n$-vertex maximal plane (planar) graph. There exists an algorithm $\mathcal C_1$ (resp.\ $\mathcal C_2$) with $\bigoh(n)$ setup time and $\bigoh(n)$ space usage that lists all canonical drawings of $G$ with $\bigoh(n)$ delay.
\end{theorem}

\begin{proof}
Algorithm $\mathcal C_1$ uses the algorithm for the proof of \cref{th:FPP-enumeration-plane-uv} applied three times, namely once for each choice of the base edge among the three edges incident to the outer face of the given maximal plane graph. Algorithm $\mathcal C_2$ uses algorithm $\mathcal C_1$ applied $4n-8$ times, since there are $4n-8$ maximal plane graphs which are isomorphic to a given maximal plane graph (see the proof of \cref{th:canonical-orientation-plane-and-planar}). Note that any two canonical drawings produced by different applications of algorithm $\mathcal C_1$ differ on the three vertices incident to the outer face, or on their coordinates in the drawing. 
\end{proof}


\section{Enumeration of Schnyder Woods and Schnyder Drawings}\label{se:schnyder}

In this section, we show how the enumeration algorithm for canonical orientations from \cref{se:canonical-orientation} can be used in order to provide efficient algorithms for the enumeration of Schnyder woods and Schnyder drawings. We start with the following theorems.

\begin{theorem} \label{th:schnyder-plane}
	Let $G$ be an $n$-vertex maximal plane (planar) graph. There exists an algorithm $\mathcal M_1$ (resp.\ $\mathcal M_2$) with $\bigoh(n)$ setup time and $\bigoh(n)$ space usage that lists all Schnyder woods of $G$ with $\bigoh(n)$ delay.
\end{theorem}

\begin{proof}
We first discuss algorithm $\mathcal M_1$, hence let $G$ be an $n$-vertex maximal plane graph. As proved by de Fraysseix and Ossona De Mendez~\cite[Theorem 3.3]{DBLP:journals/dm/FraysseixM01}, there is a  bijection between the Schnyder woods of $G$ and the canonical orientations of $G$. Given a canonical orientation $\mathcal D$ of $G$, the corresponding Schnyder wood $\mathcal W=(\mathcal T_1,\mathcal T_2,\mathcal T_3)$ can be obtained as follows (see also~\cite{b-3os3d-00,c-h-21,dpp-hdpgg-90,d-gdt-10,dtv-osr-99,kobourov2016canonical,DBLP:conf/soda/Schnyder90}). For every internal vertex $w$ of $G$, let $e_1,\dots,e_k$ be the counter-clockwise order of the incoming edges at $w$ in $\mathcal D$, where $k\geq 2$. Assign $e_1$ with color $1$ and orient it (in $\mathcal W$) so that it is outgoing at $w$; also, assign $e_k$ with color $2$ and orient it (in $\mathcal W$) so that it is outgoing at $w$; finally, assign $e_2,\dots,e_{k-1}$ with color $3$ and orient them (in $\mathcal W$) so that they are incoming at $w$. The construction of $\mathcal W$ is completed by assigning all the edges that are incident to the sink of $\mathcal D$ and that do not belong to the boundary of $G$ with color $3$ and orienting them (in $\mathcal W$) so that they are incoming at the sink of $\mathcal D$. Since the construction of $\mathcal W$ from $\mathcal D$ can be easily implemented in $\bigoh(n)$ time and space, it descends from \cref{th:canonical-orientation-plane-and-planar} that algorithm $\mathcal M_1$ satisfies the required properties.

Algorithm $\mathcal M_2$ uses algorithm $\mathcal M_1$ applied $4n-8$ times, since there are $4n-8$ maximal plane graphs isomorphic to a given maximal planar graph (see the proof of \cref{th:canonical-orientation-plane-and-planar}). Note that any two Schnyder woods produced by different applications of algorithm $\mathcal M_1$ differ on the triple of vertices that have no outgoing edge.
\end{proof}

We now turn our attention to the enumeration of the planar straight-line drawings produced by the algorithm by Schnyder~\cite{DBLP:conf/soda/Schnyder90}, known as \emph{Schnyder drawings}. We start by describing such an algorithm, which takes as input (see  \cref{fig:schnyder-drawings-a}):
\begin{itemize}
\item an $n$-vertex maximal plane graph $G$, whose outer face is delimited by a cycle $(u,v,z)$, where $u$, $v$, and $z$ appear in this counter-clockwise order along the outer face of $G$; and
\item a Schnyder wood $\mathcal W = (\mathcal T_1,\mathcal T_2, \mathcal T_3)$ of $G$.
\end{itemize}
For ease of notation, we let $u_1$, $u_2$, and $u_3$ be alternative labels for $u$, $v$, and $z$, respectively, so that $\mathcal T_i$ contains $u_i$, for $i=1,2,3$. For a cycle $\mathcal C$ in $G$, let $\#_f(\mathcal C)$ denote the number of internal faces of $G$ in the interior of $\mathcal C$.

\begin{figure}[htb]
    \centering
    \begin{subfigure}{.31\textwidth}
    \centering
    	\includegraphics[page=1, scale =.65]{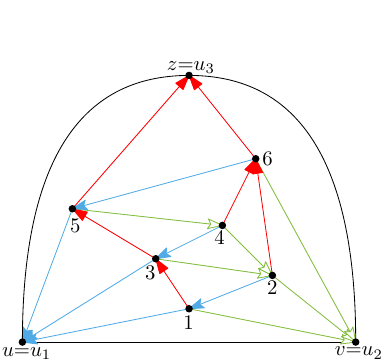}
    \caption{\label{fig:schnyder-drawings-a}}
    \end{subfigure}
    \hfill
    \begin{subfigure}{.31\textwidth}
    \centering
    	\includegraphics[page=2, scale =.65]{img/schnyder-drawings.pdf}
    \caption{\label{fig:schnyder-drawings-b}}
    \end{subfigure}
    \hfill
    \begin{subfigure}{.31\textwidth}
    \centering
    	\includegraphics[page=3, scale =.65]{img/schnyder-drawings.pdf}
     \caption{\label{fig:schnyder-drawings-c}}
    \end{subfigure}
\caption{(a) A maximal plane graph $G$ and a Schnyder wood $\mathcal W$ of $G$. (b) Paths $\mathcal P_1(4)$, $\mathcal P_2(4)$, and $\mathcal P_3(4)$, and cycles $\mathcal C_x(4)$ and $\mathcal C_y(4)$, where $4$ is  an internal vertex of $G$. Note that $\#_f(\mathcal C_x(4))=5$ and $\#_f(\mathcal C_y(4))=5$. (c) The Schnyder drawing $s(\mathcal W)$ of $G$.}
\label{fig:schnyder-drawings}
\end{figure}

Schnyder's algorithm assigns coordinates $(0,0)$, $(2n-5,0)$, and $(0,2n-5)$ to the vertices $u_1$, $u_2$, and $u_3$, respectively. Consider any internal vertex $w$. For $i=1,2,3$, properties (S-1) and (S-2) of $\mathcal W$ imply that $\mathcal T_i$ contains a directed path $\mathcal P_i(w)$ from $w$ to $u_i$; see  \cref{fig:schnyder-drawings-b}. Moreover, $\mathcal P_1(w)$, $\mathcal P_2(w)$, and $\mathcal P_3(w)$ have $w$ as the only common vertex~\cite{DBLP:conf/soda/Schnyder90}. Let $\mathcal C_x(w)$ be the cycle composed of the paths $\mathcal P_1(w)$ and $\mathcal P_3(w)$, and of the edge $(u_1,u_3)$. Also, let $\mathcal C_y(w)$ be the cycle composed of the paths $\mathcal P_1(w)$ and $\mathcal P_2(w)$, and of the edge $(u_1,u_2)$. Then Schnyder's algorithm assigns coordinates $(\#_f(\mathcal C_x(w)),\#_f(\mathcal C_y(w)))$ to $w$; see  \cref{fig:schnyder-drawings-c}.

We now prove the main ingredient of our enumeration algorithm for Schnyder drawings.

\begin{theorem} \label{th:bijection-sch}
	Let $G$ be an $n$-vertex maximal plane graph. There exists a bijective function from the Schnyder woods of $G$ to the Schnyder drawings of $G$. Also, given a Schnyder wood of $G$, the corresponding Schnyder drawing of $G$ can be constructed in $\bigoh(n)$ time.
\end{theorem}
\begin{proof}
	The bijective function $s$ that proves the statement is simply Schnyder's algorithm. The second part of the statement then follows from the fact that this algorithm can be implemented in $\bigoh(n)$ time~\cite{DBLP:conf/soda/Schnyder90}. 
	
	In order to prove that $s$ is bijective, we prove that it is injective (that is, for any two distinct Schnyder woods $\mathcal W_1$ and $\mathcal W_2$ of $G$, we have that  $s(\mathcal W_1)$ and $s(\mathcal W_2)$ are not the same drawing) and that it is surjective (that is, for any Schnyder drawing $\Gamma$, there exists a Schnyder wood $\mathcal W$ such that $s(\mathcal W)$ is $\Gamma$). That $s$ is surjective is actually obvious, as a Schnyder drawing $\Gamma$ is generated by applying Schnyder's algorithm to some Schnyder's wood $\mathcal W$. Then $s(\mathcal W)=\Gamma$. In the following, we prove that $s$ is injective.
	
	Consider any Schnyder drawing $\Gamma$. By definition of Schnyder drawing, there is {\em at least} one Schnyder wood $\mathcal W$ such that $s(\mathcal W)=\Gamma$. We prove that there is {\em at most} one such Schnyder wood, that is, $\Gamma$ uniquely determines $\mathcal W$. 
	
	\begin{figure}
		\centering
		\begin{subfigure}{0.47\textwidth}
			\centering
			\includegraphics[page=1, scale=.7]{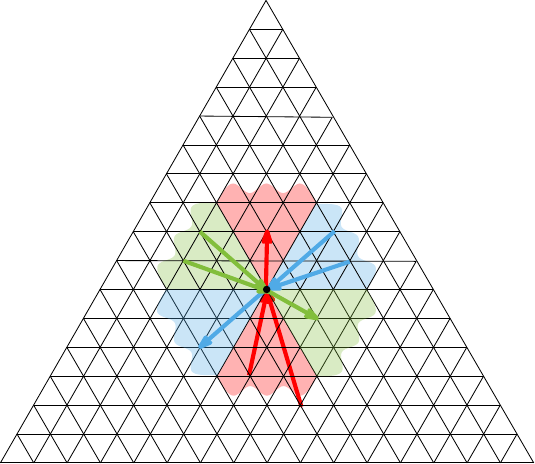}
			\caption{\label{fig:triangular-grid-angles}}
		\end{subfigure}
		\hfil
		\begin{subfigure}{0.47\textwidth}
			\centering
			\includegraphics[page=2, scale=.7]{img/schnyder-grid.pdf}
			\caption{\label{fig:square-grid-angles}}
		\end{subfigure}
		\caption{Slopes of the edges in a Schnyder drawing on (a) a triangular grid and (b) a square grid.}
		\label{fig:grid-angles}
	\end{figure}
	
	First, we recall that Schnyder drawings are often constructed on a triangular grid, rather than on the square grid. On the triangular grid, a Schnyder drawing constructed from a Schnyder wood $\mathcal W=(\mathcal T_1,\mathcal T_2,\mathcal T_3)$ has the property that, for each vertex $v$, the edges of $\mathcal T_1$, $\mathcal T_2$, and $\mathcal T_3$ incoming into $v$ have slopes in the intervals $(0^\circ,60^\circ)$, $(120^\circ,180^\circ)$, and $(240^\circ,300^\circ)$, respectively, while the edges of $\mathcal T_1$, $\mathcal T_2$, and $\mathcal T_3$ outgoing from $v$ have slopes in the intervals $(180^\circ,240^\circ)$, $(300^\circ,360^\circ)$, and $(60^\circ,120^\circ)$, respectively; see, e.g., \cite{d-gdt-10} and \cref{fig:triangular-grid-angles}. Hence, back on the square grid, a Schnyder drawing constructed from a Schnyder wood $\mathcal W=(\mathcal T_1,\mathcal T_2,\mathcal T_3)$ has the property that, for each vertex $v$, the edges of $\mathcal T_1$, $\mathcal T_2$, and $\mathcal T_3$ incoming into $v$ have slopes in the intervals $(0^\circ,90^\circ)$, $(135^\circ,180^\circ)$, and $(270^\circ,315^\circ)$, respectively, while the edges of $\mathcal T_1$, $\mathcal T_2$, and $\mathcal T_3$ outgoing from $v$ have slopes in the intervals $(180^\circ,270^\circ)$, $(315^\circ,360^\circ)$, and $(90^\circ,135^\circ)$, respectively; see \cref{fig:square-grid-angles}. Thus, for each edge $(u,v)$ of $G$, whether $(u,v)$ belongs to $\mathcal T_1$, $\mathcal T_2$, or $\mathcal T_3$, and whether $(u,v)$ is directed from $u$ to $v$ or vice versa, can be uniquely determined by the slope of the edge $(u,v)$ in $\Gamma$. This concludes the proof that $s$ is an injective function and hence the proof of \cref{th:bijection-sch}.
\end{proof}

We get the following.
\begin{theorem} \label{th:schnyder-enumeration-plane}
Let $G$ be an $n$-vertex maximal plane (planar) graph. There exists an algorithm $\mathcal N_1$ (resp.\ $\mathcal N_2$) with $\bigoh(n)$ setup time and $\bigoh(n)$ space usage that lists all Schnyder drawings of $G$ with $\bigoh(n)$ delay.
\end{theorem}

\begin{proof}
Algorithm $\mathcal N_1$ directly descends from \cref{th:schnyder-plane} and \cref{th:bijection-sch}. Algorithm $\mathcal N_2$ uses algorithm $\mathcal N_2$ applied $4n-8$ times, since there are $4n-8$ maximal plane graphs which are isomorphic to a given maximal planar graph (see the proof of \cref{th:canonical-orientation-plane-and-planar}). Note that any two Schnyder drawings produced by different applications of algorithm $\mathcal N_1$ differ on the three vertices incident to the outer face, or on their coordinates in the drawing. 
\end{proof}


\section{Conclusions}\label{se:conclusions}

In this paper, we considered the problem of enumerating two fundamental combinatorial structures of maximal planar graphs, i.e., canonical orderings and Schnyder woods. By exploiting their connection with canonical orientations, we developed efficient enumeration algorithms for such structures. We also proved novel, and in our opinion interesting, bijections between canonical orientations and canonical drawings, and between Schnyder woods and Schnyder drawings. This allowed us to empower our enumeration algorithms so that they can enumerate drawings within this classical drawing styles. The worst-case delay between two consecutive outputs of all our algorithms is linear in the graph size.

Our research initiates the study of graph-drawing enumeration algorithms and sparkles several interesting questions in this domain. In general, given a graph $G$, for any given drawing style $\cal D$, we may ask for the existence of efficient algorithms to enumerate all drawings of $G$ that respect $\cal D$. Natural examples of problems of this type include: (i) efficiently enumerating all the planar straight-line drawings of a given planar graph within a grid of prescribed size; (ii) efficiently enumerating all the orthogonal representations of a given plane graph with at most $b$ bends in total; and (iii) efficiently enumerating all the upward planar embeddings of a single-source digraph or of a triconnected digraph.


\bibliography{bibliography}


\end{document}